\documentclass[a4paper,twocolumn,10pt,accepted=2023-06-20]{quantumarticle}
\pdfoutput=1
\usepackage[utf8]{inputenc}
\usepackage[english]{babel}
\usepackage[T1]{fontenc}
\usepackage{amsmath}
\usepackage{hyperref}
\usepackage[numbers,compress]{natbib}
\usepackage{tikz}
\usepackage{lipsum}

\usepackage{graphicx}
\usepackage{float}
\usepackage{dcolumn}
\usepackage{bm}
\usepackage{amsfonts,amssymb,amscd,amsmath,amsthm}
\usepackage{algorithm}
\usepackage{algpseudocode}
\usepackage{enumerate}
\usepackage{epsfig}
\usepackage{subfigure}
\usepackage{physics}
\usepackage{svg}
\usepackage{tikz}
\usepackage[colorlinks = true]{hyperref}
\hypersetup{colorlinks=true, linkcolor=blue, citecolor=blue}
\usepackage{amsmath, amsthm, amsfonts, amssymb, amsbsy, mathtools, xcolor, bm, bbm, graphicx, changepage, cleveref}
\usepackage{xcolor}
\usepackage{physics}
\usepackage{epstopdf}
\usepackage{framed}
\usepackage{multirow}
\usepackage{color}
\usepackage{longtable}
\usepackage{comment}
\usepackage{qcircuit}
\usepackage{faktor}

\makeatletter
\newtheorem*{rep@theorem}{\rep@title}
\newcommand{\newreptheorem}[2]{%
\newenvironment{rep#1}[1]{%
 \def\rep@title{#2 \ref{##1}}%
 \begin{rep@theorem}}%
 {\end{rep@theorem}}}
 
 \newcommand{\para}[1]{~\\ \noindent {\bf #1}}
\makeatother

\newtheorem{remark}{Remark}
\newtheorem{theorem}{Theorem}
\newtheorem{proposition}{Proposition}

\newtheorem{lemma}{Lemma}
\newtheorem{assumption}{Assumption}
\newtheorem*{assumption*}{Assumption}
\newtheorem{corollary}{Corollary}
\newtheorem{definition}{Definition}

\newcommand{\sgn}[1]{\mathsf{sgn}\left[{#1}\right]}
\newenvironment{thmbis}[1]
{\renewcommand{\thetheorem}{\ref{#1}}%
   \addtocounter{theorem}{-1}%
   \begin{theorem}}
  {\end{theorem}}
\newenvironment{propbis}[1]
{\renewcommand{\theproposition}{\ref{#1}}%
   \addtocounter{proposition}{-1}%
   \begin{proposition}}
   {\end{proposition}}

\newcommand{\Prob}[1]{\mathrm{Pr}\left({#1}\right)}

\begin{document}

\title{Robust and Efficient Hamiltonian Learning}

\author{Wenjun Yu}
\email{wenjunyus@gmail.com}
\affiliation{QICI, Department of Computer Science, The University of Hong Kong, Pokfulam Road, Hong Kong, China}
\affiliation{Center on Frontiers of Computing Studies, Peking University, Beijing 100871, China}

\author{Jinzhao Sun}
\affiliation{Clarendon Laboratory, University of Oxford, Oxford OX1 3PU, UK}
\affiliation{Quantum Advantage Research, Beijing 100080, China}

\author{Zeyao Han}
\affiliation{School of Physics, Peking University, Beijing 100871, China}

\author{Xiao Yuan}
\email{xiaoyuan@pku.edu.cn}
\affiliation{Center on Frontiers of Computing Studies, Peking University, Beijing 100871, China}

\begin{abstract}
With the fast development of quantum technology, the sizes of both digital and analog quantum systems increase drastically. 
In order to have better control and understanding of the quantum hardware, an important task is to characterize the interaction, i.e., to learn the Hamiltonian, which determines both static and dynamic properties of the system. 
Conventional Hamiltonian learning methods either require costly process tomography or adopt impractical assumptions, such as prior information on the Hamiltonian structure and the ground or thermal states of the system. 
In this work, we present a robust and efficient Hamiltonian learning method that circumvents these limitations based only on mild assumptions. 
The proposed method can efficiently learn any Hamiltonian that is sparse on the Pauli basis using only short-time dynamics and local operations without any information on the Hamiltonian or preparing any eigenstates or thermal states.
The method has a scalable complexity and a vanishing failure probability regarding the qubit number. Meanwhile, it performs robustly given the presence of state preparation and measurement errors and resiliently against a certain amount of circuit and shot noise.
We numerically test the scaling and the estimation accuracy of the method for transverse
field Ising Hamiltonian with random interaction strengths and molecular Hamiltonians, both with varying sizes and manually added noise. 
All these results verify the robustness and efficacy of the method, paving the way for a systematic understanding of the dynamics of large quantum systems.

\end{abstract}

\maketitle

\section{Introduction}

Recently, we have witnessed a rapid development of quantum technologies with drastically increasing sizes of fully or partially controllable quantum systems \cite{Arute2019,neill2021accurately,ebadi2021quantum,wu2021strong,gong2021quantum,zhong2020quantum,mi2021time,kandala2017hardware,kjaergaard2019superconducting,zhang2017observation}.
Following the law of quantum mechanics, the quantum devices allow for controllable quantum operations and offer opportunities to naturally probe quantum many-body behaviors in a programmable or analog way. 
To further upgrade the quantum technology or even to realize a universal quantum computer, an essential requirement is to better understand the underlying quantum mechanisms of the systems~\cite{georgescu2014quantum,altman2021quantum,hauke2012can}.
According to the Schr\"{o}dinger equation, the static and dynamical behaviors of quantum systems are described by their Hamiltonian.
Successful detection of the Hamiltonian, i.e., the so-called \emph{Hamiltonian learning}~\cite{granade2012robust, wiebe2014quantum,wiebe2014hamiltonian,krastanov2019stochastic,wang2017experimental,o2021quantum,haah2021optimal} task, could guide us to have a better understanding and hence control of analog quantum simulators~\cite{lloyd1996universal,das2008colloquium,friedenauer2008simulating,aspuru2012photonic,georgescu2014quantum,hauke2012can}, digital quantum processors~\cite{bernstein1997quantum,Shor1995,grover1996fast,harrow2009quantum,krantz2019quantum,kjaergaard2019superconducting,trabesinger2012quantum}, quantum sensors~\cite{pirandola2018advances,degen2017quantum,boss2017quantum}, etc.


A straightforward choice of learning Hamiltonian is via \emph{quantum process tomography}~\cite{Chuang1997,PhysRevLett.90.193601,leung2003choi,Merkel2012,rahimi2011quantum}.
However, since the dimension of a many-body Hamiltonian grows exponentially with the system size, complete tomography of the Hamiltonian is costly and formidable in real experiments~\cite{mohseni2008quantum,baldwin2014quantum}.
Here we briefly review some other efficient learning methods and describe their limitations.
Wiebe et al.~\cite{wiebe2014quantum,wiebe2014hamiltonian} firstly use a trustworthy quantum simulator to inverse the unknown evolution and iteratively approximate the exact inverse. 
However, since the scaling of complexity is case-specific, the order of magnitude thereof could be prohibitively large for a certain Hamiltonian family.
Besides, the accurate simulation of some arbitrary inverse Hamiltonian evolution remains hard for experimental realizations in NISQ era~\cite{preskill2018quantum}.
To avoid the usage of expensive and inaccessible quantum resources, some protocols focus on using states that contain information on the Hamiltonia, including steady states and thermal states~\cite{Bairey_2019,evans2019scalable,anshu2020sample,qi2019determining}.
These methods are scalable with the system size by measuring states with desired properties to infer the Hamiltonian.
However, they rely on specific states containing the information of the target Hamiltonian, which are generally hard to prepare and are vulnerable to measurement errors. 

There are also learning proposals that get rid of the requirements of preparing some special initial states.
For example, Li et al.~\cite{quench} construct equations based on the energy conservation of arbitrary states during evolution.
Given the structure of the Hamiltonian, this work considers the effects of measurement errors and gives a bound on the measurement errors based on the gap between the smallest two singular values of the equation matrix.
Notably, this protocol is steady under both state preparation and measurement (SPAM) errors.
This steadiness, nevertheless, relies heavily on the model assumptions of the measurement errors and Hamiltonian structure (ansatz) errors as well as the eigenstate thermalization hypothesis, and it is hard to give a firm claim for the real-world cases based on these constraints.
Zubida et al.~\cite{zubida2021optimal} estimate the Hamiltonian parameters by constructing equations from Ehrenfest theorem about the derivatives of observables involved in the structure of the Hamiltonian.
The number of measurements scales with precision quartically $\Omega(1/\epsilon^{4})$ due to the trade-off between statistical and systematic errors, but its accuracy would be affected by state preparation and measurement errors.
Hangleiter et al.~\cite{hangleiter2021precise} measure state properties by evolving it to identify the nearest-neighbor coupled Hamiltonian in superconducting systems.
Since the structure and system nature are known, they show that measurements contain the eigenvalues of the Hamiltonian.
By fixing a time evolution as a base, this method can be robust to state preparation errors.
Nevertheless, all of the above three methods need a precise \emph{prior} about the structure of the Hamiltonian, rendering them unfavorable for tasks of learning an unknown system from scratch or reliably verifying Hamiltonian simulation with extra noise terms.

In this work, we report an efficient Hamiltonian learning algorithm that circumvents the above requirements and assumptions. Our method only utilizes short-time Hamiltonian evolution~---~getting rid of the requirement of preparing ground or thermal states, assumes sparse Hamiltonians~---~avoiding the second assumption without specifying the Hamiltonian structure~\footnote{The sparsity is defined with respect to the decomposition of the Hamiltonian in the Pauli matrix basis. Note that realistic Hamiltonian, including molecules, lattice models, quantum field theories, etc., all satisfy the sparsity requirement.}, and exploits ideas from randomized benchmarking (RB)~\cite{Flammia2019}~---~being robust against effects from SPAM errors.
The method requires only local input states, local Pauli measurements, and local Pauli operations inserted in the Hamiltonian dynamics, which is accessible for various experimental platforms, such as superconducting qubits~\cite{neill2021accurately,kjaergaard2019superconducting,krantz2019quantum,mezzacapo2014digital,garcia2015fermion,asaad2016independent,weber2017coherent}, trapped ions~\cite{zhang2017observation,haffner2005scalable,blatt2012quantum}, Rydberg atoms~\cite{saffman2016quantum,ebadi2021quantum}, nuclear magnetic resonance~\cite{cory1998nuclear,o2021quantum}, etc.

After our paper was posted to the arXiv, there come out various interesting proposals toward efficiently learning unknown Hamiltonians~\cite{franca2022efficient,gu2022practical,wilde2022scalably,huang2022learning}.
Wilde et al.~\cite{wilde2022scalably} recruit the first-order Trotter to simulate Hamiltonian with \emph{prior} parameters and use the effective fidelity between the simulated and evolved states to generate the loss function.
By imposing constraints on the loss function, they also show the accuracy guarantee of the optimization solution. 
Huang et al.~\cite{huang2022learning} randomly implement Pauli gates to project the target Hamiltonian onto some stabilizer spaces in the Pauli basis, so the eigenspace of the new Hamiltonian is easily determined.
They achieve the Heisenberg-limited scaling complexity by introducing the phase estimation method as a subroutine and assuming the bounded-degree structure.
We also expect to find more and more interesting proposals for learning an unknown Hamiltonian with fewer assumptions about the Hamiltonian itself.

\para{Contributions.}
For the first time, we introduce a provably efficient and robust method to detect an arbitrary Hamiltonian operator containing sparse decomposition parameters over Pauli basis.
In order to be SPAM error robust, our protocol employs the exponential fitting idea from randomized benchmarking to estimate the absolute values of parameters and uses simplified process tomography equations to fulfill the need of sign estimation.
The sign estimation remains robust as it only returns discrete values. 
Based on this protocol, we give a provable guarantee in Section~\ref{sec:mainresult} to the estimation results sketched from theoretical bounds of all subroutines.
Given the sparsity $s$ and qubit number $n$, we achieve a measurement complexity as $\order{{sn}/{\epsilon^4}}$ where $\epsilon$ depicts the desired precision.
We simulate the protocol on multiple systems with various Hamiltonian structures in Section~\ref{sec:numerical} to verify and monitor the scaling properties, including an 8-qubit molecular Hamiltonian with over 100 nonzero Pauli terms.
All these results conclude that our Hamiltonian learning protocol on $n$-qubit systems is scalable and robust. It is also comparably practical since it does not require the correct Hamiltonian structure or the preparation of eigenstates or thermal states and only needs single-qubit operations.

During the estimation of the absolute values (which we denote by stage 1), the protocol estimates the Pauli fidelity of the Hamiltonian evolution channel and transforms these values to be Pauli error rates, which are equal to squares of the decomposition parameters.
We construct the fidelity estimator by implementing the random circuits~\cite{Flammia2019} in an interleaved manner to estimate the extraneous Hamiltonian channel.
We need to assume that the circuit noise performs benignly as in Assumption~\ref{assump:gtm} and~\ref{assump:Pauli}, and the exponential decay and fittings achieve SPAM-robust estimation.
Suppose Hamiltonian operators to be sparse as in \textbf{A1} and \textbf{A2}, we adopt the sparse transforming idea~\cite{harper2020fast} to reduce the exponential calculation to an efficient sampling.
We also relieve the original assumption about the unbiased Gaussian noise to biased noise.
In the second stage, we introduce an optimization problem inspired by the process tomography protocol~\cite{Shabani2011} to estimate the sign information.
With the assumption about the SPAM errors in Assumption~\ref{assump:SPAM}, the sign estimator can be SPAM robust as it tolerates small disturbances.

\para{Organization.}
The remainder of this paper is organized as follows. 
We provide the basic problem setting and intuition in \Cref{sec:setting}. 
We present an overview and the main ideas of our protocol in \Cref{sec:Main} with a rigorous analysis of the validity of the protocol in \Cref{sec:mainresult}.
The numerical simulations for the Hamiltonian learning protocol are exhibited in \Cref{sec:numerical}.
We conclude the results in \Cref{sec:con}.


\section{Background}\label{sec:setting}
In this section, we introduce the main setting of the Hamiltonian learning problem and some basics thereof.
We then summarize the assumptions about the Hamiltonian in our protocol.
We also introduce the relationship between the time evolution channel and the Hamiltonian operator as a tool to extract the estimation.
In order to better control the SPAM errors and circuit noise, we propose some restrictions to keep them performing benignly.
The assumptions and constraints help to learn the information of the Hamiltonian operator robustly and efficiently.

\subsection{Problem Setting}
The Hamiltonian learning problem aims to estimate the unknown underlying Hamiltonian operator of a quantum system.
Consider an $n$-qubit Hermitian Hamiltonian $H$ with dimension $d\times d$ and $d=2^n$,
which can be decomposed on the Pauli operator basis as
\begin{gather}\label{eq:decomp}
    H=\sum_{P_\alpha \in{\sf P}^n}s_{\alpha}P_{\alpha},
\end{gather}
with real coefficients $s_\alpha$. 
Here the ${\sf P}^n=\{I,X,Y,Z\}^{\otimes n}$ is the \emph{refined Pauli group} consisting of $4^n$ Pauli matrices.
We label each element therein by a $2n$-bit string as $\alpha$ for $P_\alpha$, and Appendix~\ref{sec:PauliConcept} includes the detailed definitions of this labeling and other Pauli-related concepts.
We denote the vector ${\bf s}\in\mathbb{R}^{d^2}$ by the \emph{decomposition parameter}, and each element can be obtained by $s_\alpha=\frac{1}{d}\Tr(HP_\alpha)$ since ${\sf P}^n$ forms a basis under the Hilbert Schmidt norm.
Consequently, the Hamiltonian learning problem is treated as the reconstruction of the decomposition parameters ${\bf s}$. 

In the literature, several Hamiltonian learning methods have been proposed to efficiently learn the decomposition parameters ${\bf s}$~\cite{Bairey_2019,hou2020determining,PhysRevX.8.021026,anshu2020sample,qi2019determining,hangleiter2021precise} by imposing 
\begin{enumerate}[(1)]
    \item a certain Hamiltonian structure, such as nearest-neighbor or local interactions
    \item the ability to prepare an eigenstate or thermal state of the Hamiltonian.
\end{enumerate}
\noindent Yet, these constraints may fail in practical scenarios.

For the first requirement, a realistic system may have different types of interactions, and it would be impractical to assume a specific Hamiltonian structure. For example, we may not have full prior information on the Hamiltonian structure for some partially controllable systems~\cite{altman2021quantum}. Meanwhile, even if we know the Hamiltonian of an ideal quantum system, the imperfection of the system may lead to unknown and complex interactions~\cite{hauke2012can,altman2021quantum,zhang2017observation}, such as the multiqubit crosstalk of superconducting qubits~\cite{krantz2019quantum,kjaergaard2019superconducting,takita2017experimental,kandala2019error,sun2020mitigating}. 
Indeed, the dimension of the vector $\bf s$ increases exponentially with the number of qubits, and a specific structure of the Hamiltonian will restrict the problem size to an exponentially smaller subspace. In this work, we release the assumption to avoid prior information on the Hamiltonian. We notice that while the Hamiltonian structure may be unknown, most real-world physical systems, such as superconducting qubits, cold atoms, etc., generally have sparse decomposition parameters ${\bf s}$.
For convenience, we denote the set of nonzero parameters by ${\bf s}^\star$.
In this work, our protocol and analysis mainly focus on the Hamiltonian with sparse nonzero parameters, and we make the following assumptions to determine the decomposition parameters ${\bf s}$ of the target Hamiltonian.


\begin{assumption*}
\label{assump:HamiltAssump}
Given Hamiltonian $H$, we have
\begin{itemize}
    \item[$\bf A1$] \textit{(Random sparse support)} The support of the nonzero decomposition parameters ${\bf s}^\star$ is chosen uniformly at random from all subsets with size $s$ of the group ${\sf P}^n$. Here $s$ is polynomial with the qubit number, $n$.
    \item[$\bf A2$] \textit{(Spectral gap)} The smallest nonzero parameter of ${\bf s}^\star$ is larger than a constant $\sqrt{\epsilon_0}$.
\end{itemize}
\end{assumption*}
\noindent $\bf A1$ ensures the sparsity of the Pauli decomposition of the target Hamiltonian, which is much more general compared to the structure assumptions, i.e., indicating which Pauli term has a nonzero coefficient.
Although this assumption supposes uniformly random support, this is due to the concern of theoretical analysis and proofs.
In the numerical results, we will show the feasibility of our protocol when learning sparse Hamiltonians with different structures, including a comparably complicated molecular Hamiltonian.
The assumed lower bound of the spectral gap in $\bf A2$ helps to distinguish the nonzero parameters from overwhelming noise effects.
These two assumptions help to narrow the estimation to the sparse parameters ${\bf s}^\star$, which makes the execution efficient regarding the number of qubits.

For the second requirement, preparing the eigenstate or thermal state of an arbitrary Hamiltonian is proven to be QMA-hard~\cite{kempe2006complexity}, and this eigen- or thermal state preparation also seems impractical. Instead of solving those computationally hard problems, a more natural way is to consider the time evolution of the Hamiltonian, which is accessible for both analog and digital quantum computers~\cite{georgescu2014quantum,trabesinger2012quantum,childs2018toward,low2019hamiltonian,sun2021perturbative}. Using the Schr\"{o}dinger equation, 
the time-evolved state for an initial state $\rho$ is represented as a unitary
\begin{gather}
    \mathcal{H}_t(\rho)=U_H(t)\rho U_H(t)^\dagger,
\end{gather}
where $U_H(t)=e^{-iHt}$ corresponds to Hamiltonian $H$ and time $t$. 
Here and henceforth, we use $\mathcal{H}_t$ to denote the \emph{Hamiltonian evolution channel} with time $t$.
For a short enough time, we can approximate the Hamiltonian evolution channel by expanding and truncating the infinite series as
\begin{align}\label{eq:Hamiltonian}
    &\mathcal{H}_t(\rho)
    =\rho+it\sum_{\alpha\in{\sf P}^n}s_\alpha(\rho P_\alpha-P_\alpha\rho)\notag\\
    +&t^2\sum_{\alpha,\beta\in{\sf P}^n}s_\alpha s_\beta\left[ P_\alpha\rho P_\beta-\frac{1}{2}(P_\alpha P_\beta\rho+\rho P_\alpha P_\beta)\right]+o(t^3).
\end{align}
For simplicity, here and after in this paper we will use the index $\alpha\in{\sf P}^n$ to indicate the corresponding operator $P_\alpha\in{\sf P}^n$ when it would not incur ambiguity.
We note that the first- and second-order terms of the evolution channel faithfully store the decomposition parameters of the Hamiltonian operator, which offers the possibility for extracting the Hamiltonian information from the dynamics.
Together, we will utilize the short-time evolution under the target unknown Hamiltonian and present a Hamiltonian learning protocol for any unknown Hamiltonian with sparse decomposition parameters. 

\subsection{Noise Models}
In way to construct a robust Hamiltonian learning protocol, a crucial obstacle remaining is the unknown SPAM errors and circuit noise.
Due to the limitation of the quality and controlling quantum systems, these errors and noise are ubiquitous in NISQ devices.
The imperfect state preparation and readout noise will undermine the availability of measurement outcomes, which renders the estimation untrustworthy.
To make the protocol SPAM robust, we first need  to restrict the strength of these errors; otherwise, the unbounded errors will always be able to destroy all the effective information from measurements.
For theoretical convenience, we will hold the following assumption to quantify the device's SPAM errors.
\begin{assumption}\label{assump:SPAM}
We denote the noisy state by $\tilde{\rho}$ and the noisy measurement operator by $\tilde{M}$.
The noisy device satisfies the following constraints for every pair of positive Pauli eigenstates and Pauli measurements,
\begin{gather}
    \abs{\Tr(\tilde{M}\mathcal{H}(\tilde{\rho}))-\Tr(M\mathcal{H}(\rho))}\leq\tau,
\end{gather}
where $\mathcal{H}$ here represents the Hamiltonian evolution with all possible short evolving $t$ used in the implementation.
The positive bound $\tau$ must be in the region of $\left[0,O\left(t_1\sqrt{\frac{\epsilon_0}{s}}\right)\right]$ where $t_1$ is a typical unit time used in the regression during the sign estimation.
\end{assumption}

In order to eliminate these depicted SPAM effects, we will show shortly that the protocol would choose special circuit components, especially random Pauli gates, and implement them in a interleaved way to build a SPAM-robust estimation.
In this sense, we need to clarify the form and amplitude of the introduced gate noise in a rigorous way.
Following the idea and convention in randomized benchmarking protocols, \cite{Flammia2019,helsen2020general,helsen2019multiqubit} we impose time-independent and Markovian quantum noise.
Furthermore, for every noisy implementation of the Pauli gate, we assume that there exists some certain noise channel $\Lambda$ such that
\begin{gather}\label{eq:independent}
    \tilde{\mathcal{P}}=\Lambda\circ\mathcal{P},\ \forall\,P\in {\sf P}^n,
\end{gather}
where $\mathcal{P}$ is the implementation of the unitary $P$.
We denote the noise effect satisfies Eq.~\eqref{eq:independent} by gate-independent noise.
Therefore, we have the following model for circuit noise.
\begin{assumption}\label{assump:gtm}
The noise from the implementation of quantum circuits is {\bf \emph{time-independent}}, {\bf \emph{Markovian}}, and {\bf \emph{gate-independent}}.
\end{assumption}
\noindent Even though only gate-independent noise allows the RB circuit to construct the standard twirling structure, Wallman~\cite{Wallman2018} proved that a slight gate-dependent disturbance of the gate-independent noise still keeps the twirling results and causes limited harm to the final analysis of RB.
This might give us a clue about the possibility to relax this gate-independent requirement.

An additional assumption is ad-hoc for the execution of the interleaved protocol.
Since the implementation of random Pauli gates carries some noise inevitably, our protocol needs to calibrate and compensate for these effects to extract the net information of the Hamiltonian.
However, in order to eliminate them, the noise is supposed to be diagonal in the superoperator representation, i.e., a stochastic channel.
\begin{assumption}\label{assump:Pauli}
The noise $\Lambda$ of the circuit implementation for Pauli gates is a stochastic channel (Pauli channel), and it is close to the identity according to the operator norm as 
\begin{gather}
    \|\Lambda-\mathcal{I}\|\leq\frac{1}{3}.
\end{gather}
\end{assumption}
Note that Pauli gates are actually implemented from single-qubit pulses in realization. 
From some current experimental verification and detection results \cite{harper2020efficient,Barends2014,somoroff2021millisecond}, the implementation of local operations can be close to their ideal version up to high precision.
Therefore, Assumption~\ref{assump:Pauli} is reasonable.
In fact, the constant $\frac{1}{3}$ can be further enlarged to an arbitrary constant less than 1, which still makes sure that all the Pauli fidelity terms of $\Lambda$ are larger than 0.

\begin{figure*}[t]
    \centering
    \includegraphics[width=2\columnwidth]{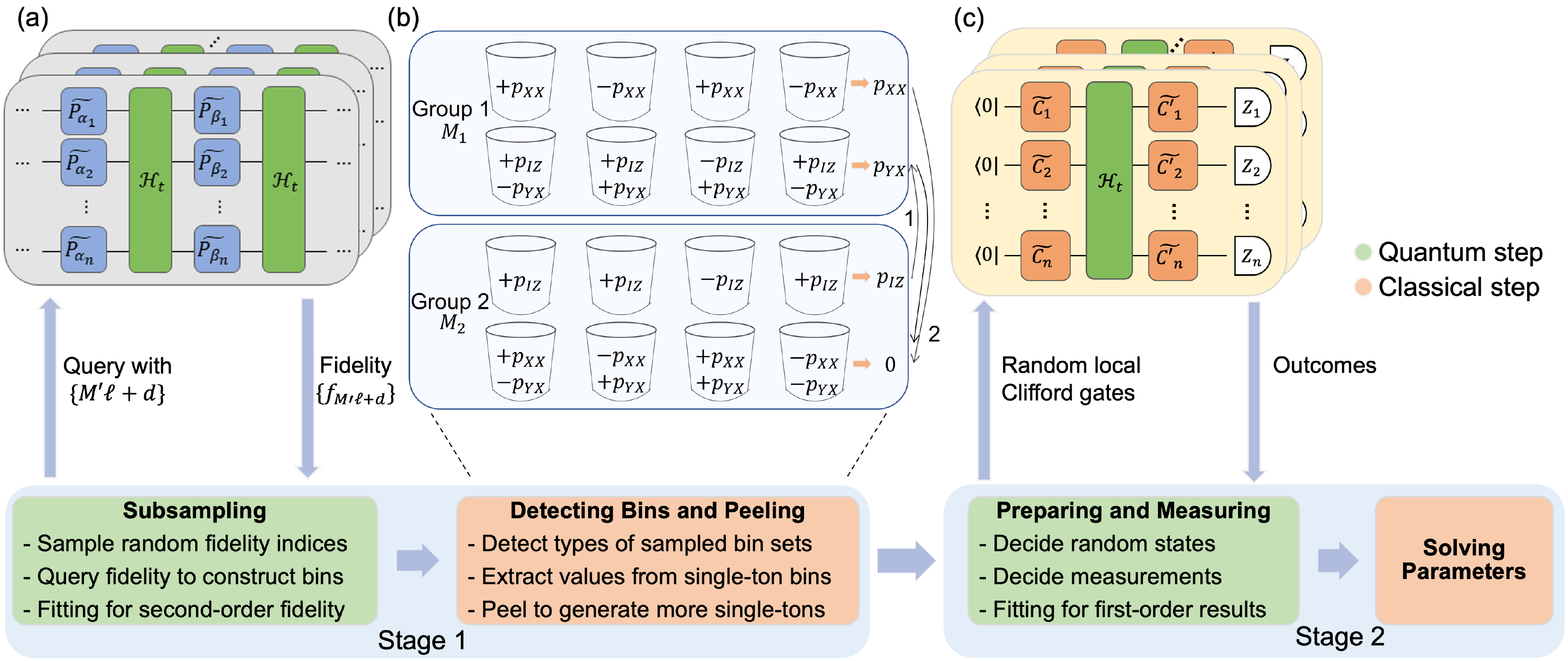}
    \caption{An overview of the Hamiltonian learning execution. The first stage can be viewed as a Pauli error rate estimator along with an independent Pauli fidelity estimator. In circuit (a), we introduce the kernel of the random circuits to get the desired fidelity terms, where $\tilde{P}$ denotes the noisy Pauli gate. By choosing the random matrix $M$, the index $\ell$, and the offset $d$, the protocol accordingly queries the fidelity terms $f_{M'\ell+d}$, where $M'$ is the shuffled $M$ as discussed in Section~\ref{sec:sparse}.
    According to the query, the procedure then prepares input states followed by the cascading structures including the noisy Pauli gates and Hamiltonian evolution.
    The sub-subscripts represent the corresponding single-qubit parts of Pauli gates.
    By measuring circuits with different depths, it acquires the fidelity terms with bounded noise and combines them to be the bin sets as in Eq.~\eqref{eq:hashfunction}.
    The formal illustration of the circuits can be viewed in Section~\ref{sec:fidelityEstimation}.
    In subfigure (b), the procedure uses multiple subsampling groups and bin sets aligned therein to detect the type of each bin set and peel to generate bins with single nonzero Pauli error rates. For example, in this two-qubit case, the procedure finds the first bin set in the first group to be single-ton and extracts the nonzero term and index, $p_{XX}$, with the aid of different signs. For the second bin set, the procedure gets stuck due to the summation of multiple terms. It then switches to the second group to detect the first bin set with result $p_{IZ}$. By learning $p_{IZ}$ is nonzero, the procedure uses it to peel that stuck case and gets $p_{YX}$. 
    Then the procedure uses estimated error rates to peel the second bin set of group two and certifies it as a zero-ton bin.
    This peeling can be viewed more detailedly in Section~\ref{sec:sparse}.
    The second stage aims to complete the Hamiltonian learning by estimating signs of decomposition parameters as illustrated in Section~\ref{sec:sign}. In circuit (c), the procedure inputs random local Pauli eigenstates and measures them with Pauli operators after evolving under Hamiltonian, where $\tilde{C}$ and $\tilde{C'}$ denote the noisy local Clifford gates for state preparation and measurements which are not conjugate. By constructing the relationship between the measurement outcomes and the unknown parameters, we can solve them and recover the sign information. }
    \label{fig:peeling}
\end{figure*}

\section{Protocol Overview}\label{sec:Main}

In this section, we illustrate the main ingredients of our protocol to estimate the unknown Hamiltonian operator. 
The basic idea is to extract the information of the sparse parameter vector $\bf s^\star$ sequentially from the second- and first-order terms of the series expansion of Hamiltonian evolution in Eq.~\eqref{eq:Hamiltonian}. 
In this sense, our protocol is naturally analog to a quantum channel detection scheme. 


We briefly summarize our protocol as follows. We also refer to Figure~\ref{fig:peeling} for a schematic summary. 
\begin{itemize}
    \item At the first stage, the procedure first  extracts the Pauli fidelity with Pauli inputs and measurements in a similar vein to RB as in Figure~\ref{fig:peeling} (a), which is robust against SPAM errors.
    By linear regression with $t^2$, the results correspond to the second-order terms in the evolution channel.
    Then we show a procedure to efficiently transform the extracted fidelity to the squares of $\bf s^\star$, and an example is illustrated in Figure~\ref{fig:peeling} (b).

    \item At the second stage, we construct and solve the process equations to estimate the remaining sign information of $\bf s^\star$ from the first-order terms as in Figure~\ref{fig:peeling} (c). 
    With knowledge of the support position about $\bf s^\star$ from the first stage and the nature that small noise on results will not cause any sign flips, the sign estimator can be executed robustly and efficiently for both quantum and classical complexity.
\end{itemize}

\noindent One may argue that a straightforward strategy is to directly consider the first-order expansion of Eq.~\eqref{eq:Hamiltonian} and extract $\bf s^\star$ with certain input states and measurements. 
However, the first-order term corresponds to coherent (off-diagonal) information of the channel, which is generally sensitive to SPAM errors. Indeed, most quantum channel detection schemes that are robust to SPAM errors, such as randomized benchmarking~\cite{Magesan2012a,Dankert2009,Erhard2019,Helsen2018}, can only extract incoherent (diagonal) information of channels. 
Moreover, as in Eq.~\eqref{eq:linearordereq}, the overhead  for  classical computing to construct the process equations from state preparations and measurements would be exponentially large without the information of the sparse support of the Hamiltonian effect since there are exponential  items on the Pauli basis.
Therefore, we choose this composite algorithm to circumvent these obstacles.


We discuss in detail the two stages in the remainder of this section and refer to the next section and appendix for an informal summary and detailed proofs of the protocol, respectively.


\subsection{Stage 1.~Absolute value estimation}\label{sec:Pauli}

The first stage is to estimate the squares (absolute values) of the decomposition parameters, which come from the second-order effects of the Hamiltonian evolution.

Considering the $\chi$-representation using the Pauli basis, the Hamiltonian evolution channel can be decomposed as
\begin{align}\label{eq:processmatrix}
    \mathcal{H}_t(\rho)=\sum_{\alpha,\beta\in{\sf P}^n}\chi_{\alpha\beta}P_\alpha\rho P_\beta,
\end{align}
where $p_\alpha\coloneqq\chi_{\alpha\alpha}$ is denoted by the \emph{Pauli error rate}.
Comparing Eq.~\eqref{eq:processmatrix} with Eq.~\eqref{eq:Hamiltonian}, the relationship between the parameters in $\bf s$ and Pauli error rates $\{p_\alpha\}_{\alpha\in{\sf P}^n}$ is
\begin{gather}\label{eq:chi}
    p_\alpha=
    \begin{cases}
    s_\alpha^2 t^2+o(t^3)& \text{$\alpha\neq 0$}\\
    1-t^2\sum_{\beta\in{\sf P}^n,\beta\neq0}s_\beta^2+o(t^3)& \text{$\alpha=0$}
    \end{cases}.
\end{gather}
Since the square of decomposition parameter $s_\alpha^2$ is proportional to the Pauli error rate $p_\alpha$ with a small remainder from small $t$, we can approximately obtain every $|s_\alpha|$ given the corresponding Pauli error rate $p_\alpha$. 
Note that the sparsity of $|s_\alpha|$ also implies the sparsity of $p_\alpha$. 
Therefore, the basic idea is to exploit the recently proposed Pauli estimation protocol to learn all the sparse Pauli error rates~\cite{harper2020fast} from Pauli fidelity (which is defined later). 
Nevertheless, there are two challenges if we directly adopt existing Pauli estimation protocols. 
\begin{itemize}
    \item First, the Hamiltonian channel also contains coherent parts, while we only need the incoherent information here.
    \item Second, Pauli error rates of a short-time Hamiltonian channel are only approximately sparse to the second-order expansion. It is actually dense if we consider infinite expansions. 
\end{itemize}
Fortunately, we can circumvent these in our estimation.
For the first one, our circuit would apply the Pauli twirling to the Hamiltonian evolution channel to convert it into a Pauli (diagonal) channel. 
In brief, \emph{Pauli twirling} annihilates the off-diagonal elements of a channel and keeps all the Pauli (diagonal) information the same. 
Denote a Pauli gate by $\mathcal{P}_\alpha$, and the evolution channel $\mathcal{H}_t$ after Pauli twirling must be
\begin{gather}\label{eq:PauliTwirling}
    \mathcal{H}_t^{{\sf P}^n}(\rho)\coloneqq\frac{1}{\abs{{\sf P}^n}}\sum_{\alpha\in{\sf P}^n}\mathcal{P}_\alpha^\dagger\circ\mathcal{H}_t\circ\mathcal{P}_\alpha(\rho)=\sum_{\alpha\in{\sf P}^n}p_\alpha P_\alpha \rho P_\alpha.
\end{gather}
For the second one, we apply regression to the results with different small $t$ to suppress the effects induced by higher-order terms in the evolution. 
We further prove that the protocol is noise-resilient and efficient for large-scale Hamiltonian evolution channels.


\begin{figure}[t]
    \centering
    \includegraphics[width=\columnwidth]{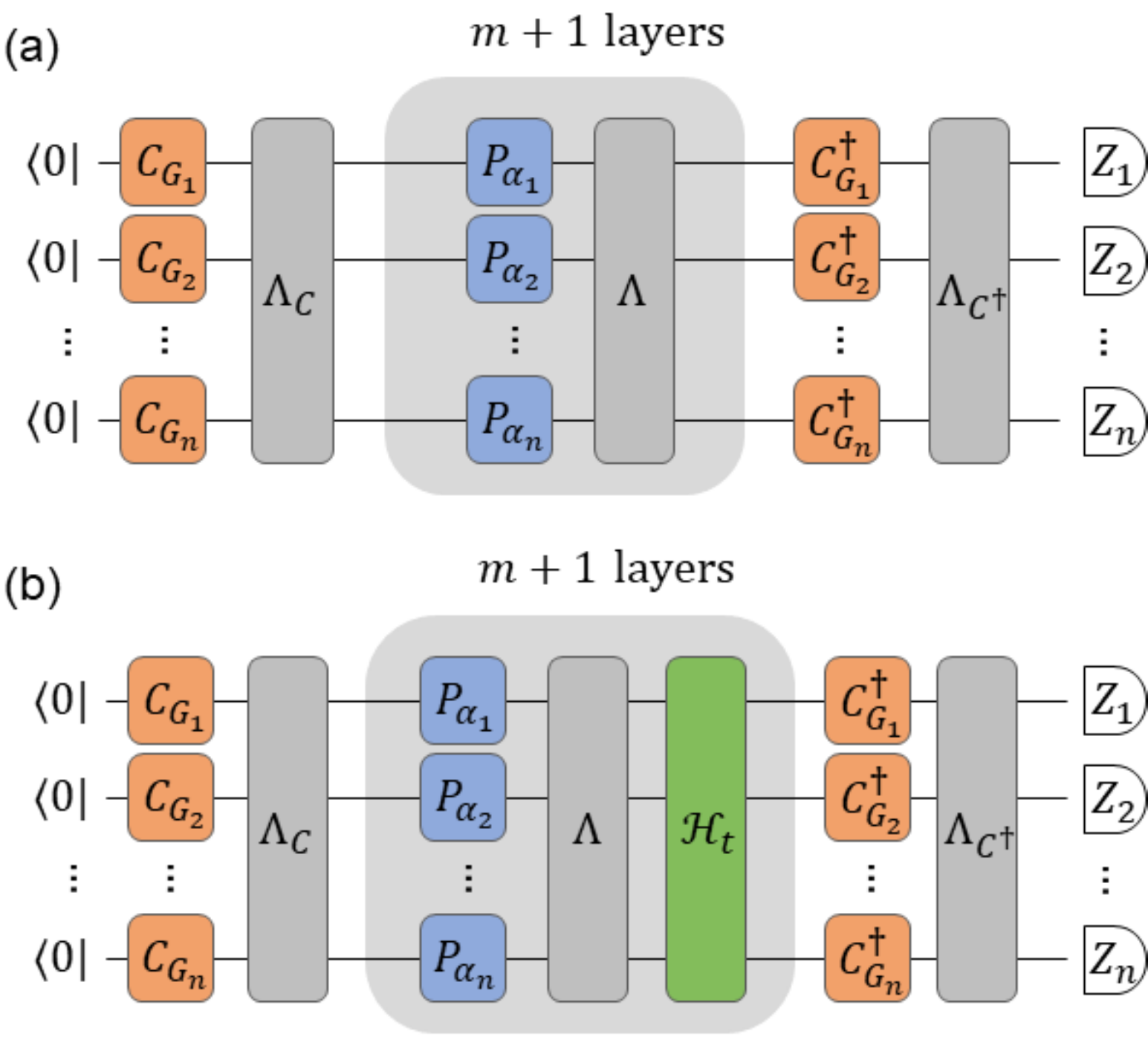}
    \caption{An illustrative diagram of the gate-cascading estimator of the Pauli fidelity of the Hamiltonian evolution channel. This estimator performs robustly against both SPAM errors and the implementation noise from Pauli gates. (a) In order to rule out the effects of the implementation noise $\Lambda$, we recruit the circuit in~\cite{Flammia2019} as a reference circuit. The procedure inputs the stabilizer states of stabilizer group $G$ in the Pauli group and measures the Pauli syndromes using local Clifford gates $C$ and $C^\dagger$.
    Random Pauli gates $P$ twirl noise channels into the diagonal version with only fidelity left. 
    The sub-subscripts of the locally tensored gates represent the corresponding single-qubit parts.
    Therefore, the procedure varies the gate length and gets a decaying factor by tracking the overall syndromes from both Pauli gates and the syndrome measurements, which helps to estimate the fidelity information of $\Lambda$. (b) Since we have access to the short-time evolution of the target Hamiltonian, we implement it in the estimation circuit. Similarly to the previous circuit, we estimate the fidelity of the composite channel $\mathcal{H}_t\circ\Lambda$. Under Assumption~\ref{assump:Pauli}, a ratio estimator contributes to the extraction of the net fidelity information of the Hamiltonian channel from the composite results from both circuits.}
    \label{fig:rbcircuit}
\end{figure}

\subsubsection{Pauli fidelity estimation}\label{sec:fidelityEstimation}
In this subsection, we focus on getting data from measuring quantum circuits consisting of the Hamiltonian evolution.
Particularly, the protocol  extracts Pauli fidelity as a type of Pauli information for further analysis and extraction of the target parameters.

Inputting and measuring the same operator $P_x$, the outcome corresponds to a \emph{Pauli fidelity} of an arbitrary channel $\mathcal{E}$ as defined in the following
\begin{equation}\label{eq:Paulifidelity}
    f_x \coloneqq \frac{1}{2^n}\tr[P_x \mathcal{E}(P_x)] = \sum_{\alpha\in{\sf P}^n}(-1)^{\langle x,\alpha\rangle_p}p_\alpha,\ \forall\, x\in{\sf P}^n,
\end{equation}
where $\langle x,\alpha\rangle_p\in\{0,1\}$ is a useful operation named \emph{Pauli inner product}.
Depending on whether $P_\alpha$ and $P_x$ commute or not, $\langle x,\alpha\rangle_p$ returns 0 or 1.
We give the mathematical definition of this inner product using bit strings of indices in Appendix~\ref{sec:PauliConcept}.
Eq.~\eqref{eq:PauliTwirling} and Eq.~\eqref{eq:Paulifidelity} depict the major tool and the main target of the estimation in this step, respectively.





The fidelity estimator contains two circuits to extract the Pauli fidelity of the target Hamiltonian channel, as shown in Figure~\ref{fig:rbcircuit}. The idea of the circuit (a), as well as the Pauli estimation, is proposed in~\cite{Flammia2019}.
We introduce an additional circuit (b) to rule out the noise effects.
Therefore, each round of execution consists of implementing these two circuits.

As shown in circuit (a) of Figure~\ref{fig:rbcircuit}, we prepare a stabilizer state using a local Clifford gate $C$ and $\ket{0}^{\otimes n}$, where the state is the simultaneous positive eigenstate of all Pauli operators in a stabilizer group $G$ that includes $P_x$ given the target fidelity is $f_x$.
Here the group can be chosen arbitrarily given it contains $P_x$.
The concepts related to the stabilizer group in the Pauli group are reviewed in Appendix~\ref{sec:Stabilizer}.
We then implement random Pauli gates, and
these gates twirl the noise channel $\Lambda$ into a Pauli channel $\Lambda^{\sf{P}^n}$ which is a diagonal matrix in the superoperator representation with elements to be all Pauli fidelity.
In this case, the circuit is equal to the combination of diagonal channel from the twirled noise channel.
We finally measure the syndrome $\beta\in A_G=\faktor{{\sf P}^n}{G}$ of the state, \emph{i.e.}, do $C^\dagger$ to the state and measure it on $Z^{\otimes n}$ basis.
A brief introduction of the syndrome measurement is also in Appendix~\ref{sec:Stabilizer}.

The estimator sums all the indices of the sampled Pauli gates as $a$ and decides the corresponding syndrome or the decomposition of $a$ in the quotient group: $\alpha\in A_G$.
It is found that the distribution probability of the overall error syndrome $\alpha+\beta$ is the linear combination of exponentially decaying fidelity terms with signs depending on the inner products of fidelity indices and the overall syndrome.
We can further isolate every single decaying by adding an additional sign for each overall syndrome and taking an expectation over these syndromes as shown in Lemma 6 of \cite{Flammia2019}.
By varying $m$, we can extract the desired fidelity term $f_x$ of $\Lambda$ robustly as SPAM errors do not affect the decay rates.
An advantage of considering the stabilizer formalism is that we can simultaneously estimate all the fidelity in one stabilizer group within one round of experiments by different post-processing.
While it would be hidden in the asymptotic complexity, this advantage can expedite real-world processing considerably.

In circuit (b) of Figure~\ref{fig:rbcircuit}, we interleave Pauli gates and the additional time evolution and keep the state preparation and measurements the same. The target channel now becomes $\mathcal{H}_t\circ\Lambda$ with the same analysis as the previous circuit. We can, therefore, get the Pauli fidelity $f_{\mathcal{H}_t\circ\Lambda}$.
In the remainder of this paper, we will omit the subscript when it is clear of which channel the fidelity is.

With the executions of both circuits, we can eliminate the Pauli gate noise $\Lambda$ and extract the net information of the Pauli fidelity of $\mathcal{H}_t$ solely.
As in assumption~\ref{assump:Pauli}, we assume that the implementation noise $\Lambda$ is itself a Pauli channel, and a simple ratio estimator helps to detect the corresponding Pauli fidelity of the Hamiltonian evolution $\mathcal{H}_t$.
Therefore, we get a noise-robust estimation of the Pauli fidelities of the Hamiltonian evolution channel $\mathcal{H}_t$.

Observing Eq.~\eqref{eq:chi} and Eq.~\eqref{eq:Paulifidelity}, we realize that fidelity terms of the evolution channel also contain higher-order effects of the evolving time.
In order to extract the net second-order coefficients that directly contain the information of decomposition parameters, our procedure applies an ordinary linear regression for multiple pairs of Pauli fidelity of evolution channels and squares of evolving times.
This helps to estimate the \emph{second-order Pauli fidelity}, $\{f_x^{(2)}\}$, of the Hamiltonian evolution channel, which is the second-order coefficient of the ordinary fidelity of the evolution,
\begin{align}\label{eq:fidelityexpan}
    f_{\mathcal{H}_t,x}=\frac{1}{2^n}\tr[P_x \mathcal{H}_t^{{\sf P}^n}(P_x)]=f_x^{(0)}+f_x^{(2)}t^2+o(t^2).
\end{align}
This fidelity denotes the second-order coefficient of the ordinary Pauli fidelity, and we have the following definition from Eq.~\eqref{eq:Hamiltonian} and Eq.~\eqref{eq:fidelityexpan}
\begin{align}
    f_x^{(2)}\coloneqq\sum_{\alpha\in{\sf P}^n}s^2_\alpha[(-1)^{\langle\alpha,x\rangle}-1],
\end{align}
which is time-independent.

We summarize the whole fidelity estimator in Algorithm~\ref{alg:query}.
\begin{table}[thb]
    \begin{algorithm}[H]
      \caption{\label{alg:query} Fidelity Estimator: \textbf{FE}($\sf G,X$, $l,t_0$)}
      \begin{algorithmic}[1]
        \State {\tt Input:} a target set $\sf X$ belonging to a stabilizer group $\sf G$, a positive integer $l$, and the unit time length $t_0$
        \State Initialize the identity array $I$ with size $|\sf{X}|$ to be all 1
        \State $\hat{f}_\Lambda\leftarrow I-\textbf{FEstimator}({\sf G,X},l)$
        \For{$i=1,\cdots,5$}
            \State $t\leftarrow i\cdot t_0$
            \State $\hat{f}_{\mathcal{H}_t\circ\Lambda,i}\leftarrow I-\textbf{HFEstimator}({\sf G,X},l,t)$
            \State $\hat{f}_{\mathcal{H}_t,i}\leftarrow \hat{f}_{\mathcal{H}_t\circ\Lambda,i}/\hat{f}_\Lambda$\Comment \textit{ratio estimator}
        \EndFor
        \State Do ordinary linear regression for every fidelity in $\sf X$ from $(\hat{f}_{\mathcal{H}_{it_0}},(i\cdot t_0)^2)$ pairs and get $\hat{f}^{(2)}$
        \State\Return the second-order fidelity array $\hat{f}^{(2)}$
      \end{algorithmic}
    \end{algorithm}
\end{table}
The subroutines, \textbf{FEstimator} and \textbf{HFEstimator}, denote the implementation of circuit $(a)$ and $(b)$ in Figure~\ref{fig:rbcircuit} with post-processing, respectively.
The algorithm uses the former estimation as a reference and implements the latter for multiple evolving times to both eliminate the Pauli noise channel and extract the second-order terms. 
For instance, here, we choose five evenly-increasing evolving times for Hamiltonian evolutions.
Dividing the composite fidelity by the reference fidelity of noise channel $\Lambda$, we can estimate the fidelity of the Hamiltonian channel with the corresponding evolving time, as shown in Lemma~\ref{lm:fe}.
We denote this as the \emph{ratio estimator}.
We will elaborate on these subroutines in detail in Appendix~\ref{sec:fidelity}.

For this second-order fidelity estimation, we hold a guarantee about the accuracy and complexity of the results.
In this theoretical statement, we assign the ordinary least square as our regression method for convenience.
\begin{proposition}
\label{lm:f2bound}
Let $\sf G$ be a stabilizer group in ${\sf P}^n$ and ${\sf X}\subset{\sf G}$,
and suppose Assumption~\ref{assump:gtm} and~\ref{assump:Pauli} hold. 
For any small $\epsilon,\delta>0$, the following happens with probability $1-\delta$:
Run Algorithm~\ref{alg:query} $\textbf{FE}({\sf X,G},l,t_0)$ with $t_0$ satisfies $\|\mathcal{H}_{5t_0}-\mathcal{I}\|<\frac{1}{4}$ and  $l = \frac{2}{\epsilon^2} \log\bigl(\frac{2|\kappa|}{\delta}\bigr)$ where $\kappa$ is the set of variant sequence lengths. regression, and the resulting regression array $\hat{f}^{(2)}$ satisfies
\begin{gather*}
    \abs{\hat{f}^{(2)}_x-f^{(2)}_x}\leq\frac{\sigma O(\epsilon)(f_{\mathcal{H}_t,x}r_{\Lambda,x}+r_{\mathcal{H}_t\circ\Lambda,x})}{\left(f_{\Lambda,x}-O(\epsilon)r_{\Lambda,x}\right)t_0^2}+o(t_0^2),
\end{gather*}
where the upper bound term adopts the largest bound term referred from Lemma~\ref{lm:fe} among multiple evolution with different $t$, and $\sigma$ is a constant regrading the regression method.
\end{proposition}
\begin{proof}
The formal proof can be viewed in Appendix~\ref{sec:fidelity}.
The main idea is to combine the accuracy bound from the subroutines of this algorithm.
By tracking the processes of the ratio estimator and the regression, the accuracy of the second-order fidelity can be determined from those inherited bounds.
\end{proof}
\begin{remark}\label{rm: pha1t}
\rm Note that in this analysis, we choose the ordinary least square regression with an instance of five evenly increasing evolving times, and the parameter to represent the noise propagation during fitting is denoted by $\sigma$.
In order to minimize the overall noise in this phase, we can adjust the unit time $t_0$ to reconcile these two noise sources. 
Specifically, under this fitting model, we need to choose $t_0\sim\sqrt[4]{ O(\sigma\epsilon r_\Lambda)}$.
\end{remark}

Algorithm~\ref{alg:query} serves as a supplier of second-order Pauli fidelity for the later bin construction as stated in Algorithm~\ref{alg:subsampling}, which is a key structure in the transform from fidelity to error rates in Sec.~\ref{sec:sparse}.
The bin construction specifies the fidelity terms needed with indices included in the subset ${\sf X}\subset {\sf G}$ since we do not need the whole estimation as explained in the next subsection.
Algorithm~\ref{alg:query} will return the set of desired second-order fidelity $\sf X$ in stabilizer group $\sf G$.
Even though this estimator can estimate all the fidelity in a stabilizer group using only one round of measurements, we still count each query separately regarding the worst-case complexity.


\subsubsection{From Pauli fidelity to Pauli error rates}\label{sec:sparse}
In the last subsection, we have constructed a second-order fidelity estimator.
According to the previous discussion, the rest of this stage is to use this subroutine to get the desired values, namely, to transform the second-order fidelity to squares of decomposition parameters.
We first define the \emph{second-order Pauli error rate} $p_{\alpha}^{(2)}$ which stands for the second-order coefficient of $p_\alpha$.
According to Eq.~\eqref{eq:chi}, we know the specific representation of $p^{(2)}$:
\begin{gather}\label{eq:chi_wot}
    p_\alpha^{(2)}\coloneqq
    \begin{cases}
    s_\alpha^2& \text{$\alpha\neq 0$}\\
    -\sum_{\beta\in{\sf P}^n,\beta\neq0}s_\beta^2& \text{$\alpha=0$}
    \end{cases}.
\end{gather}
Therefore, the target is reduced to implement the transform from second-order fidelity to second-order Pauli error rates.

We know from the definition that the following \emph{Walsh-Hadamard transform} with Pauli inner product depicts the relation between Pauli fidelity and Pauli error rates.
In the remainder of this paper, we use \emph{Walsh-Hadamard transform} to denote the Pauli variant when it causes no ambiguity.
\begin{align}\label{eq:PauliWHT}
    \begin{split}
        p_\alpha=&\frac{1}{4^n}\sum_{x\in{\sf P}^n}(-1)^{\langle x,\alpha\rangle_p}f_x,\ \forall\, \alpha\in{\sf P}^n,\\
        f_x=&\sum_{\alpha\in{\sf P}^n}(-1)^{\langle x,\alpha\rangle_p}p_\alpha,\ \forall\, x\in{\sf P}^n.
    \end{split}
\end{align}
 Since second-order fidelity and Pauli error rates are linear components of the ordinary fidelity and Pauli error rates, they can also be linked by the Walsh-Hadamard transform.
Even though the second-order error rates are sparse according to Assumption $\bf A1$, the second-order fidelity is generally dense. 
A faithful implementation of Walsh-Hadamard transform thus requires the summation of an exponential number of terms.  
Instead of directly following Eq.~\eqref{eq:PauliWHT} with an exponential number of summations, our protocol recruits the idea for Pauli properties estimation~\cite{harper2020fast} to efficiently obtain the sparse Pauli error information.

Along with the idea firstly introduced in Ref.~\cite{harper2020fast}, our protocol utilizes a special structure named \emph{bin}.
A bin $U$ is a linear combination of a subset of noisy second-order fidelity terms, as defined in Eq.~\eqref{eq:hashfunction}, which also indicates the value of a set of second-order Pauli error rates. 
The procedure first determines a size index $b< 2n$.
Then each bin is constructed from $2^b$ randomly sampled second-order fidelity terms.
In order to implement the random sampling, the procedure employs several $2n\times b$ random binary matrices $M$ and transforms them to corresponding $M'$, where $M'\coloneqq J_nM$ and $J_n\coloneqq I_n\otimes X$. 
By enumerating all $b$-bit string $\ell$, the random matrices $M'$ help to disperse them to $2n$-bit which are valid Pauli indices.
This conversion from $M$ to $M'$ aims to keep the dispersion in the form of Eq.~\eqref{eq:hashfunction}.
Furthermore, the procedure recruits additional $2n$-bit vectors $d$ as offsets.
The detailed construction can be viewed in Algorithm~\ref{alg:subsampling}.
After this construction, we observe that each bin consists of $2^{2n-b}$ second-order Pauli error rates and that different bins have no overlap for a given $M$ (or $M'$) as follows,
\begin{align}\label{eq:hashfunction}
    U[j]&\coloneqq\frac{1}{2^b}\sum_{\ell\in\mathbb{F}^b_2}(-1)^{\langle\ell,j\rangle}\tilde{f}^{(2)}_{M'\ell+d}=\sum_{\alpha:M^T\alpha=j}(-1)^{\langle \alpha,d\rangle_p}\tilde{p}^{(2)}_\alpha.
\end{align}
This equation can be checked in Lemma~\ref{lm:prop_hashing_obs}.
Here $\tilde{f}^{(2)}_{M'\ell+d}$ denotes the noisy Pauli fidelity after the linear regression, and $\tilde{p}^{(2)}_\alpha$ is a noisy second-order Pauli error rate with noise comes from the noisy fidelity.
Note that each bin $U[j]$ selects all the error rates of which indices satisfy $M^T\alpha=j$.
The bin construction actually disperses the whole set of error rates into bins with different indices.
This is similar to a hash function that uses $\alpha$ to uniquely decide the bin index $j$ that contains $\tilde{p}^{(2)}_\alpha$.
Therefore, Given a matrix $M$ (or $M'$), all bins form a partition over the second-order Pauli error rates.
We denote every sampling matrix $M$ (or $M'$) a \emph{subsampling group} as indices it disperses form a subgroup in $\mathbb{F}_2^{2n}$.

Based on the sparse support assumption and a large enough subsampling size $b\sim O(\log s)$, the procedure constructs bins that are likely to contain only one or few nonzero Pauli error rates.
Therefore, we only need polynomial number of fidelity terms.
This detection of the number of nonzero terms in every bin can be achieved with the aid of offset strings $d$.
Intuitively, we construct a bin $U[j]$ with multiple offsets $d$, and we denote these bins with the same bin index $j$ and different offsets to belong to a \emph{bin set}.
Notably, all bins in a single bin set contain the same set of second-order error rates with different signs according to Eq.~\eqref{eq:hashfunction}.
Suppose a bin set only contains one nonzero term $p^{(2)}_\alpha$ and some small noise, we would find absolute values of bins in this bin set remain close.
Similarly, absolute values that are all close to zero indicate that bins contain all zero Pauli error rates, while significant but diverse absolute values imply that there are multiple nonzero rates in a bin set.
Therefore, we can classify bin sets by \emph{zero-ton}, \emph{single-ton}, and \emph{multi-ton} bins.
This bin detector is formalized in Algorithm~\ref{alg:bin_detect} of Appendix~\ref{sec:errorRate}.

When we decide the bin is a single-ton, the averaged absolute value would be regarded as the estimation of that second-order Pauli error rate.
Furthermore, the Pauli index can be learned from signs of bins with different offsets.
According to Eq.~\eqref{eq:hashfunction}, when the noise is not large enough to flip the sign, there comes an equation to describe the linkage between a bin's sign and the Pauli index of the nonzero term,
\begin{gather}
    \sgn{U}=\langle \alpha,d\rangle_p,
\end{gather}
where the sign function is defined as follows,
\begin{equation}
    \sgn{x} = \begin{cases}
    0 & \text{if } x \ge 0,\\
    1 & \text{if } x < 0.
    \end{cases}
\end{equation}
By choosing multiple offsets $d$, the procedure can solve the index $\alpha$.

However, it is hard to deal with a multi-ton bin since it is a summation of multiple error rates, and we can hardly extract the information of any single rates directly.
In order to avoid this stuck case, the procedure uses multiple sampling matrices $M$ and constructs different groups of bins.
In each group, the partition of Pauli error rates varies from others.
In this case, if the procedure finds a multi-ton bin, it is supposed to find some single-ton bins in another group with overlapped Pauli error rates.

By peeling out the overlapping error rates detected from single-ton bins, the multi-ton bin contains fewer nonzero error rates and is degraded to single-ton eventually.
We denote this  by the \emph{peeling process}, and we show an example of the peeling part of a 2-qubit system in Figure~\ref{fig:peeling}.
Here in Algorithm~\ref{alg:peelingIn}, we give a concise overview of how the peeling process works as mentioned above in this subsection and temporarily remove all the technical ingredients.
Shortly, this algorithm invokes Algorithm~\ref{alg:subsampling} to initialize several groups of random bins as above-explained.
With these bins, the peeling process enumerates all of them and uses Algorithm~\ref{alg:bin_detect} to detect their nature, including their types, and possible values and indices.
We can then record the desired second-order Pauli rates when finding single-ton bins.
\begin{table}[thb]
\begin{algorithm}[H] 
  \caption{Peeling Decoder\label{alg:peelingIn} (Informal)}
  \begin{algorithmic}[1]
    \State ${\tt Input}:$ Number of sampling matrices $C$, number of bins $2^b$
    \State ${\tt Initialize}:$ $\mathcal{P}\gets$ empty list to store the sparse indices and second-order Pauli error rates ($\alpha,{p}^{(2)}_{\alpha}$)
    \State Run Algorithm~\ref{alg:subsampling} to construct sampling matrices $\{\mathbf{M}_c\}$ and bins $\{\mathbf{U}_c[j]\}$
    \While{we can find new single-ton bins}
    
    	\ForAll{$c\in C, j\in \mathbb{F}_2^b$}
			\State Detect bins $\mathbf{U}_c[j]$ by Algorithm~\ref{alg:bin_detect}
			\If{the bin is single-ton}
			\State record ($\alpha,{p}^{(2)}_{\alpha}$) in $\mathcal{P}$
			\ForAll{other groups $c'\neq c$}
			    \State find bins $\mathbf{U}_{c'}[j']$ contains this ${p}^{(2)}_{\alpha}$
			        \State subtract ${p}^{(2)}_{\alpha}$ from $\mathbf{U}_{c'}[j']$
			 \EndFor
			\ElsIf {the bin is not single-ton}
				 \State continue to next bin
			\EndIf	
			
		\EndFor
    \EndWhile
    \State Return: the stored indices and second-order Pauli error rates $\mathcal{P}$
  \end{algorithmic}
\end{algorithm}
\end{table}
We will further give an exhaustive illustration of this peeling process in Algorithm~\ref{alg:peeling} in Appendix~\ref{sec:errorRate}.

It is worth noting that since the procedure needs to peel out nonzero rates and generate new single-ton bins, it triggers noise propagation among bins during peeling.
In Appendix~\ref{sec:errorRate}, we analyze the behavior of the whole transformation with noisy fidelity terms gained from Algorithm~\ref{alg:query}.
We show that an index array $\mathbf{T}$ indicates the size of variances for noise in each bin.
Moreover, since the noise of fidelity terms from Algorithm~\ref{alg:query} is not necessarily unbiased, we also consider the bias effects in the noise propagation during the peeling process, which is more complicated than the ideal version of transform introduced in \cite{harper2020fast}.
Combining the noise analysis and the mechanisms of the sampling and peeling, as well as Proposition~\ref{lm:f2bound}, we can claim the following.
\begin{proposition}\label{prop:Pauli}
Suppose the Assumption A1 \& A2, \ref{assump:gtm}, and \ref{assump:Pauli} hold.
Execute Algorithm~\ref{alg:peeling} with $t_0$ satisfies $\|\mathcal{H}_{5t_0}-\mathcal{I}\|<\frac{1}{4}$ and $l=\frac{2}{\epsilon^4}\log\left(\frac{4ns|\kappa|}{\delta}\right)$ sequences with a set $\kappa$ of variant sequence lengths, $B=2^b=\max\left\{O(s),O\left(\frac{\epsilon^4}{t_0^4\epsilon_0^2}\right)\right\}$, $C=O(1)$, the unit time length $t_0$, and offsets $\mathbf{D}$ with $P=O(n)$ for each subsampling group.
The transformer will estimate all Pauli error rates $\hat{\mathbf{p}}$
with accurate support information and error bounds to be $\|\hat{\mathbf{p}}-\mathbf{p}\|_\infty\leq O\left(\frac{\epsilon^2}{t_0^2\sqrt{s}}\right)$. 
Therefore, the absolute values of these nonzero decomposition parameters can be estimated by $\|\abs{\hat{s}^\star}-\abs{s^\star}\|_\infty\leq O\left(\frac{\epsilon}{t_0\sqrt[4]{s}}\right)$
The estimation works successfully with probability at least $1-\delta-O\left(\frac{1}{s}\right)$.
\end{proposition}
\begin{proof}
The proof relies on several lemmas stated in Appendix~\ref{sec:errorRate}.
The proof starts by fetching the error bounds from Proposition~\ref{lm:f2bound} and analyzing the noise in every bin.
By proving Lemma~\ref{lm:prop_hashing_obs} and Lemma~\ref{lem:tailbound}, it shows that sampling bins serve as hash functions and that the bin detector works accurately with a given level of noise.
Based on these, the index $T$ records the noise size and helps the procedure to adjust the noise expectation.
Therefore, we can prove the success of this peeling process.
The formal proof can be viewed in Appendix~\ref{sec:bounds}
\end{proof}

In summary, the first stage can detect the second-order Pauli error rates of the Hamiltonian evolution channel or, equivalently, the squared decomposition parameters of the Hamiltonian.
The protocol first uses random gate twirling to detect the Pauli fidelity of the Hamiltonian channel, followed by regression to eliminate evolving times. 
It then constructs bins from several sets of second-order fidelity as the estimator of all nonzero second-orderPauli error rates.
The detailed protocol with formalized subroutines is illustrated in Appendix~\ref{Estimator}.
The validity and efficiency of estimation for the whole set of squared decomposition parameters ${\bf s}^\star$ is proven rigorously in Appendix~\ref{sec:bounds}.
We also discuss the possible benefits the estimation gets if some prior knowledge about the underlying structure of the Hamiltonian is available.
The knowledge might be collected from the understanding of the experimental devices or the properties of the tasks running on the devices.
We can expect a more accurate and reliable result with this kind of aid, which is illustrated in Appendix~\ref{sec:prior_info}.

\subsection{Stage 2.~Sign Estimation}\label{sec:sign}
After the first stage, we can detect all nonzero second-order Pauli error rates of the Hamiltonian channel.
Nevertheless, these terms only carry the square values of parameters as stated in Eq.~\eqref{eq:chi_wot}.
As a complement, the protocol for this stage is to estimate all the sign information.
The idea of our sign estimator is based on constructing equations as in quantum process tomography~\cite{PhysRevLett.90.193601,PhysRevA.77.032322,PhysRevLett.78.390} and utilizing the information gained from the first stage.


Suppose the procedure chooses $m$ states and measurement settings $\{(\rho_k,M_k)\}_{k=1}^m$.
As shown in circuit (c) of Figure~\ref{fig:peeling}, the procedure measures these states on the Pauli operators after evolving under the target Hamiltonian.
For each pair, the procedure collects multiple measurement results of the evolved state with different evolving times $\{t\}$.
By the regression over $\{t\}$, the procedure focuses on the first-order coefficients of the measurements. 
According to the expansion in Eq.~\eqref{eq:Hamiltonian}, the first-order measurement for a pair of SPAM settings is depicted by the \emph{process equation} as \begin{gather}\label{eq:linearordereq}
    \Tr(\mathcal{H}^{(1)}_t(\rho)M)=i\sum_{\alpha\in{\sf P}^n}s_\alpha \Tr(\rho [P_\alpha, M]),
\end{gather}
where $\mathcal{H}^{(1)}_t$ is the first-order expansion of the Hamiltonian evolution channel, and $[P_\alpha,M]$ is the commutator.
After the procedure computes all these first-order coefficients from measurements, we collect a series of observations for different SPAM settings $$\mathcal{M}=\left\{\Tr(\mathcal{H}_t^{(1)}(\rho_1)M_1),\cdots,\Tr(\mathcal{H}_t^{(1)}(\rho_m)M_m)\right\},$$ where each corresponds to a linear combination of trace results for commutators as in Eq.~\eqref{eq:linearordereq}.
This reduces the problem of estimating Hamiltonian parameters ${\bf s}$ to solving the following process equations,
\begin{gather}\label{eq:LEQ}
    \mathcal{M}=\Phi \cdot {\bf s},
\end{gather}
where $\Phi$ is the matrix consisting of $\{i \Tr(\rho_k [P_\alpha, M_k])\}_{k,\alpha}$.

Generally, the size of unknown parameters $\bf s$ scales exponentially with $n$, rendering the solution intractable.
In our execution, the sparsity assumption guarantees that there are only sparse nonzero variables to be solved.
Besides, from the first stage, our procedure has extracted the square values of decomposition parameters, which also includes the precise sparse support information.
The size of linear equations is reduced to $s$ by shrinking the target parameter from $\bf s$ to ${\bf s}^\star$ and reducing the coefficient matrix.
Therefore, we have the following equations to illustrate the polynomial-size problem,
\begin{gather}
    \mathcal{M}=\Phi \cdot {\bf s}^\star.
\end{gather}



The problem will be more knotty when SPAM errors are taken into account.
Suppose measurements of the evolved states for the chosen observable are corrupted by additive noise.
The noise will be transmitted in the linear regression stage, and the effects of noise will be enlarged due to the short evolving time $t$.
In this case, the first-order part of measurements $\mathcal{M}$ will also be corrupted by the noise as $\hat{\mathcal{M}}=\mathcal{M}+\omega$.
To suppress these effects, by choosing the coefficient matrix $\Phi$ with random local Pauli eigenstates $\rho$ and Pauli measurements $M$, we can construct a coefficient matrix which limits the noise propagation during solving the unknown parameters.
Regarding the noise and biased observation vector $\hat{\mathcal{M}}$, the original equations will be derived to a new optimization problem,
\begin{align}\label{eq:optimize}
\begin{split}
    \min_x \|x\|_1&\\
    ||\Phi x- \hat{\mathcal{M}}||_2&\leq \epsilon,
\end{split}
\end{align}
where $\epsilon$ is set to avoid the stuck case when the equation has no solution, and $x$ is an unknown vector with dimension $s$.
It is easy to see that when $\epsilon\geq\|\omega\|_2$, the true parameter ${\bf s}^\star$ lies in the feasible set.
The algorithm for sign estimation is shown in the following with a proper choice of $\epsilon\geq\|\omega\|_2$ according to Lemma~\ref{prop:signRIP}.
\begin{table}[thb]
\begin{algorithm}[H] 
  \caption{\label{alg:sign} Sign Estimation: \textbf{SE}($\{\abs{p_\alpha},\alpha\},m,t_1$)}
  \begin{algorithmic}[1]
    \State {\tt Input:} The nonzero error rates support of $\mathcal{P}$ from Algorithm~\ref{alg:peeling}, a positive integer $m$, and the unit time length $t_1$
    \State Prepare $m$ random local Pauli eigenstates and $m$ random Pauli operators for measurements $(\rho_{\gamma_i},P_{\beta_i})_{i=1}^m$
    \State Choose evolving times $\{i\cdot t_1\}$ which are multiples of $t_1$ and construct the Hamiltonian evolution channel $\{\mathcal{H}_t\}$ for evolving time
    \State Measure $\{\Tr(\mathcal{H}_t(\rho_\gamma)P_\beta)\}$ and calculate the coefficient matrix $\Phi$ for Pauli indices belong to $\mathcal{P}$
    \State Process ordinary linear regression on the measurement outcomes along the time $\{i\cdot t_1\}$ and get observations $\hat{\mathcal{M}}$
    \State Calculate $\sigma\coloneqq\frac{\sum\abs{i\cdot t_1-\bar{t}}}{\sum(i\cdot t_1-\bar{t})^2}\cdot t_1$ as the regression coefficient
    \State Choose $\epsilon=\sqrt{m}(\sigma\cdot\tau/t_1+o(t_1))$ for the optimization problem in Eq.~\eqref{eq:optimize}
    \State Solve the organized optimization problem with solution $x^\star$
    \State \Return solution $x^\star$
  \end{algorithmic}
\end{algorithm}
\end{table}

By sampling the input states and measurement operators, the $\Phi$ matrix possesses approximate restricted isometry property (RIP) as defined in Definition~\ref{de:rip}.
This property, by definition, makes sure that the norm of an arbitrary vector after the matrix product is close to the norm of the original vector.
We prove this RIP of the randomly sampled process matrix for all possible decomposition parameters with a high probability according to concentration inequalities in Lemma~\ref{prop:signRIP}, which we refer to as the approximate RIP.
Based on the definition of RIP, the noise of resulting nonzero parameters can be bounded by the chosen relaxation $\epsilon$.
Since the sign estimator aims at extracting the discrete $\pm1$ information, it tolerates a reasonable amount of noise that would not cause a sign flip in every solved parameter.
Therefore, the noise is not destructive for the estimator, given that we use multiple rounds of equations to minimize the statistical effects.
\begin{proposition}\label{prop:phase2}
Suppose the absolute value estimator perfectly returns the support information of decompose parameters of the Hamiltonian, and Assumption A1 \& A2 and \ref{assump:SPAM} hold.
Run Algorithm~\ref{alg:sign} with the support information of decomposition parameters and the unit time length $t_0$.
By setting $m$ to be $O(s)\leq m\leq O\left(\frac{t_1^2\epsilon_0}{\tau^2}\right) $ The solution of Algorithm~\ref{alg:sign} contains perfect sign information of all nonzero decomposition parameters with probability at least $1-\mathrm{e}^{-O(m-s)}$.
\end{proposition}
\begin{proof}
The full proof is shown in Appendix~\ref{sec:signestimate}.
\end{proof}
\begin{remark}
\rm As there exists an upper bound of the $m$ due to the accumulation of the SPAM errors, we cannot directly use larger $m$ to significantly suppress the sign errors. On the other hand, we can choose multiple blocks of equations to estimate signs for multiple times. Based on this, we can further reduce the failure probability by an additional majority vote mechanism when we hold multiple sets of signs.
\end{remark}

In this stage, we propose a protocol to extract the sign information of the Hamiltonian channel by constructing equations from random state preparation and measurements.
With the aid of absolute values from the former stage, the size of this problem is reduced to linear with the sparsity.
Based on the randomness, the equations will be robust against the measurement noise.
Therefore, this efficient and robust estimation fulfills the need for sign estimation.

\section{Main Theorem}\label{sec:mainresult}
Executing the two-stage protocol illustrated above, we can estimate the whole information of the nonzero decomposition parameters ${\bf s}^\star$, namely, we can recover the Hamiltonian operator faithfully from this estimation as in Algorithm~\ref{alg:main}.
\begin{table}[thb]
\begin{algorithm}[H] 
  \caption{Hamiltonian Estimator\label{alg:main}}
  \begin{algorithmic}[1]
    \State {\tt Input:} An oracle to control the unknown Hamiltonian $H$, positive integer $l$, $b$, $C$, $P$ and $m$, and the unit time lengths $t_0$ and $t_1$
    \State {\tt Input:} Offsets $\mathbf{D}$ for subsampling
    \State Call Algorithm~\ref{alg:peeling} with $l$, $b$, $C$, $P$, $t_0$ and offsets $\mathbf{D}$ to construct bin sets with random subsampling groups and estimate all the second-order error rates $\mathcal{P}$
    \State $\text{Supp}({\bf s}^\star)\leftarrow\{\alpha\,|\,p_\alpha\neq0\}$
    \State Call Algorithm~\ref{alg:sign} with $\text{Supp}({\bf s}^\star)$, $m$ and $t_1$ to estimate sign information $\sgn{{\bf s}^\star}$ nonzero parameters
    \State ${\bf s}^\star\leftarrow\{\sqrt{p_\alpha}\cdot\sgn{s_\alpha}\,|\,p_\alpha\in\mathcal{P}\}$
    \State \Return ${\bf s}^\star$
  \end{algorithmic}
\end{algorithm}
\end{table}

We combine the guarantee of our protocol from the guarantees of the two stages since our main protocol is a composite algorithm.
Our guarantee, therefore, does rely on some mild assumptions, which are either inherited from foregoing subroutines or from requirements to bond these stages in a self-consistent way.

We summarize the results in Theorem~\ref{thm:main}, which states that our algorithm can estimate the decomposition parameters of an unknown sparse Hamiltonian with high accuracy and vanishing failure probability.
It shows that our protocol is provably robust against circuit noise and SPAM errors under mild assumptions, which can be viewed from the guarantees of the fidelity estimator and the sign estimator separately.
This algorithm is also provably scalable for both quantum measurement complexity and classical computational complexity.
Moreover, the quantum complexity claimed above is for the worst case, and this complexity can be further reduced by merging all queries of qubitwise commutative Pauli fidelity, which relies on the choices of stabilizer groups elaborated in Appendix~\ref{sec:fidelity}.
This merging can usually reduce the needed rounds of fidelity detectors by one or two orders of magnitude.
Therefore, our protocol provides an implementable method for Hamiltonian learning.
\begin{theorem}\label{thm:main} 
Suppose Assumption A1 \& A2, \ref{assump:SPAM}, \ref{assump:gtm}, and \ref{assump:Pauli} hold. Run Algorithm~\ref{alg:main} with $l=\frac{2}{\epsilon^4}\log\left(\frac{4ns|\kappa|}{\delta}\right)$ sequences for each length where $\kappa$ is the set of variant sequence lengths, $B=2^b=\max\left\{O(s),O\left(\frac{\epsilon^4}{t_0^4\epsilon_0^2}\right)\right\}$, $C=O(1)$, $m=O(s)$, unit times $t_0$ satisfies $\|\mathcal{H}_{5t_0}-\mathcal{I}\|<\frac{1}{4}$ and  and $t_1$, and the offsets $D$ with size $P=O(n)$ for each subsampling group.
The Hamiltonian estimator will return all nonzero decomposition parameters with the perfect support estimation and $\|\hat{\bf s}^\star-{\bf s}^\star\|_\infty\leq O\left(\frac{\epsilon}{t_0\sqrt[4]{s}}\right)$, which succeeds with probability at least $1-\delta-O\left(\frac{1}{s}\right)$.
The circuit measurement complexity of this execution is $\tilde{O}\left(\frac{sn}{\epsilon^4}\right)$ where $\tilde{O}$ notation ignores the logarithmic terms.
The post-processing complexity is $\text{poly}(n,s)$, where poly denotes that the scaling is polynomial with the elements.
\end{theorem}
\begin{proof}
This statement can be regarded as a  combination of Proposition~\ref{lm:f2bound},~\ref{prop:Pauli}, and~\ref{prop:phase2}.
During the combination, we need first to consider the overall precision resulting from subsequently invoking subroutines.
According to the target overall precision, we shall calculate the precision for each subroutine and then count and sum up the required complexities for all these subroutines.
The rigorous counting and derivation are given in Appendix~\ref{sec:mainproof}.
\end{proof}
\begin{remark}
It is true that we can choose other fitting models and strategies for extracting the corresponding information in both stages, and we leave a detailed study of the strategy for future studies.
In the execution of the algorithm, we need two inputs, $t_0$ and $t_1$, to serve as the fitting times of stages 1 and 2, respectively. 
Nevertheless, there is a natural trade-off toward the size of those evolving times $t_0$ and $t_1$.
If the evolving times are too short, this could suppress systematic errors from high-order terms in the expansion, while it would enlarge the errors from imperfect observations and measurements by fitting.
If the times are too long, this will lead to large systematic errors.
Therefore, the balanced choice of $t_0$ and $t_1$ can be viewed in Remark~\ref{rm: pha1t} and \ref{rm:pha2t}, respectively.
To further improve the performance, a possible modification is to consider those higher-order terms in the fitting model to mitigate the systematic errors.
In appendix~\ref{sec:fitting}, we exhibit different trials of executions by fitting higher-order terms and taking the desired coefficients.
This can vastly reduce the amplitude of the systematic errors and allow us to implement a longer time evolution.
Moreover, the idea from~\cite{yuan2016simulating} can help to construct an exponential suppression of the high-order terms with the number of data points.
In general, we expect the choices of evolving times to satisfy that the effects of fluctuation errors from observations dominate.
\end{remark}

\section{Numerical Results}\label{sec:numerical}
In this section, we exhibit several numerical results to show the validity and accuracy of the protocol for learning an unknown Hamiltonian with sparse decomposition parameters on the Pauli basis.
These numerical results work as strong evidence that verifies the foregoing intuitively illustrated protocol.
They also reveal some critical properties of the Hamiltonian learning method in the realistic implementation as we discuss shortly.

In order to simplify the numerical simulation without modifying the core components of the protocol, we make some mild assumptions in the simulation, which would not weaken the results arguably. 
Among all of these simulations, instead of running the cascading circuits to extract the fidelity information as in \Cref{sec:Pauli}, we assume there exists an oracle to return measurement outcomes of input states under the evolution with the input time of the unknown Hamiltonian, which helps to directly get the Pauli fidelity terms of evolution.
In this sense, we count each desired fidelity term as one query, and the number of calls for the oracle represents the total measurements.
This is based on the understanding that the complexity of the first stage is always dominant in the total number of circuit measurements, which makes the tracking of fidelity query a faithful indicator of the total quantum complexity.

As for the simulation of noise effects, we add zero-mean Gaussian noise $\mathcal{N}(0,10^{-3})$ to the queried Hamiltonian fidelity terms to imitate the fluctuation noise from the fidelity estimator.
Along the fitting process and the sparse transformation, the noise effects would be exhibited in the absolute value estimation.
Moreover, we add the Gaussian noise $\mathcal{N}(0,10^{-3})$ to the measurement outcomes in the sign estimation as a representation of the SPAM effects when constructing the process equations.
These errors will be propagated by the fitting and the optimization stated in stage 2, and the resulting signs reflect overall noise effects.
This makes it possible to observe the noise-resilience of our learning protocol from the accuracy of each part in the simulation.\\

\begin{figure}[t]
    \centering
    \includegraphics[width=\columnwidth]{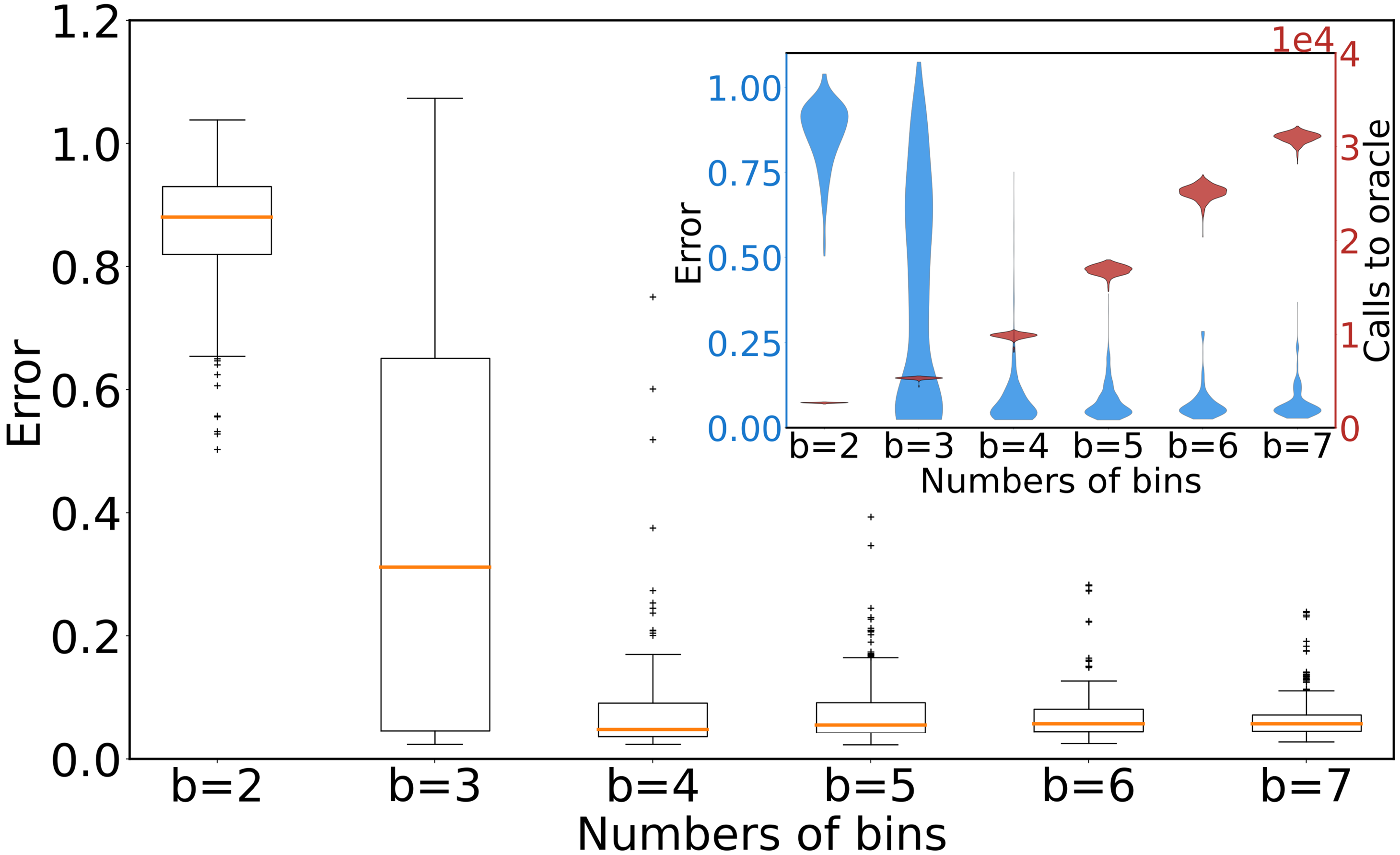}
    \caption{Estimation of the 6-qubit random transverse field Ising model with different choices of the bin numbers, $b$. In the box plot, we exhibit the statistical properties of errors from the outcomes using different $b$.
    In each case, the data is gained from the estimation of 50 random TFIM Hamiltonians with 10 rounds each.
    For every box, the orange line denotes the median of the error distribution, while the upper and lower bounds represent the $75\%$ and $25\%$, respectively.
    By increasing $b$, the procedure estimates all the decomposition parameters with smaller errors. A sharp vanishing over the distribution of the estimation error is witnessed when $b$ grows up to $4$.
    In the inset, we use the violin plot to show a complete view of the distribution of the query numbers and recovery errors.
    Every violin along its axis represents a distribution of the corresponding value, and the widths everywhere illustrate the probability density.
    There is also a sharp shrink of the error bar while the querying number is growing along with the increasing $b$.}
    \label{fig:threshold}
\end{figure}

\begin{figure*}[t]
    \centering
    \includegraphics[width=2\columnwidth]{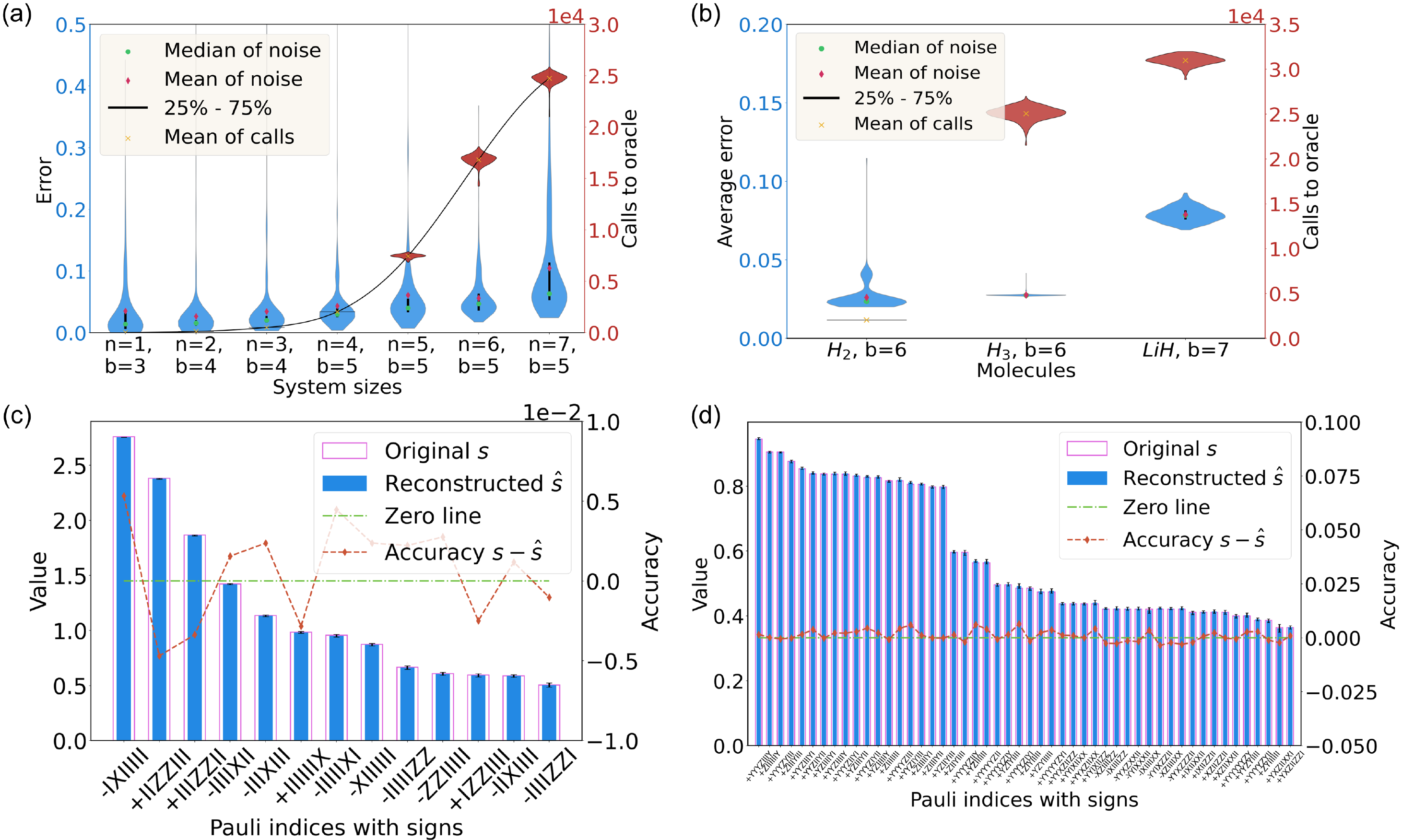}
    \caption{Numerical results for the transverse field Ising model and molecular Hamiltonians. In (a), we show the distribution of the results for random TFIM Hamiltonian with different system sizes varying from 1 to 7.
    For each $n$, the procedure runs for 50 random TFIM Hamiltonian operators with 10 rounds for each to detect the average reconstruction results.
    The parameter $b$ is chosen by witnessing the threshold behaviors by varying $b$ as a trial process. For example, we choose $b=5$ for 6-qubit random TFIM Hamiltonian learning, and all others can be viewed in the appendix.
    In (b), we show the ability to recover different molecular Hamiltonians, H$_2$ (4 qubits), H$_3$ (6 qubits), and LiH (6 qubits). 
    For each molecular system, we execute our method for 500 independent rounds to gather these statistical properties.
    The bin parameters are also chosen according to the threshold behaviors. 
    In (c), we show an average reconstruction among 20 rounds of the 7-qubit TFIM Hamiltonian, and we use the error bars to display the standard deviations of parameters. The terms are sorted by the absolute values of decomposition parameters.
    The labels of the $x$-axis indicate both the Pauli indices as well as the signs of the estimation (the signs are all corrected estimated, as one can see from (a)).
    In (d), the procedure estimates the Hamiltonian for the H$_4$~(8 qubits) molecule, and we exhibit the 50 dominant terms of both ideal and reconstructed parameters among 184 nonzero terms. 
    This result is an average of 8 rounds, and the error bars are constructed according to standard deviations.
    In each round, the simulation uses 5 blocks to vote for the final estimation of signs to suppress the failure probability of sign flips.
    Similar to (c), this plot also uses $x$-axis coordinates to denote the Pauli indices and signs, and the estimation shows great agreement with the ideal parameters.
    }
    \label{fig:numerical}
\end{figure*}

First, we study how different choices of the number of bins $b$ affect the efficacy of the protocol.
We consider the 6-qubit transverse field Ising model (TFIM) Hamiltonian with random interaction parameters,
\begin{gather}\label{eq:randtfim}
    H=\sum_{i=1}^{n-1}\alpha_iZ_iZ_{i+1}+\sum_{j=1}^n\beta_jX_j.
\end{gather}
For different choices of $b$, we run the protocol for 50 randomly chosen Hamiltonian operators, where each of them is estimated for 10 independent rounds. 
As defined in \Cref{sec:Pauli}, the parameter $b$ reflects the number of bin sets in every subsampling group.
With a larger $b$, the procedure hashes Pauli error rates more dispersively, meaning that every bin contains fewer underlying Pauli error rates as well as nonzero error rates.
Since the peeling process relies on the sparsity of nonzero terms in every bin, procedures with larger $b$ would generate more accurate estimations, where the reconstruction error is depicted by the \emph{relative 1-norm distance},
\begin{gather}\label{Eq:relativeerror}
    e_1\coloneqq\frac{\|s-\hat{s}\|_1}{\|s\|_1}.
\end{gather}
On the other hand, a larger $b$ is equivalent to querying more fidelity terms in the subsampling step, which leads to a trade-off of choosing a proper $b$.
In Figure~\ref{fig:threshold}, we can view in the box-plot that the recovery errors decay vastly when the procedure increases $b$ from $3$ to 4. 
The error is even smaller when it continues to increase $b$.
However, as shown in the inset of Figure~\ref{fig:threshold}, the execution overhead also increases with a $b$. 
Therefore, a wise choice of $b$ can be explored as the threshold under which the error is significantly large, where we choose it as $b=5$ in our following simulation.
This remains consistent with the previous analysis about the peeling step, where the procedure expects the bin number to be larger than the sparsity. 

However, we may not have access to the error in practice for an unknown quantum system. 
The problem then rises as to how to decide a proper $b$ in practice.
That relies on adaptively increasing $b$ from $b_0$. to choose the proper bin size $b$. 
For every $b$, we run the protocol multiple times to first observe the peeling process to find out whether some multi-ton bins get stuck after the peeling process.
If this happens with significant frequencies, it indicates that the current $b$ is not enough.
We can further observe the convergence of results from different rounds.
As shown in Figures~\ref{fig:threshold} and~\ref{fig:ising}, the variance between different rounds of estimation will decrease sharply when $b$ is large enough i.e., $b\geq\lceil\log s\rceil$, and we denote the first $b$ to observe this decaying as the threshold.
Since the execution complexity scales as $O(n2^b/\epsilon^4)$ when we choose $b$ arbitrarily, the overall adaptive trial of $b$ to $b=\lceil\log s\rceil$ only increases the complexity $O(ns/\epsilon^4)$ by a constant factor.

Next, we study the scaling of the protocol for different Hamiltonians with varying system sizes. 
We first consider the TFIM with different numbers of qubits. 
As in Figure~\ref{fig:numerical}(a), we show a complete view of the distribution of outcomes in different $n$. 
We choose $b$ for each $n$ according to the abovementioned threshold idea.
The errors are also measured by the relative 1-norm distance, and the major parts concentrate on small errors in each case.
While an interesting observation is that the lower bound of these errors grows with the system size, this is mainly due to the regression error as we execute the ordinary least-square linear regression in this simulation.
Even though we only consider the second-order effects here, We further analyze the effects of different fitting settings in the following and in Appendix~\ref{sec:fitting}, which suggests that considering the higher even-order terms in the fitting can significantly improve the estimation.

The TFIM model still has a simple structure with nearest-neighbour 1D interactions. 
We then consider more complicated examples for benchmark, i.e., molecular Hamiltonians, which have been demonstrated to be potential candidates for quantum computing and quantum simulation~\cite{troiani2005molecular,santini2011molecular}.
We encode the fermionic Hamiltonian with qubits, which generally has much more complicated Pauli decomposition parameters that are not local but have sophisticated multiqubit interaction. 
We note that though the nonzero terms are sparse compared to the exponential size of Pauli decomposition, the nonzero support is still quite large, which grows quartically with the system size. If we could measure each Pauli decomposition parameter to accuracy $\varepsilon$, the reconstruction error defined in Eq.~\eqref{Eq:relativeerror} also grows quartically. In order to show a system-size independent error rate,  we recruit the \emph{average 1-norm distance},
\begin{gather}
    e_a\coloneqq\frac{\|s-\hat{s}\|_1}{\|s\|_0},
\end{gather}
to depict the reconstruct errors in the molecular cases.
In Figure~\ref{fig:numerical}(d), we display the distribution of the estimation of several molecular Hamiltonians only considering the zero- and second-order terms during fitting.
Generally, both the quantum complexity increase when system sizes grow, while the average reconstruction errors remain comparably constant.

In the remaining two subplots, we use the bar plot to exhibit the estimation of every parameter.
In Figure~\ref{fig:numerical}(c), we exhibit the average estimation of 50 independent rounds of executions for a 7-qubit TFIM Hamiltonian with random parameters.
Furthermore, we use the higher-order fitting considering the fourth-order terms of the Hamiltonian expansion to extract a much more precise estimation.
The higher-order fitting will require an additional one or two times query complexity, while it brings a greatly precise estimation.
Therefore, as can be viewed in (c), the distances between pairs of original values and reconstructed values are generally close to zero.
In Figure~\ref{fig:numerical}(d), we estimate the Hamiltonian of the hydrogen chain $H_4$.
We show the 50 dominant terms with both the ideal and reconstructed decomposition parameters \footnote{There are $184$ nonzero terms in the $H_4$ Hamiltonian in total.}.
Since the number of nonzero terms is great and they vary from large terms to very small terms, we also consider the fourth-order terms' effects during fitting to improve the estimation of absolute values.
In order to provide the perfect sign estimation, we add a further majority voting from 5 blocks of resulting signs when deciding the final signs.
In this average estimation of 8 independent rounds, random noise causes little harm to the reconstruction, and our procedure returns precise estimations of the most significant terms among the nonzero parameters.
On the other hand, the error may grow when the underlying value turns small since the procedure cannot divide the noise effects from the small parameters accurately. 
This result sheds light on how the algorithm would execute on Hamiltonian with much more nonzero decomposition parameters.

We refer to Appendix~\ref{sec:addnum} for more discussions on the implementation and numerical results. The codes to generate these simulations are shared in~\cite{GITHUB}.

\section{Conclusion}\label{sec:con}
In this work, we have proposed a universal, robust, and efficient Hamiltonian learning protocol. The protocol is universal in the sense that it estimates any $n$-qubit Hamiltonians with sparse Pauli decomposition parameters without prior knowledge of the Hamiltonian structure. For example, the protocol can estimate Hamiltonians consisting of global interactions, or namely those high-weight Pauli terms. 
Next, our protocol is robust since it estimates the Hamiltonian noise-resiliently against SPAM errors and certain amounts of the circuit and shot noise.
The protocol is also efficient with polynomially scaling classical and quantum complexities.
Besides, the protocol directly exploits the time evolution of the Hamiltonian instead of requiring eigenstates or thermal states of the system, which make it reliable in the real-world implementation. 
All these results have been rigorously proven in our theoretical analysis in Appendix as well as numerically verified for different sizes of transverse field Ising models and different electronic Hamiltonians in the main text.

Our work provides a scalable and more practically-minded Hamiltonian learning method compared to the previous protocols based either on explicit structures or informative states.
This protocol could be applied for calibrating quantum computing hardware, studying the interaction structures of intricate many-body quantum systems, detecting interesting many-body physics phenomena, etc.
As our method exhibits a direct linkage between the fidelity estimation and Hamiltonian learning, we can expect that development in the fidelity estimation also improve the learning of Hamiltonian.
For example, the proposal in~\cite{chen2021quantum} implies that entanglement states can help to reduce the overheads for fidelity estimation, which can also be employed in the Hamiltonian learning. 
Our method only exploits the incoherent part of the Hamiltonian time evolution process to estimate the absolute value of the decomposition parameters, and one may expect to make use of the coherent part (the first-order expansion) to more efficiently extract the Hamiltonian information. Nevertheless, as we have discussed previously, this would require a new SPAM-robust process detection mechanism, which also deserves its own investigation. 
Another point full of practical interest is to combine the fitting method that holds a clear trade-off between the higher-order effects and the random noise with our protocol. 
As inspired by Figure~\ref{fig:numerical} and \ref{fig:varm}, a wise choice of fitting method can significantly improve the estimation, while we still need to balance the effects of the overall errors.
Moreover, a rigorous analysis of more general non-sparse Hamiltonian learning remains open. 
we are interested in the case when the parameters of the unknown Hamiltonian distribute unevenly rather than sparsely.
As a rough idea, compressed sensing might help in this circumstance.


\section*{Acknowledgement}
This work is supported by the National Natural Science Foundation of China Grant No.~12175003.
The numerics is supported by High-performance Computing Platform of Peking University.

\bibliographystyle{unsrtnat}
\bibliography{mybibli}

\begin{thebibliography}{94}
\providecommand{\natexlab}[1]{#1}
\providecommand{\url}[1]{\texttt{#1}}
\expandafter\ifx\csname urlstyle\endcsname\relax
  \providecommand{\doi}[1]{doi: #1}\else
  \providecommand{\doi}{doi: \begingroup \urlstyle{rm}\Url}\fi

\bibitem[Arute et~al.(2019)Arute, Arya, Babbush, Bacon, Bardin, Barends,
  Biswas, Boixo, Brandao, Buell, Burkett, Chen, Chen, Chiaro, Collins,
  Courtney, Dunsworth, Farhi, Foxen, Fowler, Gidney, Giustina, Graff, Guerin,
  Habegger, Harrigan, Hartmann, Ho, Hoffmann, Huang, Humble, Isakov, Jeffrey,
  Jiang, Kafri, Kechedzhi, Kelly, Klimov, Knysh, Korotkov, Kostritsa, Landhuis,
  Lindmark, Lucero, Lyakh, Mandr{\`{a}}, McClean, McEwen, Megrant, Mi,
  Michielsen, Mohseni, Mutus, Naaman, Neeley, Neill, Niu, Ostby, Petukhov,
  Platt, Quintana, Rieffel, Roushan, Rubin, Sank, Satzinger, Smelyanskiy, Sung,
  Trevithick, Vainsencher, Villalonga, White, Yao, Yeh, Zalcman, Neven, and
  Martinis]{Arute2019}
Frank Arute, Kunal Arya, Ryan Babbush, Dave Bacon, Joseph~C. Bardin, Rami
  Barends, Rupak Biswas, Sergio Boixo, Fernando G. S.~L. Brandao, David~A.
  Buell, Brian Burkett, Yu~Chen, Zijun Chen, Ben Chiaro, Roberto Collins,
  William Courtney, Andrew Dunsworth, Edward Farhi, Brooks Foxen, Austin
  Fowler, Craig Gidney, Marissa Giustina, Rob Graff, Keith Guerin, Steve
  Habegger, Matthew~P. Harrigan, Michael~J. Hartmann, Alan Ho, Markus Hoffmann,
  Trent Huang, Travis~S. Humble, Sergei~V. Isakov, Evan Jeffrey, Zhang Jiang,
  Dvir Kafri, Kostyantyn Kechedzhi, Julian Kelly, Paul~V. Klimov, Sergey Knysh,
  Alexander Korotkov, Fedor Kostritsa, David Landhuis, Mike Lindmark, Erik
  Lucero, Dmitry Lyakh, Salvatore Mandr{\`{a}}, Jarrod~R. McClean, Matthew
  McEwen, Anthony Megrant, Xiao Mi, Kristel Michielsen, Masoud Mohseni, Josh
  Mutus, Ofer Naaman, Matthew Neeley, Charles Neill, Murphy~Yuezhen Niu, Eric
  Ostby, Andre Petukhov, John~C. Platt, Chris Quintana, Eleanor~G. Rieffel,
  Pedram Roushan, Nicholas~C. Rubin, Daniel Sank, Kevin~J. Satzinger, Vadim
  Smelyanskiy, Kevin~J. Sung, Matthew~D. Trevithick, Amit Vainsencher, Benjamin
  Villalonga, Theodore White, Z.~Jamie Yao, Ping Yeh, Adam Zalcman, Hartmut
  Neven, and John~M. Martinis.
\newblock Quantum supremacy using a programmable superconducting processor.
\newblock \emph{Nature}, 574\penalty0 (7779):\penalty0 505--510, October 2019.
\newblock \doi{10.1038/s41586-019-1666-5}.

\bibitem[Neill et~al.(2021)Neill, McCourt, Mi, Jiang, Niu, Mruczkiewicz,
  Aleiner, Arute, Arya, Atalaya, et~al.]{neill2021accurately}
C~Neill, T~McCourt, X~Mi, Z~Jiang, MY~Niu, W~Mruczkiewicz, I~Aleiner, F~Arute,
  K~Arya, J~Atalaya, et~al.
\newblock Accurately computing the electronic properties of a quantum ring.
\newblock \emph{Nature}, 594\penalty0 (7864):\penalty0 508--512, 2021.
\newblock \doi{10.1038/s41586-021-03576-2}.

\bibitem[Ebadi et~al.(2021)Ebadi, Wang, Levine, Keesling, Semeghini, Omran,
  Bluvstein, Samajdar, Pichler, Ho, et~al.]{ebadi2021quantum}
Sepehr Ebadi, Tout~T Wang, Harry Levine, Alexander Keesling, Giulia Semeghini,
  Ahmed Omran, Dolev Bluvstein, Rhine Samajdar, Hannes Pichler, Wen~Wei Ho,
  et~al.
\newblock Quantum phases of matter on a 256-atom programmable quantum
  simulator.
\newblock \emph{Nature}, 595\penalty0 (7866):\penalty0 227--232, 2021.
\newblock \doi{10.1038/s41586-021-03582-4}.

\bibitem[Wu et~al.(2021)Wu, Bao, Cao, Chen, Chen, Chen, Chung, Deng, Du, Fan,
  et~al.]{wu2021strong}
Yulin Wu, Wan-Su Bao, Sirui Cao, Fusheng Chen, Ming-Cheng Chen, Xiawei Chen,
  Tung-Hsun Chung, Hui Deng, Yajie Du, Daojin Fan, et~al.
\newblock Strong quantum computational advantage using a superconducting
  quantum processor.
\newblock \emph{Physical review letters}, 127\penalty0 (18):\penalty0 180501,
  2021.
\newblock \doi{10.1103/PhysRevLett.127.180501}.

\bibitem[Gong et~al.(2021)Gong, Wang, Zha, Chen, Huang, Wu, Zhu, Zhao, Li, Guo,
  et~al.]{gong2021quantum}
Ming Gong, Shiyu Wang, Chen Zha, Ming-Cheng Chen, He-Liang Huang, Yulin Wu,
  Qingling Zhu, Youwei Zhao, Shaowei Li, Shaojun Guo, et~al.
\newblock Quantum walks on a programmable two-dimensional 62-qubit
  superconducting processor.
\newblock \emph{Science}, 372\penalty0 (6545):\penalty0 948--952, 2021.
\newblock \doi{10.1126/science.abg7812}.

\bibitem[Zhong et~al.(2020)Zhong, Wang, Deng, Chen, Peng, Luo, Qin, Wu, Ding,
  Hu, et~al.]{zhong2020quantum}
Han-Sen Zhong, Hui Wang, Yu-Hao Deng, Ming-Cheng Chen, Li-Chao Peng, Yi-Han
  Luo, Jian Qin, Dian Wu, Xing Ding, Yi~Hu, et~al.
\newblock Quantum computational advantage using photons.
\newblock \emph{Science}, 370\penalty0 (6523):\penalty0 1460--1463, 2020.
\newblock \doi{10.1126/science.abe877}.

\bibitem[Mi et~al.(2021)Mi, Ippoliti, Quintana, Greene, Chen, Gross, Arute,
  Arya, Atalaya, Babbush, et~al.]{mi2021time}
Xiao Mi, Matteo Ippoliti, Chris Quintana, Ami Greene, Zijun Chen, Jonathan
  Gross, Frank Arute, Kunal Arya, Juan Atalaya, Ryan Babbush, et~al.
\newblock Time-crystalline eigenstate order on a quantum processor.
\newblock \emph{Nature}, pages 1--1, 2021.
\newblock \doi{10.1038/s41586-021-04257-w}.

\bibitem[Kandala et~al.(2017)Kandala, Mezzacapo, Temme, Takita, Brink, Chow,
  and Gambetta]{kandala2017hardware}
Abhinav Kandala, Antonio Mezzacapo, Kristan Temme, Maika Takita, Markus Brink,
  Jerry~M Chow, and Jay~M Gambetta.
\newblock Hardware-efficient variational quantum eigensolver for small
  molecules and quantum magnets.
\newblock \emph{Nature}, 549\penalty0 (7671):\penalty0 242--246, 2017.
\newblock \doi{10.1038/nature23879}.

\bibitem[Kjaergaard et~al.(2020)Kjaergaard, Schwartz, Braumüller, Krantz,
  Wang, Gustavsson, and Oliver]{kjaergaard2019superconducting}
Morten Kjaergaard, Mollie~E. Schwartz, Jochen Braumüller, Philip Krantz, Joel
  I.-J. Wang, Simon Gustavsson, and William~D. Oliver.
\newblock Superconducting qubits: Current state of play.
\newblock \emph{Annual Review of Condensed Matter Physics}, 11\penalty0
  (1):\penalty0 369--395, 2020.
\newblock \doi{10.1146/annurev-conmatphys-031119-050605}.

\bibitem[Zhang et~al.(2017)Zhang, Pagano, Hess, Kyprianidis, Becker, Kaplan,
  Gorshkov, Gong, and Monroe]{zhang2017observation}
Jiehang Zhang, Guido Pagano, Paul~W Hess, Antonis Kyprianidis, Patrick Becker,
  Harvey Kaplan, Alexey~V Gorshkov, Z-X Gong, and Christopher Monroe.
\newblock Observation of a many-body dynamical phase transition with a 53-qubit
  quantum simulator.
\newblock \emph{Nature}, 551\penalty0 (7682):\penalty0 601, 2017.
\newblock \doi{10.1038/nature24654}.

\bibitem[Georgescu et~al.(2014)Georgescu, Ashhab, and
  Nori]{georgescu2014quantum}
Iulia~M Georgescu, Sahel Ashhab, and Franco Nori.
\newblock Quantum simulation.
\newblock \emph{Reviews of Modern Physics}, 86\penalty0 (1):\penalty0 153,
  2014.
\newblock \doi{10.1103/RevModPhys.86.153}.

\bibitem[Altman et~al.(2021)Altman, Brown, Carleo, Carr, Demler, Chin, DeMarco,
  Economou, Eriksson, Fu, et~al.]{altman2021quantum}
Ehud Altman, Kenneth~R Brown, Giuseppe Carleo, Lincoln~D Carr, Eugene Demler,
  Cheng Chin, Brian DeMarco, Sophia~E Economou, Mark~A Eriksson, Kai-Mei~C Fu,
  et~al.
\newblock Quantum simulators: Architectures and opportunities.
\newblock \emph{PRX Quantum}, 2\penalty0 (1):\penalty0 017003, Feb 2021.
\newblock \doi{10.1103/PRXQuantum.2.017003}.

\bibitem[Hauke et~al.(2012)Hauke, Cucchietti, Tagliacozzo, Deutsch, and
  Lewenstein]{hauke2012can}
Philipp Hauke, Fernando~M Cucchietti, Luca Tagliacozzo, Ivan Deutsch, and
  Maciej Lewenstein.
\newblock Can one trust quantum simulators?
\newblock \emph{Reports on Progress in Physics}, 75\penalty0 (8):\penalty0
  082401, 2012.
\newblock \doi{10.1088/0034-4885/75/8/082401}.

\bibitem[Granade et~al.(2012)Granade, Ferrie, Wiebe, and
  Cory]{granade2012robust}
Christopher~E Granade, Christopher Ferrie, Nathan Wiebe, and David~G Cory.
\newblock Robust online hamiltonian learning.
\newblock \emph{New Journal of Physics}, 14\penalty0 (10):\penalty0 103013,
  2012.
\newblock \doi{10.1088/1367-2630/14/10/103013}.

\bibitem[Wiebe et~al.(2014{\natexlab{a}})Wiebe, Granade, Ferrie, and
  Cory]{wiebe2014quantum}
Nathan Wiebe, Christopher Granade, Christopher Ferrie, and David Cory.
\newblock Quantum hamiltonian learning using imperfect quantum resources.
\newblock \emph{Physical Review A}, 89\penalty0 (4):\penalty0 042314,
  2014{\natexlab{a}}.
\newblock \doi{10.1103/PhysRevA.89.042314}.

\bibitem[Wiebe et~al.(2014{\natexlab{b}})Wiebe, Granade, Ferrie, and
  Cory]{wiebe2014hamiltonian}
Nathan Wiebe, Christopher Granade, Christopher Ferrie, and David~G Cory.
\newblock Hamiltonian learning and certification using quantum resources.
\newblock \emph{Physical review letters}, 112\penalty0 (19):\penalty0 190501,
  2014{\natexlab{b}}.
\newblock \doi{10.1103/PhysRevLett.112.190501}.

\bibitem[Krastanov et~al.(2019)Krastanov, Zhou, Flammia, and
  Jiang]{krastanov2019stochastic}
Stefan Krastanov, Sisi Zhou, Steven~T Flammia, and Liang Jiang.
\newblock Stochastic estimation of dynamical variables.
\newblock \emph{Quantum Science and Technology}, 4\penalty0 (3):\penalty0
  035003, 2019.
\newblock \doi{10.1088/2058-9565/ab18d5}.

\bibitem[Wang et~al.(2017)Wang, Paesani, Santagati, Knauer, Gentile, Wiebe,
  Petruzzella, O’Brien, Rarity, Laing, et~al.]{wang2017experimental}
Jianwei Wang, Stefano Paesani, Raffaele Santagati, Sebastian Knauer, Antonio~A
  Gentile, Nathan Wiebe, Maurangelo Petruzzella, Jeremy~L O’Brien, John~G
  Rarity, Anthony Laing, et~al.
\newblock Experimental quantum hamiltonian learning.
\newblock \emph{Nature Physics}, 13\penalty0 (6):\penalty0 551--555, 2017.
\newblock \doi{https://doi.org/10.1038/nphys4074}.

\bibitem[O'Brien et~al.(2021)O'Brien, Ioffe, Su, Fushman, Neven, Babbush, and
  Smelyanskiy]{o2021quantum}
Thomas~E O'Brien, Lev~B Ioffe, Yuan Su, David Fushman, Hartmut Neven, Ryan
  Babbush, and Vadim Smelyanskiy.
\newblock Quantum computation of molecular structure using data from
  challenging-to-classically-simulate nuclear magnetic resonance experiments.
\newblock \emph{arXiv preprint arXiv:2109.02163}, 2021.
\newblock \doi{10.48550/arXiv.2109.02163}.

\bibitem[Haah et~al.(2021)Haah, Kothari, and Tang]{haah2021optimal}
Jeongwan Haah, Robin Kothari, and Ewin Tang.
\newblock Optimal learning of quantum hamiltonians from high-temperature gibbs
  states.
\newblock \emph{arXiv preprint arXiv:2108.04842}, 2021.
\newblock \doi{10.48550/arXiv.2108.04842}.

\bibitem[Lloyd(1996)]{lloyd1996universal}
Seth Lloyd.
\newblock Universal quantum simulators.
\newblock \emph{Science}, pages 1073--1078, 1996.
\newblock \doi{10.1126/science.273.5278.1073}.

\bibitem[Das and Chakrabarti(2008)]{das2008colloquium}
Arnab Das and Bikas~K Chakrabarti.
\newblock Colloquium: Quantum annealing and analog quantum computation.
\newblock \emph{Reviews of Modern Physics}, 80\penalty0 (3):\penalty0 1061,
  2008.
\newblock \doi{10.1103/RevModPhys.80.1061}.

\bibitem[Friedenauer et~al.(2008)Friedenauer, Schmitz, Glueckert, Porras, and
  Sch{\"a}tz]{friedenauer2008simulating}
Axel Friedenauer, Hector Schmitz, Jan~Tibor Glueckert, Diego Porras, and Tobias
  Sch{\"a}tz.
\newblock Simulating a quantum magnet with trapped ions.
\newblock \emph{Nature Physics}, 4\penalty0 (10):\penalty0 757--761, 2008.
\newblock \doi{https://doi.org/10.1038/nphys1032}.

\bibitem[Aspuru-Guzik and Walther(2012)]{aspuru2012photonic}
Al{\'a}n Aspuru-Guzik and Philip Walther.
\newblock Photonic quantum simulators.
\newblock \emph{Nature physics}, 8\penalty0 (4):\penalty0 285--291, 2012.
\newblock \doi{10.1038/nphys2253}.

\bibitem[Bernstein and Vazirani(1997)]{bernstein1997quantum}
Ethan Bernstein and Umesh Vazirani.
\newblock Quantum complexity theory.
\newblock \emph{SIAM Journal on computing}, 26\penalty0 (5):\penalty0
  1411--1473, 1997.
\newblock \doi{10.1137/S0097539796300921}.

\bibitem[Shor(1995)]{Shor1995}
Peter~W. Shor.
\newblock Scheme for reducing decoherence in quantum computer memory.
\newblock \emph{Phys. Rev. A}, 52\penalty0 (4):\penalty0 R2493--R2496, Oct
  1995.
\newblock \doi{10.1103/PhysRevA.52.R2493}.

\bibitem[Grover(1996)]{grover1996fast}
Lov~K Grover.
\newblock A fast quantum mechanical algorithm for database search.
\newblock In \emph{Proceedings of the twenty-eighth annual ACM symposium on
  Theory of computing}, pages 212--219, 1996.
\newblock \doi{10.1145/237814.237866}.

\bibitem[Harrow et~al.(2009)Harrow, Hassidim, and Lloyd]{harrow2009quantum}
Aram~W Harrow, Avinatan Hassidim, and Seth Lloyd.
\newblock Quantum algorithm for linear systems of equations.
\newblock \emph{Physical review letters}, 103\penalty0 (15):\penalty0 150502,
  2009.
\newblock \doi{10.1103/PhysRevLett.103.150502}.

\bibitem[Krantz et~al.(2019)Krantz, Kjaergaard, Yan, Orlando, Gustavsson, and
  Oliver]{krantz2019quantum}
Philip Krantz, Morten Kjaergaard, Fei Yan, Terry~P Orlando, Simon Gustavsson,
  and William~D Oliver.
\newblock A quantum engineer's guide to superconducting qubits.
\newblock \emph{Applied Physics Reviews}, 6\penalty0 (2):\penalty0 021318,
  2019.
\newblock \doi{10.1063/1.5089550}.

\bibitem[Trabesinger(2012)]{trabesinger2012quantum}
Andreas Trabesinger.
\newblock Quantum simulation.
\newblock \emph{Nature Physics}, 8\penalty0 (4):\penalty0 263--263, 2012.
\newblock \doi{https://doi.org/10.1038/nphys2258}.

\bibitem[Pirandola et~al.(2018)Pirandola, Bardhan, Gehring, Weedbrook, and
  Lloyd]{pirandola2018advances}
Stefano Pirandola, B~Roy Bardhan, Tobias Gehring, Christian Weedbrook, and Seth
  Lloyd.
\newblock Advances in photonic quantum sensing.
\newblock \emph{Nature Photonics}, 12\penalty0 (12):\penalty0 724--733, 2018.
\newblock \doi{10.1038/s41566-018-0301-6}.

\bibitem[Degen et~al.(2017)Degen, Reinhard, and Cappellaro]{degen2017quantum}
Christian~L Degen, F~Reinhard, and Paola Cappellaro.
\newblock Quantum sensing.
\newblock \emph{Reviews of modern physics}, 89\penalty0 (3):\penalty0 035002,
  2017.
\newblock \doi{10.1103/RevModPhys.89.035002}.

\bibitem[Boss et~al.(2017)Boss, Cujia, Zopes, and Degen]{boss2017quantum}
Jens~M Boss, KS~Cujia, Jonathan Zopes, and Christian~L Degen.
\newblock Quantum sensing with arbitrary frequency resolution.
\newblock \emph{Science}, 356\penalty0 (6340):\penalty0 837--840, 2017.
\newblock \doi{10.1126/science.aam7009}.

\bibitem[Chuang and Nielsen(1997)]{Chuang1997}
Isaac~L. Chuang and M.~A. Nielsen.
\newblock Prescription for experimental determination of the dynamics of a
  quantum black box.
\newblock \emph{J. Mod. Opt.}, 44\penalty0 (11-12):\penalty0 2455--2467, 1997.
\newblock \doi{10.1080/09500349708231894}.

\bibitem[Altepeter et~al.(2003)Altepeter, Branning, Jeffrey, Wei, Kwiat, Thew,
  O'Brien, Nielsen, and White]{PhysRevLett.90.193601}
J.~B. Altepeter, D.~Branning, E.~Jeffrey, T.~C. Wei, P.~G. Kwiat, R.~T. Thew,
  J.~L. O'Brien, M.~A. Nielsen, and A.~G. White.
\newblock Ancilla-assisted quantum process tomography.
\newblock \emph{Phys. Rev. Lett.}, 90:\penalty0 193601, May 2003.
\newblock \doi{10.1103/PhysRevLett.90.193601}.

\bibitem[Leung(2003)]{leung2003choi}
Debbie~W Leung.
\newblock Choi’s proof as a recipe for quantum process tomography.
\newblock \emph{Journal of Mathematical Physics}, 44\penalty0 (2):\penalty0
  528--533, 2003.
\newblock \doi{10.1063/1.1518554}.

\bibitem[Merkel et~al.(2013)Merkel, Gambetta, Smolin, Poletto, C{\'o}rcoles,
  Johnson, Ryan, and Steffen]{Merkel2012}
Seth~T. Merkel, Jay~M. Gambetta, John~A. Smolin, S.~Poletto, A.~D.
  C{\'o}rcoles, B.~R. Johnson, Colm~A. Ryan, and M.~Steffen.
\newblock Self-consistent quantum process tomography.
\newblock \emph{Phys. Rev. A}, 87:\penalty0 062119, 2013.
\newblock \doi{10.1103/PhysRevA.87.062119}.

\bibitem[Rahimi-Keshari et~al.(2011)Rahimi-Keshari, Scherer, Mann, Rezakhani,
  Lvovsky, and Sanders]{rahimi2011quantum}
Saleh Rahimi-Keshari, Artur Scherer, Ady Mann, Ali~T Rezakhani, AI~Lvovsky, and
  Barry~C Sanders.
\newblock Quantum process tomography with coherent states.
\newblock \emph{New Journal of Physics}, 13\penalty0 (1):\penalty0 013006,
  2011.
\newblock \doi{10.1088/1367-2630/13/1/013006}.

\bibitem[Mohseni et~al.(2008{\natexlab{a}})Mohseni, Rezakhani, and
  Lidar]{mohseni2008quantum}
Masoud Mohseni, Ali~T Rezakhani, and Daniel~A Lidar.
\newblock Quantum-process tomography: Resource analysis of different
  strategies.
\newblock \emph{Physical Review A}, 77\penalty0 (3):\penalty0 032322,
  2008{\natexlab{a}}.
\newblock \doi{10.1103/PhysRevA.77.032322}.

\bibitem[Baldwin et~al.(2014)Baldwin, Kalev, and Deutsch]{baldwin2014quantum}
Charles~H Baldwin, Amir Kalev, and Ivan~H Deutsch.
\newblock Quantum process tomography of unitary and near-unitary maps.
\newblock \emph{Physical Review A}, 90\penalty0 (1):\penalty0 012110, 2014.
\newblock \doi{10.1103/PhysRevA.90.012110}.

\bibitem[Preskill(2018)]{preskill2018quantum}
John Preskill.
\newblock Quantum {C}omputing in the {NISQ} era and beyond.
\newblock \emph{{Quantum}}, 2:\penalty0 79, August 2018.
\newblock ISSN 2521-327X.
\newblock \doi{10.22331/q-2018-08-06-79}.
\newblock URL \url{https://doi.org/10.22331/q-2018-08-06-79}.

\bibitem[Bairey et~al.(2019)Bairey, Arad, and Lindner]{Bairey_2019}
Eyal Bairey, Itai Arad, and Netanel~H. Lindner.
\newblock Learning a local hamiltonian from local measurements.
\newblock \emph{Phys. Rev. Lett.}, 122\penalty0 (2):\penalty0 020504, Jan 2019.
\newblock ISSN 1079-7114.
\newblock \doi{10.1103/physrevlett.122.020504}.
\newblock URL \url{http://dx.doi.org/10.1103/PhysRevLett.122.020504}.

\bibitem[Evans et~al.(2019)Evans, Harper, and Flammia]{evans2019scalable}
Tim~J Evans, Robin Harper, and Steven~T Flammia.
\newblock Scalable bayesian hamiltonian learning.
\newblock \emph{arXiv preprint arXiv:1912.07636}, 2019.
\newblock \doi{10.48550/arXiv.1912.07636}.

\bibitem[Anshu et~al.(2021)Anshu, Arunachalam, Kuwahara, and
  Soleimanifar]{anshu2020sample}
Anurag Anshu, Srinivasan Arunachalam, Tomotaka Kuwahara, and Mehdi
  Soleimanifar.
\newblock Sample-efficient learning of interacting quantum systems.
\newblock \emph{Nature Physics}, 17\penalty0 (8):\penalty0 931--935, 2021.
\newblock \doi{10.1038/s41567-021-01232-0}.

\bibitem[Qi and Ranard(2019)]{qi2019determining}
Xiao-Liang Qi and Daniel Ranard.
\newblock Determining a local hamiltonian from a single eigenstate.
\newblock \emph{Quantum}, 3:\penalty0 159, 2019.
\newblock \doi{10.22331/q-2019-07-08-159}.

\bibitem[Li et~al.(2020)Li, Zou, and Hsieh]{quench}
Zhi Li, Liujun Zou, and Timothy~H. Hsieh.
\newblock Hamiltonian tomography via quantum quench.
\newblock \emph{Phys. Rev. Lett.}, 124:\penalty0 160502, Apr 2020.
\newblock \doi{10.1103/PhysRevLett.124.160502}.
\newblock URL \url{https://link.aps.org/doi/10.1103/PhysRevLett.124.160502}.

\bibitem[Zubida et~al.(2021)Zubida, Yitzhaki, Lindner, and
  Bairey]{zubida2021optimal}
Assaf Zubida, Elad Yitzhaki, Netanel~H Lindner, and Eyal Bairey.
\newblock Optimal short-time measurements for hamiltonian learning.
\newblock \emph{arXiv preprint arXiv:2108.08824}, 2021.
\newblock \doi{https://doi.org/10.48550/arXiv.2108.08824}.

\bibitem[Hangleiter et~al.(2021)Hangleiter, Roth, Eisert, and
  Roushan]{hangleiter2021precise}
Dominik Hangleiter, Ingo Roth, Jens Eisert, and Pedram Roushan.
\newblock Precise hamiltonian identification of a superconducting quantum
  processor.
\newblock \emph{arXiv preprint arXiv:2108.08319}, 2021.
\newblock \doi{10.48550/arXiv.2108.08319}.

\bibitem[Flammia and Wallman(2020)]{Flammia2019}
Steven~T. Flammia and Joel~J. Wallman.
\newblock Efficient estimation of pauli channels.
\newblock \emph{ACM Transactions on Quantum Computing}, 1\penalty0 (1), dec
  2020.
\newblock ISSN 2643-6809.
\newblock \doi{10.1145/3408039}.
\newblock URL \url{https://doi.org/10.1145/3408039}.

\bibitem[Mezzacapo et~al.(2014)Mezzacapo, Las~Heras, Pedernales, DiCarlo,
  Solano, and Lamata]{mezzacapo2014digital}
Antonio Mezzacapo, U~Las~Heras, JS~Pedernales, L~DiCarlo, E~Solano, and
  L~Lamata.
\newblock Digital quantum rabi and dicke models in superconducting circuits.
\newblock \emph{Scientific reports}, 4\penalty0 (1):\penalty0 1--4, 2014.
\newblock \doi{10.1038/srep07482}.

\bibitem[Garc{\'\i}a-{\'A}lvarez et~al.(2015)Garc{\'\i}a-{\'A}lvarez, Casanova,
  Mezzacapo, Egusquiza, Lamata, Romero, and Solano]{garcia2015fermion}
L~Garc{\'\i}a-{\'A}lvarez, J~Casanova, A~Mezzacapo, IL~Egusquiza, L~Lamata,
  G~Romero, and E~Solano.
\newblock Fermion-fermion scattering in quantum field theory with
  superconducting circuits.
\newblock \emph{Physical review letters}, 114\penalty0 (7):\penalty0 070502,
  2015.
\newblock \doi{10.1103/PhysRevLett.114.070502}.

\bibitem[Asaad et~al.(2016)Asaad, Dickel, Langford, Poletto, Bruno, Rol,
  Deurloo, and DiCarlo]{asaad2016independent}
Serwan Asaad, Christian Dickel, Nathan~K Langford, Stefano Poletto, Alessandro
  Bruno, Michiel~Adriaan Rol, Duije Deurloo, and Leonardo DiCarlo.
\newblock Independent, extensible control of same-frequency superconducting
  qubits by selective broadcasting.
\newblock \emph{npj Quantum Information}, 2\penalty0 (1):\penalty0 1--7, 2016.
\newblock \doi{10.1038/npjqi.2016.29}.

\bibitem[Weber et~al.(2017)Weber, Samach, Hover, Gustavsson, Kim, Melville,
  Rosenberg, Sears, Yan, Yoder, et~al.]{weber2017coherent}
Steven~J Weber, Gabriel~O Samach, David Hover, Simon Gustavsson, David~K Kim,
  Alexander Melville, Danna Rosenberg, Adam~P Sears, Fei Yan, Jonilyn~L Yoder,
  et~al.
\newblock Coherent coupled qubits for quantum annealing.
\newblock \emph{Physical Review Applied}, 8\penalty0 (1):\penalty0 014004,
  2017.
\newblock \doi{10.1103/PhysRevApplied.8.014004}.

\bibitem[H{\"a}ffner et~al.(2005)H{\"a}ffner, H{\"a}nsel, Roos, Benhelm,
  Chwalla, K{\"o}rber, Rapol, Riebe, Schmidt, Becher,
  et~al.]{haffner2005scalable}
Hartmut H{\"a}ffner, Wolfgang H{\"a}nsel, CF~Roos, Jan Benhelm, Michael
  Chwalla, Timo K{\"o}rber, UD~Rapol, Mark Riebe, PO~Schmidt, Christoph Becher,
  et~al.
\newblock Scalable multiparticle entanglement of trapped ions.
\newblock \emph{Nature}, 438\penalty0 (7068):\penalty0 643--646, 2005.
\newblock \doi{10.1038/nature04279}.

\bibitem[Blatt and Roos(2012)]{blatt2012quantum}
Rainer Blatt and Christian~F Roos.
\newblock Quantum simulations with trapped ions.
\newblock \emph{Nature Physics}, 8\penalty0 (4):\penalty0 277, 2012.
\newblock \doi{10.1038/nphys2252}.

\bibitem[Saffman(2016)]{saffman2016quantum}
Mark Saffman.
\newblock Quantum computing with atomic qubits and rydberg interactions:
  progress and challenges.
\newblock \emph{Journal of Physics B: Atomic, Molecular and Optical Physics},
  49\penalty0 (20):\penalty0 202001, 2016.
\newblock \doi{10.1088/0953-4075/49/20/202001}.

\bibitem[Cory et~al.(1998)Cory, Price, and Havel]{cory1998nuclear}
David~G Cory, Mark~D Price, and Timothy~F Havel.
\newblock Nuclear magnetic resonance spectroscopy: An experimentally accessible
  paradigm for quantum computing.
\newblock \emph{Physica D: Nonlinear Phenomena}, 120\penalty0 (1-2):\penalty0
  82--101, 1998.
\newblock \doi{10.1016/S0167-2789(98)00046-3}.

\bibitem[Franca et~al.(2022)Franca, Markovich, Dobrovitski, Werner, and
  Borregaard]{franca2022efficient}
Daniel~Stilck Franca, Liubov~A Markovich, VV~Dobrovitski, Albert~H Werner, and
  Johannes Borregaard.
\newblock Efficient and robust estimation of many-qubit hamiltonians.
\newblock \emph{arXiv preprint arXiv:2205.09567}, 2022.
\newblock \doi{https://doi.org/10.48550/arXiv.2205.09567}.

\bibitem[Gu et~al.(2022)Gu, Cincio, and Coles]{gu2022practical}
Andi Gu, Lukasz Cincio, and Patrick~J Coles.
\newblock Practical black box hamiltonian learning.
\newblock \emph{arXiv preprint arXiv:2206.15464}, 2022.
\newblock \doi{https://doi.org/10.48550/arXiv.2206.15464}.

\bibitem[Wilde et~al.(2022)Wilde, Kshetrimayum, Roth, Hangleiter, Sweke, and
  Eisert]{wilde2022scalably}
Frederik Wilde, Augustine Kshetrimayum, Ingo Roth, Dominik Hangleiter, Ryan
  Sweke, and Jens Eisert.
\newblock Scalably learning quantum many-body hamiltonians from dynamical data.
\newblock \emph{arXiv preprint arXiv:2209.14328}, 2022.
\newblock \doi{https://doi.org/10.48550/arXiv.2209.14328}.

\bibitem[Huang et~al.(2023)Huang, Tong, Fang, and Su]{huang2022learning}
Hsin-Yuan Huang, Yu~Tong, Di~Fang, and Yuan Su.
\newblock Learning many-body hamiltonians with heisenberg-limited scaling.
\newblock \emph{Phys. Rev. Lett.}, 130:\penalty0 200403, May 2023.
\newblock \doi{10.1103/PhysRevLett.130.200403}.
\newblock URL \url{https://link.aps.org/doi/10.1103/PhysRevLett.130.200403}.

\bibitem[Harper et~al.(2021)Harper, Yu, and Flammia]{harper2020fast}
Robin Harper, Wenjun Yu, and Steven~T. Flammia.
\newblock Fast estimation of sparse quantum noise.
\newblock \emph{PRX Quantum}, 2:\penalty0 010322, Feb 2021.
\newblock \doi{10.1103/PRXQuantum.2.010322}.
\newblock URL \url{https://link.aps.org/doi/10.1103/PRXQuantum.2.010322}.

\bibitem[Shabani et~al.(2011)Shabani, Kosut, Mohseni, Rabitz, Broome, Almeida,
  Fedrizzi, and White]{Shabani2011}
A.~Shabani, R.~L. Kosut, M.~Mohseni, H.~Rabitz, M.~A. Broome, M.~P. Almeida,
  A.~Fedrizzi, and A.~G. White.
\newblock Efficient measurement of quantum dynamics via compressive sensing.
\newblock \emph{Phys. Rev. Lett.}, 106\penalty0 (10):\penalty0 100401, Mar
  2011.
\newblock \doi{10.1103/PhysRevLett.106.100401}.

\bibitem[Hou et~al.(2020)Hou, Cao, Lu, Shen, Poon, and
  Zeng]{hou2020determining}
Shi-Yao Hou, Ningping Cao, Sirui Lu, Yi~Shen, Yiu-Tung Poon, and Bei Zeng.
\newblock Determining system hamiltonian from eigenstate measurements without
  correlation functions.
\newblock \emph{New Journal of Physics}, 22\penalty0 (8):\penalty0 083088,
  2020.
\newblock \doi{10.1088/1367-2630/abaacf}.

\bibitem[Garrison and Grover(2018)]{PhysRevX.8.021026}
James~R. Garrison and Tarun Grover.
\newblock Does a single eigenstate encode the full hamiltonian?
\newblock \emph{Phys. Rev. X}, 8:\penalty0 021026, Apr 2018.
\newblock \doi{10.1103/PhysRevX.8.021026}.
\newblock URL \url{https://link.aps.org/doi/10.1103/PhysRevX.8.021026}.

\bibitem[Takita et~al.(2017)Takita, Cross, C{\'o}rcoles, Chow, and
  Gambetta]{takita2017experimental}
Maika Takita, Andrew~W Cross, AD~C{\'o}rcoles, Jerry~M Chow, and Jay~M
  Gambetta.
\newblock Experimental demonstration of fault-tolerant state preparation with
  superconducting qubits.
\newblock \emph{Physical review letters}, 119\penalty0 (18):\penalty0 180501,
  2017.
\newblock \doi{10.1103/PhysRevLett.119.180501}.

\bibitem[Kandala et~al.(2019)Kandala, Temme, C{\'o}rcoles, Mezzacapo, Chow, and
  Gambetta]{kandala2019error}
Abhinav Kandala, Kristan Temme, Antonio~D C{\'o}rcoles, Antonio Mezzacapo,
  Jerry~M Chow, and Jay~M Gambetta.
\newblock Error mitigation extends the computational reach of a noisy quantum
  processor.
\newblock \emph{Nature}, 567\penalty0 (7749):\penalty0 491, 2019.
\newblock \doi{10.1038/s41586-019-1040-7}.

\bibitem[Sun et~al.(2021{\natexlab{a}})Sun, Yuan, Tsunoda, Vedral, Benjamin,
  and Endo]{sun2020mitigating}
Jinzhao Sun, Xiao Yuan, Takahiro Tsunoda, Vlatko Vedral, Simon~C. Benjamin, and
  Suguru Endo.
\newblock Mitigating realistic noise in practical noisy intermediate-scale
  quantum devices.
\newblock \emph{Phys. Rev. Applied}, 15:\penalty0 034026, Mar
  2021{\natexlab{a}}.
\newblock \doi{10.1103/PhysRevApplied.15.034026}.
\newblock URL \url{https://link.aps.org/doi/10.1103/PhysRevApplied.15.034026}.

\bibitem[Kempe et~al.(2006)Kempe, Kitaev, and Regev]{kempe2006complexity}
Julia Kempe, Alexei Kitaev, and Oded Regev.
\newblock The complexity of the local hamiltonian problem.
\newblock \emph{Siam journal on computing}, 35\penalty0 (5):\penalty0
  1070--1097, 2006.
\newblock \doi{10.1007/978-3-540-30538-5_31}.

\bibitem[Childs et~al.(2018)Childs, Maslov, Nam, Ross, and
  Su]{childs2018toward}
Andrew~M Childs, Dmitri Maslov, Yunseong Nam, Neil~J Ross, and Yuan Su.
\newblock Toward the first quantum simulation with quantum speedup.
\newblock \emph{Proceedings of the National Academy of Sciences}, 115\penalty0
  (38):\penalty0 9456--9461, 2018.
\newblock \doi{10.1073/pnas.1801723115}.

\bibitem[Low and Chuang(2019)]{low2019hamiltonian}
Guang~Hao Low and Isaac~L Chuang.
\newblock Hamiltonian simulation by qubitization.
\newblock \emph{Quantum}, 3:\penalty0 163, 2019.
\newblock \doi{10.22331/q-2019-07-12-163}.

\bibitem[Sun et~al.(2021{\natexlab{b}})Sun, Endo, Lin, Hayden, Vedral, and
  Yuan]{sun2021perturbative}
Jinzhao Sun, Suguru Endo, Huiping Lin, Patrick Hayden, Vlatko Vedral, and Xiao
  Yuan.
\newblock Perturbative quantum simulation.
\newblock \emph{arXiv preprint arXiv:2106.05938}, 2021{\natexlab{b}}.
\newblock \doi{10.48550/arXiv.2106.05938}.

\bibitem[Helsen et~al.(2022)Helsen, Roth, Onorati, Werner, and
  Eisert]{helsen2020general}
J.~Helsen, I.~Roth, E.~Onorati, A.H. Werner, and J.~Eisert.
\newblock General framework for randomized benchmarking.
\newblock \emph{PRX Quantum}, 3:\penalty0 020357, Jun 2022.
\newblock \doi{10.1103/PRXQuantum.3.020357}.

\bibitem[Helsen et~al.(2019{\natexlab{a}})Helsen, Wallman, Flammia, and
  Wehner]{helsen2019multiqubit}
Jonas Helsen, Joel~J Wallman, Steven~T Flammia, and Stephanie Wehner.
\newblock Multiqubit randomized benchmarking using few samples.
\newblock \emph{Physical Review A}, 100\penalty0 (3):\penalty0 032304,
  2019{\natexlab{a}}.
\newblock \doi{10.1103/PhysRevA.100.032304}.

\bibitem[Wallman(2018)]{Wallman2018}
Joel Wallman.
\newblock Randomized benchmarking with gate-dependent noise.
\newblock \emph{Quantum}, 2:\penalty0 47, 2018.
\newblock \doi{10.22331/q-2018-01-29-47}.

\bibitem[Harper et~al.(2020)Harper, Flammia, and Wallman]{harper2020efficient}
Robin Harper, Steven~T Flammia, and Joel~J Wallman.
\newblock Efficient learning of quantum noise.
\newblock \emph{Nature Physics}, 16\penalty0 (12):\penalty0 1184--1188, 2020.
\newblock \doi{https://doi.org/10.1038/s41567-020-0992-8}.

\bibitem[Barends et~al.(2014)Barends, Kelly, Megrant, Veitia, Sank, Jeffrey,
  White, Mutus, Fowler, Campbell, Chen, Chen, Chiaro, Dunsworth, Neill,
  O'Malley, Roushan, Vainsencher, Wenner, Korotkov, Cleland, and
  Martinis]{Barends2014}
R.~Barends, J.~Kelly, A.~Megrant, A.~Veitia, D.~Sank, E.~Jeffrey, T.~C. White,
  J.~Mutus, A.~G. Fowler, B.~Campbell, Y.~Chen, Z.~Chen, B.~Chiaro,
  A.~Dunsworth, C.~Neill, P.~O'Malley, P.~Roushan, A.~Vainsencher, J.~Wenner,
  A.~N. Korotkov, A.~N. Cleland, and John~M. Martinis.
\newblock Superconducting quantum circuits at the surface code threshold for
  fault tolerance.
\newblock \emph{Nature}, 508\penalty0 (7497):\penalty0 500--503, apr 2014.
\newblock \doi{10.1038/nature13171}.

\bibitem[Somoroff et~al.(2021)Somoroff, Ficheux, Mencia, Xiong, Kuzmin, and
  Manucharyan]{somoroff2021millisecond}
Aaron Somoroff, Quentin Ficheux, Raymond~A Mencia, Haonan Xiong, Roman~V
  Kuzmin, and Vladimir~E Manucharyan.
\newblock Millisecond coherence in a superconducting qubit.
\newblock \emph{arXiv preprint arXiv:2103.08578}, 2021.
\newblock \doi{10.48550/arXiv.2103.08578}.

\bibitem[Magesan(2012)]{Magesan2012a}
Easwar Magesan.
\newblock \emph{Characterizing Noise in Quantum Systems}.
\newblock PhD thesis, University of Waterloo, Waterloo, Ontario, Canada, 2012.
\newblock URL
  \url{https://uwspace.uwaterloo.ca/bitstream/handle/10012/6832/Magesan_Easwar.pdf}.

\bibitem[Dankert et~al.(2009)Dankert, Cleve, Emerson, and Livine]{Dankert2009}
Christoph Dankert, Richard Cleve, Joseph Emerson, and Etera Livine.
\newblock Exact and approximate unitary 2-designs and their application to
  fidelity estimation.
\newblock \emph{Phys. Rev. A}, 80:\penalty0 012304, Jul 2009.
\newblock \doi{10.1103/PhysRevA.80.012304}.

\bibitem[Erhard et~al.(2019)Erhard, Wallman, Postler, Meth, Stricker, Martinez,
  Schindler, Monz, Emerson, and Blatt]{Erhard2019}
Alexander Erhard, Joel~J. Wallman, Lukas Postler, Michael Meth, Roman Stricker,
  Esteban~A. Martinez, Philipp Schindler, Thomas Monz, Joseph Emerson, and
  Rainer Blatt.
\newblock Characterizing large-scale quantum computers via cycle benchmarking.
\newblock \emph{Nat. Commun.}, 10\penalty0 (1):\penalty0 5347, November 2019.
\newblock \doi{10.1038/s41467-019-13068-7}.
\newblock URL \url{https://doi.org/10.1038/s41467-019-13068-7}.

\bibitem[Helsen et~al.(2019{\natexlab{b}})Helsen, Xue, Vandersypen, and
  Wehner]{Helsen2018}
Jonas Helsen, Xiao Xue, Lieven~MK Vandersypen, and Stephanie Wehner.
\newblock A new class of efficient randomized benchmarking protocols.
\newblock \emph{npj Quantum Inf.}, 5:\penalty0 71, 2019{\natexlab{b}}.
\newblock \doi{10.1038/s41534-019-0182-7}.

\bibitem[Mohseni et~al.(2008{\natexlab{b}})Mohseni, Rezakhani, and
  Lidar]{PhysRevA.77.032322}
M.~Mohseni, A.~T. Rezakhani, and D.~A. Lidar.
\newblock Quantum-process tomography: Resource analysis of different
  strategies.
\newblock \emph{Phys. Rev. A}, 77:\penalty0 032322, Mar 2008{\natexlab{b}}.
\newblock \doi{10.1103/PhysRevA.77.032322}.
\newblock URL \url{https://link.aps.org/doi/10.1103/PhysRevA.77.032322}.

\bibitem[Poyatos et~al.(1997)Poyatos, Cirac, and Zoller]{PhysRevLett.78.390}
J.~F. Poyatos, J.~I. Cirac, and P.~Zoller.
\newblock Complete characterization of a quantum process: The two-bit quantum
  gate.
\newblock \emph{Phys. Rev. Lett.}, 78:\penalty0 390--393, Jan 1997.
\newblock \doi{10.1103/PhysRevLett.78.390}.
\newblock URL \url{https://link.aps.org/doi/10.1103/PhysRevLett.78.390}.

\bibitem[Yuan et~al.(2016)Yuan, Zhang, L{\"u}tkenhaus, and
  Ma]{yuan2016simulating}
Xiao Yuan, Zhen Zhang, Norbert L{\"u}tkenhaus, and Xiongfeng Ma.
\newblock Simulating single photons with realistic photon sources.
\newblock \emph{Physical Review A}, 94\penalty0 (6):\penalty0 062305, 2016.
\newblock \doi{10.1103/PhysRevA.94.062305}.

\bibitem[Troiani et~al.(2005)Troiani, Ghirri, Affronte, Carretta, Santini,
  Amoretti, Piligkos, Timco, and Winpenny]{troiani2005molecular}
Filippo Troiani, Alberto Ghirri, Marco Affronte, S~Carretta, P~Santini,
  G~Amoretti, S~Piligkos, G~Timco, and REP Winpenny.
\newblock Molecular engineering of antiferromagnetic rings for quantum
  computation.
\newblock \emph{Physical review letters}, 94\penalty0 (20):\penalty0 207208,
  2005.
\newblock \doi{10.1103/PhysRevLett.94.207208}.

\bibitem[Santini et~al.(2011)Santini, Carretta, Troiani, and
  Amoretti]{santini2011molecular}
P~Santini, S~Carretta, F~Troiani, and G~Amoretti.
\newblock Molecular nanomagnets as quantum simulators.
\newblock \emph{Physical review letters}, 107\penalty0 (23):\penalty0 230502,
  2011.
\newblock \doi{10.1103/PhysRevLett.107.230502}.

\bibitem[Han and Yu(2021)]{GITHUB}
Zeyao Han and Wenjun Yu.
\newblock {Github: HamiltonianLearning}.
\newblock \url{https://github.com/zyHan2077/HamiltonianLearning}, December
  2021.

\bibitem[Chen et~al.(2022)Chen, Zhou, Seif, and Jiang]{chen2021quantum}
Senrui Chen, Sisi Zhou, Alireza Seif, and Liang Jiang.
\newblock Quantum advantages for pauli channel estimation.
\newblock \emph{Phys. Rev. A}, 105:\penalty0 032435, Mar 2022.
\newblock \doi{10.1103/PhysRevA.105.032435}.
\newblock URL \url{https://link.aps.org/doi/10.1103/PhysRevA.105.032435}.

\bibitem[Li et~al.(2014)Li, Bradley, Pawar, and Ramchandran]{Li2015}
Xiao Li, Joseph~Kurata Bradley, Sameer Pawar, and Kannan Ramchandran.
\newblock The {SPRIGHT} algorithm for robust sparse {H}adamard transforms.
\newblock In \emph{2014 {IEEE} International Symposium on Information Theory},
  Honolulu, Hawaii, USA, June 2014. {IEEE}.
\newblock \doi{10.1109/isit.2014.6875155}.

\bibitem[{Candes} and {Tao}(2005)]{2005decode}
E.~J. {Candes} and T.~{Tao}.
\newblock Decoding by linear programming.
\newblock \emph{IEEE Transactions on Information Theory}, 51\penalty0
  (12):\penalty0 4203--4215, 2005.
\newblock \doi{10.1109/TIT.2005.858979}.

\bibitem[Candes et~al.(2006)Candes, Romberg, and Tao]{candes2006stable}
Emmanuel~J Candes, Justin~K Romberg, and Terence Tao.
\newblock Stable signal recovery from incomplete and inaccurate measurements.
\newblock \emph{Communications on Pure and Applied Mathematics: A Journal
  Issued by the Courant Institute of Mathematical Sciences}, 59\penalty0
  (8):\penalty0 1207--1223, 2006.
\newblock \doi{10.1002/cpa.20124}.

\bibitem[Baraniuk et~al.(2008)Baraniuk, Davenport, DeVore, and
  Wakin]{baraniuk2008simple}
Richard Baraniuk, Mark Davenport, Ronald DeVore, and Michael Wakin.
\newblock A simple proof of the restricted isometry property for random
  matrices.
\newblock \emph{Constructive Approximation}, 28\penalty0 (3):\penalty0
  253--263, 2008.
\newblock \doi{10.1007/s00365-007-9003-x}.

\bibitem[Candès(2008)]{CANDES2008589}
Emmanuel~J. Candès.
\newblock The restricted isometry property and its implications for compressed
  sensing.
\newblock \emph{Comptes Rendus Mathematique}, 346\penalty0 (9):\penalty0
  589--592, 2008.
\newblock ISSN 1631-073X.
\newblock \doi{https://doi.org/10.1016/j.crma.2008.03.014}.
\newblock URL
  \url{https://www.sciencedirect.com/science/article/pii/S1631073X08000964}.

\end{thebibliography}

\onecolumn
\newpage
\appendix

\section{Preliminaries}
\subsection{Pauli Concept}\label{sec:PauliConcept}
In this section, we will introduce the commonly-used Pauli group and the label to represent every element.
We will start from the single-qubit case.

For the single-qubit Pauli group $\mathbb{P}$, it is generated by three Pauli operators, 
\begin{gather}
    \mathbb{P}=\langle X,Y,Z\rangle.
\end{gather}
The multiplication among these generators will cause some extra phase terms for some operators, so this group contains several elements that are not Hermitian.
On the other hand, all the Pauli elements serve either as observables or gates.
Therefore, we only take the \emph{refined Pauli group} $\sf P$ with all Hermitian elements into account.
In group $\sf{P}$, we treat elements with only differences on phase the same, which is equivalent to making a quotient over the phase element
\begin{gather}
    \sf{P}=\faktor{\mathbb{P}}{\langle iI\rangle}.
\end{gather}
Moreover, for every integer $n$, the \emph{refined $n$-qubit Pauli group} $\sf{P}^n$ is constructed from tensor products of $\sf{P}$, and we have
\begin{gather}
    \sf{P}^n=\bigotimes_{i=1}^n\sf{P}.
\end{gather}

In order to specify every element among $4^n$ in the refined $n$-qubit Pauli group, a canonical way will be using a $2n$-bit string to label every Pauli operator in the refined group.
We can start to illustrate this convention from the case of the single-qubit Pauli group $\sf P$.
This refined group contains four elements which are the three Pauli matrices and an identity operator.
As for the label, we refer to a bit string $\alpha$ with a length of two,
\begin{table}[H]
\centering
\begin{tabular}{lllll}
                      & $I$  & $X$  & $Y$  & $Z$  \\
$\alpha$ & 00 & 01 & 10 & 11
\end{tabular}.
\end{table}
Therefore, if we want to label an $n$-qubit Pauli element, we recruit $2n$-bit strings to record the Pauli element in each qubit.
For a $2n$-bit label $\alpha$, the $2i-1^{\text{th}}$ and $2i^{\text{th}}$ bits, $\alpha_{2i-1}\alpha_{2i}$, characterize the corresponding single Pauli operator of $P_\alpha$ on the $i$th qubit.
For example, if we want to represent the operator $XZIX$, then its label must be $\alpha=01110001$.

Since there exists an isomorphism between the $2n$-bit labels and the Pauli elements, we need to define the corresponding operations that use the label to express some properties of Pauli operators.
The basic operation is the addition.
We define the addition of Pauli labels following the multiplication of Pauli operators.
Regardless of phases, the summation is equal to the multiplication of every pair of Pauli operators so that
\begin{gather}
    P_\alpha\cdot P_\beta=P_{\alpha+\beta},\ \ \forall\,P_\alpha,P_\beta\in \sf{P}^n.
\end{gather}
Note that we will use $\alpha\in\sf{P}^n$ for short of denoting $P_\alpha\in\sf{P}^n$ when it would not incur any confusion.
We next introduce a new inner product between labels of two Pauli operators.
For two $n$-qubit Pauli labels $\alpha$ and $\beta$, we define the \emph{Pauli inner product} as follows
\begin{gather}
    \langle\alpha,\beta\rangle_p\coloneqq\sum_{i=1}^n(\alpha_{2i-1}\cdot\beta_{2i}+\beta_{2i-1}\cdot\alpha_{2i})\mod{2},
\end{gather}
where we use the subscript to distinguish the Pauli inner product from the ordinary inner product.
This operation actually describes the commutation of every pair of Pauli operators.
By calculating the qubitwise commutation and accumulating them, the resulting Boolean value indicates the commutation property as
\begin{align}
    [P_\alpha,P_\beta]=\left(1-(-1)^{\langle\alpha,\beta\rangle_p}\right)P_\alpha P_\beta.
\end{align}

\subsection{Pauli Stabilizer Group}\label{sec:Stabilizer}
There exists a kind of special subgroup in the refined Pauli group, namely the \emph{stabilizer group}.
\begin{definition}
    A stabilizer group $S$ is defined to be a subgroup of the refined Pauli group ${\sf P}^n$ such that $\langle a,b\rangle_p=0$ for all $P_a,P_b\in S$.
\end{definition}

It can be directly inferred from the definition that each stabilizer group is a normal subgroup of the refined Pauli group.
Therefore, we can find further the quotient group $A_S\coloneqq \faktor{{\sf P}^n}{S}$.
By group theory, we know that for an arbitrary $P_c\in{\sf P}^n$, there exists a unique pair of $e_c\in A_S$ and $s_c\in S$ so that $P_c=e_c\cdot s_c$.
Since only $e_c$ contributes to the syndrome measurement after implementing $P_c$ to a stabilizer state, we denote $e_c$ by the error syndrome of $P_c$.

It is easy to prove that the maximum size of a stabilizer group is $2^n$ for ${\sf P}^n$.
Therefore, the size of the quotient group of a maximum stabilizer group is $2^n$.
Consider a maximum stabilizer group generated by $n$ generator $G\coloneqq\langle g_1,\cdots,g_n\rangle$, and for any Pauli operator $P_c=e_c\cdot s_c$, we can get the full commuting relation between $P_c$ and $G$ by the commuting relation between $P_c$ and set $\{g_1,\cdots,g_n\}$, which a $n$-bit string can represent.
Therefore, each element in $A_S$ maps onto a certain $n$-bit string, and this serves the basics of the stabilizer algebra.
Consequently, if we implement a simultaneous measurement of all the generators referred to as the syndrome measurement, the results turn out to be the bit string representation of the error syndrome of the equivalent Pauli gate to the whole circuit.

During the analysis of the fidelity estimator in Sec.~\ref{sec:fidelityEstimation}, the protocol tracks and adds the syndromes from sampled Pauli gates and the final measurement.
The summation is processed by the corresponding bit strings by binary summation.
And the result is the overall syndrome of the circuit.

\subsection{Pauli Channel Characterizations}\label{sec:APauli}
The commonly used \emph{Pauli channel} can be regarded as a stochastic quantum channel.
The is due to the definition of Pauli channels,
\begin{gather}\label{eq:Paulichannel}
    \mathcal{E}^{\sf P}(\cdot)\coloneqq\sum_{\alpha\in{\sf P}^n}p_\alpha P_\alpha(\cdot)P_\alpha,
\end{gather}
where coefficients $\{p_\alpha\}$ are \emph{Pauli error rates}.
By the CPTP conditions for every physical channel, we have the following constraints over these Pauli error rates,
\begin{gather}
    \sum_{\alpha\in{\sf P}^n}p_\alpha=1,\ \ p_\alpha\geq0\ \forall\, \alpha\in{\sf P}^n.
\end{gather}
Therefore, the set $\{p_\alpha\}$ constitutes a probability distribution over $4^n$ possible Pauli operations.
Consequently, a Pauli channel with the form as Eq.~\eqref{eq:Paulichannel} can be interpreted as applying a Pauli gate $P_\alpha(\cdot)P_\alpha$ with probability $p_\alpha$.

Moreover, there is a second set of $4^n$ parameters that can fully characterize a Pauli channel $\mathcal{E}^{\sf P}$, which is defined in \Cref{sec:Pauli} as Pauli fidelity,
\begin{gather}\label{eq:Paulifidelity2}
    f_\alpha\coloneqq\frac{1}{2^n}\Tr(P_\alpha\mathcal{E}^{\sf P}(P_\alpha)),\ \ \forall\,\alpha\in{\sf P}^n.
\end{gather}
Note we can also deduce this value from the definition of Pauli channels.
From Eq.~\eqref{eq:Paulichannel}, every Pauli operator is an eigenoperator of an arbitrary Pauli channel, while the eigenvalue of $P_\alpha$ is just $f_\alpha$.
Thereby, some also refer values defined in Eq.~\eqref{eq:Paulifidelity} as \emph{Pauli eigenvalues}.

\section{Pauli Properties Estimation}\label{Estimator}
In this work, we estimate an arbitrary Hamiltonian operator with a sparse decomposition vector efficiently.
As stated in \Cref{sec:setting}, the decomposition vector $\bf s$ can faithfully describe a Hamiltonian operator.
Therefore, the problem of Hamiltonian estimation is, in principle, a scalable task.
According to Eq.~\eqref{eq:Hamiltonian}, the Hamiltonian channel contains the full information of the vector.
We choose the Hamiltonian evolution as the target, so the question is deduced to estimate a Hamiltonian channel.
In this section, we introduce the fidelity estimator in Appendix~\ref{sec:fidelity}. The estimator applies to the reference circuits and the circuit with the Hamiltonian evolution sequentially to extract the net Pauli fidelity of the Hamiltonian channel.
Moreover, we recruit the sparse Pauli error rates estimator in Appendix~\ref{sec:errorRate} to detect the error rates of the Hamiltonian channel, which are directly linked with the decomposition parameters.

\subsection{Pauli Fidelity Estimation}\label{sec:fidelity}
Here we first introduce the procedure and theoretical guarantees presented in \cite{Flammia2019}, and it will be recruited as a subroutine for the estimation of the Hamiltonian channel's Pauli fidelity.
This algorithm was first proposed to estimate the implementation quality of all Pauli gates.
To quantify this quality, we regard the practical implementation of a Pauli gate $P$ as $\tilde{P}=\Lambda\circ P$.
Therefore, Algorithm~\ref{alg:Fidelity} estimates the Pauli fidelity of this noise channel $\Lambda$ over a set of Pauli indices, $\sf X$, which belongs to a stabilizer subgroup $\sf G$ in the whole Pauli group $\sf{P}^n$.

For the implementation, we choose the input as a stabilizer state of stabilizer subgroup $\sf{G}$, and the final measurement will be a syndrome measurement of the group $\sf{G}$.
Note both the input states and the measurements contain some noise.
Hence the protocol must be noise-resilient of SPAM errors.
The execution starts with a series of random Pauli gates.
Intuitively, it first sets the length index $m$ as 0 as a reference, and $m$ increases exponentially from 1.
We denote this sequence of $m$ by $\kappa$.
The procedure then applies $m+1$ layers of uniformly random Pauli gates to the input state for each $m\in\kappa$.
Since the randomly twirled channel will be a diagonal channel in the superoperator representation, the expected channel of this circuit equals the composition of $m$ identically diagonal channels with some SPAM error effects, where the diagonal channel is equal to the diagonal part of the noise channel, $\Lambda$, namely, the Pauli fidelity terms of $\Lambda$.
By employing the stabilizer states and syndrome measurements, the results bring some diagonal information about the noise channel tarnished by SPAM errors.
More specifically, the results will be an exponential decay with the length index $m$ with factors to be the diagonal terms.
When a certain measurement result decays down to one-third of the reference, the algorithm then collects all results along $m\in\kappa$ to fit and get the diagonal information.
The Algorithm~\ref{alg:Fidelity} is shown in the following, where $l$ denotes the number of random sequences for every  length needed to construct the random estimator.
\begin{table}[thb]
\begin{algorithm}[H]
	\caption{\label{alg:Fidelity} FEstimator(${\sf G,X},l$)}
	\begin{algorithmic}[1]
	\State {\tt Input:} a target set $\sf X$ belonging to a stabilizer group $\sf{G}$, and a positive integer $l$
	\State Initialize an array $\hat{r}_\Lambda$ with size $|\sf X|$ to be all -1; Initialize an array $\bm{V}$ with size $|\sf X|$ to be all 0
	\State Input states will be positive stabilizer states $\tilde{\rho}_{\sf G}$; The measurement will be syndrome measurement $M_{\sf G}$ of $\sf G$  
	   \ForAll{$k\in[l]$}
	    \State Apply a random Pauli gate $P_{\alpha}\in\sf{P}^n$ to state $\tilde{\rho}_{\sf{G}}$, and perform measurement $M_{\sf{G}}$ with result $b$
        \ForAll{$x\in \sf{X}$}
         \State $\bm{V}[x]+=(-1)^{\langle x,\alpha+b\rangle_p}$
        \EndFor
	\EndFor
	\State $\hat{V} \leftarrow \frac{1}{l}\bm{V}$; $\bm{V}\leftarrow \bm{0}$; $m \leftarrow 1$
        \While {$\exists x \in \sf{X}$ such that $\hat{r}_\Lambda[x]= -1$}
        \ForAll{$k\in[l]$}
	    \State Apply $m+1$ random Pauli gates $\{P_{\alpha_0},\cdots,P_{\alpha_m}\}$ to state $\tilde{\rho}_{\sf G}$, and perform measurement $M_{\sf G}$ with result $b$
        \ForAll{$x\in \sf{X}$}
         \State $\bm{V}[x]+=(-1)^{\langle x,\sum_{i=0}^m \alpha_i+b\rangle_p}$
        \EndFor
	\EndFor
            \State $\hat{W} \leftarrow \frac{1}{l}\bm{V}$; $\bm{V}\leftarrow \bm{0}$
            \ForAll {$x\in \sf{X}$ such that $\hat{r}_\Lambda[x]=-1$}
                \State $v \leftarrow \hat{V}[x]$ and $w \leftarrow \hat{W}[x]$
    	        \State If $w \leq v / 3$ and $w,v>0$, $\hat{r}_\Lambda[x] \leftarrow 1 - \bigl(w/v\bigr)^{1/m}$
        	    \State Else If $w\le0$ or $v\le0$, $\hat{r}_\Lambda[x] \leftarrow 1$
        	\EndFor
        	\State $m\leftarrow2m$
        \EndWhile
	\State \Return the array $\hat{r}_\Lambda $
	\end{algorithmic}
\end{algorithm}
\end{table}

Even though the above algorithm works for a subset of a stabilizer group, we can simply generalize this method to detect a general subset $\sf{X}$ in the Pauli group.
To achieve that, we employ the concept of stabilizer covering.
For a set $\sf{X}$ of Pauli operators, its stabilizer covering is a collection of stabilizer groups, of which the union set covers $\sf{X}$.
Therefore, we can divide the target set exclusively into multiple stabilizer groups and execute the algorithm for each piece to estimate the whole set.
\begin{proposition}[Rephrasing Proposition 8 in~\cite{Flammia2019}]
\label{prop:RatioProp}
Let $\sf G$ be a stabilizer group in ${\sf P}^n$ and ${\sf X}\subset{\sf G}$.
Suppose the noise $\Lambda$ satisfies the assumption~\ref{assump:gtm} and $\|\Lambda-\mathcal{I}\|<\frac{1}{2}$, and the following holds with probability $1-\delta$ by choosing small $\epsilon,\delta>0$:
Running ${\bf FEstimator}({\sf G,X},l)$ with $l = \frac{2}{\epsilon^2} \log\bigl(\frac{2|\sf{X}||\kappa|}{\delta}\bigr)$ where $\kappa$ is the set of exponentially increasing sequence lengths. 
The procedure uses at most $O\left(\log{\frac{1}{\Delta_{\sf{X}}}}\right)$ sequence lengths with largest length $m$ to be $O\bigl( \tfrac{1}{\Delta_{\sf{X}}}\bigr)$, where $\Delta_{\sf{X}}$ denotes the smallest residual $r_\Lambda\coloneqq1-f_\Lambda$ among the set $\sf{X}$ regarding Pauli fidelity $f_\Lambda$ of noise channel $\Lambda$.
Moreover, the output $\hat{r}_\Lambda$ satisfies
\begin{align*}
    |\hat{r}_{\Lambda,x} - r_{\Lambda,x}| \leq O(\epsilon) r_{\Lambda,x}\ \ \forall x\in\sf{X}.
\end{align*}
\end{proposition}

Beyond this standard Pauli fidelity estimation for the implementation noise of Pauli gates, we want further to detect the fidelity information of the Hamiltonian evolution.
The idea of this detection is similar to that of the standard procedure for the noise channel.
We will use random Pauli gates to twirl the desired Hamiltonian channel.

For convenience, we denote the extended algorithm by $\textbf{HFEstimator}({\sf G,X},l,t)$ where the additional parameter $t$ represents the evolving time of the Hamiltonian channel.
The major difference between this protocol and the original one is that the new procedure will estimate the composite channel of the implementation noise and the Hamiltonian evolution, which is $\mathcal{H}_t\circ \Lambda$.
After implementing a random Pauli gate, the new procedure appends the Hamiltonian channel $\mathcal{H}_t$ to the circuit, as shown in Algorithm~\ref{alg:HFidelity}.
In summary, the modified procedure will implement $m+1$ random gates along with the channel $\mathcal{H}_t$, alternatively.
The circuit is illustrated in Figure~\ref{fig:rbcircuit}(b).
\begin{table}[thb]
\begin{algorithm}[H]
	\caption{\label{alg:HFidelity} HFEstimator(${\sf G,X},l,t$)}
	\begin{algorithmic}[1]
	\State {\tt Input:} a target set $\sf X$ belonging to a stabilizer group $\sf{G}$, a positive integer $l$, and the evolving time $t$
	\State Initialize an array $\hat{r}_{\mathcal{H}\circ\Lambda}$ with size $|\sf X|$ to be all -1; Initialize an array $\bm{V}$ with size $|\sf X|$ to be all 0
	\State Input states will be positive stabilizer states $\tilde{\rho}_{\sf G}$; The measurement will be syndrome measurement $M_{\sf G}$ of $\sf G$  
	   \ForAll{$k\in[l]$}
	    \State Apply a random Pauli gate $P_{\alpha}\in\sf{P}^n$ and Hamiltonian evolution $\mathcal{H}_t$ to state $\tilde{\rho}_{\sf{G}}$; Perform measurement $M_{\sf{G}}$ with result $b$
        \ForAll{$x\in \sf{X}$}
         \State $\bm{V}[x]+=(-1)^{\langle x,\alpha+b\rangle_p}$
        \EndFor
	\EndFor
	\State $\hat{V} \leftarrow \frac{1}{l}\bm{V}$; $\bm{V}\leftarrow \bm{0}$; $m \leftarrow 1$
        \While {$\exists x \in \sf{X}$ such that $\hat{r}_{\mathcal{H}\circ\Lambda}[x]= -1$}
        \ForAll{$k\in[l]$}
	    \State Apply $m+1$ random Pauli gates $\{P_{\alpha_0},\cdots,P_{\alpha_m}\}$ interleaved with $\mathcal{H}_t$ to $\tilde{\rho}_{\sf G}$; Perform measurement $M_{\sf G}$ with result $b$
        \ForAll{$x\in \sf{X}$}
         \State $\bm{V}[x]+=(-1)^{\langle x,\sum_{i=0}^m \alpha_i+b\rangle_p}$
        \EndFor
	\EndFor
            \State $\hat{W} \leftarrow \frac{1}{l}\bm{V}$; $\bm{V}\leftarrow \bm{0}$
            \ForAll {$x\in \sf{X}$ such that $\hat{r}_{\mathcal{H}\circ\Lambda}[x]=-1$}
                \State $v \leftarrow \hat{V}[x]$ and $w \leftarrow \hat{W}[x]$
    	        \State If $w \leq v / 3$ and $w,v>0$, $\hat{r}_{\mathcal{H}\circ\Lambda}[x] \leftarrow 1 - \bigl(w/v\bigr)^{1/m}$
        	    \State Else If $w\le0$ or $v\le0$, $\hat{r}_{\mathcal{H}\circ\Lambda}[x] \leftarrow 1$
        	\EndFor
        	\State $m\leftarrow2m$
        \EndWhile
	\State \Return the array $\hat{r}_{\mathcal{H}\circ\Lambda} $
	\end{algorithmic}
\end{algorithm}
\end{table}

Similar to Algorithm~\ref{alg:Fidelity}, this algorithm actually replaces the noise channel $\Lambda$ by the composite channel $\mathcal{H}_t\circ\Lambda$.
Therefore, the resulting estimation from Algorithm~\ref{alg:HFidelity} keeps very similar guarantees to Proposition~\ref{prop:RatioProp}.

\begin{corollary}
\label{lm:hfPrecision}
Let $\sf{X}\subseteq\sf{P}^n$ and $\sf G$ be a stabilizer group containing $\sf X$.
Suppose the noise $\Lambda$ satisfies assumption~\ref{assump:gtm} and assumption~\ref{assump:Pauli}, and the following holds with probability $1-\delta$ by choosing small $\epsilon,\delta>0$:
Running ${\bf HFEstimator}({\sf G,X},l,t)$ with $t$ satisfies $\|\mathcal{H}_t-\mathcal{I}\|<\frac{1}{4}$ and $l = \frac{2}{\epsilon^2} \log\bigl(\frac{2|\sf{X}||\kappa|}{\delta}\bigr)$ where $\kappa$ is the set of exponentially increasing sequence lengths. The procedure uses at most $O\left(\log{\frac{1}{\Delta_{\sf{X}}}}\right)$ sequence lengths with largest length $m$ to be $O\bigl( \tfrac{1}{\Delta_{\sf{X}}}\bigr)$, where $\Delta_{\sf{X}}$ denotes the smallest residual $r_{\mathcal{H}\circ\Lambda}\coloneqq1-f_{\mathcal{H}\circ\Lambda}$ among the set $\sf{X}$ regarding Pauli fidelity $f_{\mathcal{H}\circ\Lambda}$ of the composite channel $\mathcal{H}_t\circ\Lambda$.
Moreover, the output $\hat{r}_{\mathcal{H}\circ\Lambda}$ satisfies
\begin{align*}
    |\hat{r}_{\mathcal{H}\circ\Lambda,x} - r_{\mathcal{H}\circ\Lambda,x}| \leq O(\epsilon) r_{\mathcal{H}\circ\Lambda,x}\ \ \forall x\in\sf{X}.
\end{align*}
\end{corollary}
\begin{proof}
Since the Hamiltonian channel is defined to be $\mathcal{H}_t(\rho)=e^{-iHt}\rho e^{iHt}$, the channel is close to the identity channel when the evolving time $t$ is small.
With the assumption stated above, the composite channel $\mathcal{H}_t\circ\Lambda$ is close to the identity under the Pauli twirling.
Specifically, we have
\begin{align}\label{eq:closetoi}
    \|(\mathcal{H}_t\circ\Lambda)^{{\sf P}^n}-\mathcal{I}\|=&\max_{\alpha\in{\sf P}^n}\abs{1-f_{\mathcal{H}_t,\alpha}f_{\Lambda,\alpha}}\notag\\
    \leq&1-\min_{\alpha\in{\sf P}^n}f_{\mathcal{H}_t,\alpha}\min_{\beta\in{\sf P}^n}f_{\Lambda,\beta}\notag\\
    <&\frac{1}{2},
\end{align}
where the first inequality comes from the fact that all Pauli fidelity is upper bounded by 1, and the last inequality comes from the requirement of $t$ and the assumption~\ref{assump:Pauli}.
Therefore, according to the Proposition~\ref{prop:RatioProp}, \textbf{HFEstimator} can estimate the composite channel precisely as stated.
\end{proof}

With these two estimators, the full detection of the Hamiltonian's Pauli fidelity can be executed as in Algorithm~\ref{alg:query}.
The procedure first employs \textbf{FEstimator} to estimate information on the implementation noise of Pauli gates.
We call this detection the reference circuit and denote the detected fidelity by $f_\Lambda$.
Since we will also implement Pauli gates in the following circuits, this reference knowledge helps to rule out the effect of this noise.
In the second step, the procedure invokes \textbf{HFEstimator} to estimate circuits with the unknown Hamiltonian channel.
By definition, this subroutine returns the fidelity $f_{\mathcal{H}\circ\Lambda}$ about the composite of implementation noise and Hamiltonian channels.
By Assumption~\ref{assump:Pauli}, the implementation noise is the Pauli channel,  or namely, it is a diagonal matrix.
Thus, with a direct ratio estimator, the procedure eliminates the effects of noise and keeps the net estimation of Pauli fidelity of the Hamiltonian channel.
We denote this fidelity by $f_{\mathcal{H}_t}$ with residual $r_{\mathcal{H}_t}$.
\begin{lemma}\label{lm:fe}
Let $\sf{X}\subseteq\sf{P}^n$ and $\sf G$ be a stabilizer group containing $\sf X$.
For any small $\epsilon,\delta>0$ and noise under Assumption~\ref{assump:gtm} and \ref{assump:Pauli}, the following holds with probability $1-\delta$:
Run Algorithm~\ref{alg:query} with $t_0$ satisfies $\|\mathcal{H}_{5t_0}-\mathcal{I}\|<\frac{1}{4}$ and $l = \frac{2}{\epsilon^2} \log\bigl(\frac{2|\kappa|}{\delta}\bigr)$ where $\kappa$ is the set of variant sequence lengths.
The ratio estimator can estimate the Pauli fidelity information of the desired $\mathcal{H}_t$ satisfying that
\begin{gather*}
    \abs{\hat{f}_{\mathcal{H}_t,x}-f_{\mathcal{H}_t,x}}\leq\frac{O(\epsilon)(f_{\mathcal{H}_t,x}r_{\Lambda,x}+r_{\mathcal{H}_t\circ\Lambda,x})}{f_{\Lambda,x}-O(\epsilon)r_{\Lambda,x}},\ \forall\,x\in\sf{X}.
\end{gather*}
\end{lemma}
\begin{proof}
The guarantees from Proposition~\ref{prop:RatioProp} and Corollary~\ref{lm:hfPrecision} can help to provide this accuracy bound.
In the first round, the procedure estimates the fidelity information of the native noise $\Lambda$ by the stated complexity up to an accuracy
\begin{gather*}
    \abs{\hat{f}_{\Lambda,x}-f_{\Lambda,x}}\leq O(\epsilon) r_{\Lambda,x}.
\end{gather*}
Similarly, in the second round, the procedure executes the circuit including the composite $\mathcal{H}_t\circ\Lambda$ and returns another accuracy with a very similar complexity
\begin{gather*}
    \abs{\hat{f}_{\mathcal{H}_t\circ\Lambda,x}-f_{\mathcal{H}_t\circ\Lambda,x}}\leq O(\epsilon) r_{\mathcal{H}_t\circ\Lambda,x}.
\end{gather*}
According to Assumption~\ref{assump:Pauli}, we can show the fidelity is $f_{\mathcal{H}_t\circ\Lambda,x}=f_{\Lambda,x}f_{\mathcal{H}_t,x}$.
As $\Lambda$ is itself a Pauli channel, we have
\begin{align}
    f_{\mathcal{H}_t\circ\Lambda,x}=&\frac{1}{2^n}\Tr(P_i(\mathcal{H}_t\circ\Lambda)^{\sf P}(P_i))=\frac{1}{2^n}\Tr(P_i\mathcal{H}_t\circ\Lambda(P_i))\notag\\
    =&\frac{f_{\Lambda,x}}{2^n}\Tr(P_i\mathcal{H}_t(P_i))=\frac{f_{\Lambda,x}}{2^n}\Tr(P_i\mathcal{H}_t^{\sf P}(P_i))\notag\\
    =&f_{\Lambda,x}f_{\mathcal{H}_t,x}.
\end{align}
Thus we can calculate the estimation by the ratio estimator.
The upper and lower bounds of this estimation are therefore calculated as follows,
\begin{align}
    \frac{\hat{f}_{\mathcal{H}_t\circ\Lambda,x}}{\hat{f}_{\Lambda,x}}-f_{\mathcal{H}_t,x}&\leq\frac{f_{\Lambda,x}f_{\mathcal{H}_t,x}+O(\epsilon) r_{\mathcal{H}_t\circ\Lambda,x}}{f_{\Lambda,x}-O(\epsilon) r_{\Lambda,x}}-f_{\mathcal{H}_t,x}\notag\\
    &=\frac{O(\epsilon)(f_{\mathcal{H}_t,x}r_{\Lambda,x}+r_{\mathcal{H}_t\circ\Lambda,x})}{f_{\Lambda,x}-O(\epsilon)r_{\Lambda,x}}.
\end{align}
For the lower bound, we have a similar way of quantifying that
\begin{align}
    f_{\mathcal{H}_t,x}-\frac{\hat{f}_{\mathcal{H}_t\circ\Lambda,x}}{\hat{f}_{\Lambda,x}}&\leq f_{\mathcal{H}_t,x}-\frac{f_{\Lambda,x}f_{\mathcal{H}_t,x}-O(\epsilon) r_{\mathcal{H}_t\circ\Lambda,x}}{f_{\Lambda,x}+O(\epsilon) r_{\Lambda,x}}\notag\\
    &=\frac{O(\epsilon)(f_{\mathcal{H}_t,x}r_{\Lambda,x}+r_{\mathcal{H}_t\circ\Lambda,x})}{f_{\Lambda,x}-O(\epsilon)r_{\Lambda,x}}.
\end{align}
\end{proof}
\begin{remark}
\rm In the statement, we claim the errors of estimations will be bounded by a multiplicative accuracy.
Although there are different components involved in the relative part, this is still precise since the residual is small and all the fidelity terms are less than 1.
As we have introduced in the main text, the local gate noise is usually very close to the identity channel~\cite{harper2020efficient,Barends2014,somoroff2021millisecond}, which means $f_{\Lambda}$ is close to 1, and $r_{\Lambda}$ is negligible.
Likewise, the Hamiltonian channels are assigned with short evolving times in our learning method.
Thus, these channels are also well-behaved with large $f_{\mathcal{H}_t}$ and vanishing $r_{\mathcal{H}_t\circ\Lambda}$.
Therefore, our estimation will be of a multiplicative precision, and this error is much less than the predetermined $\epsilon$.
\end{remark}

In this analysis, we do not rule out the effects of evolving times, so this bound of estimation may vary from different Hamiltonian and evolving times.
To resolve this dependence, the procedure \textbf{FE} increases the evolving times evenly among multiple calls of \textbf{HFEstimator}.
With a regression over these multiple $t$, the procedure can estimate the second-order fidelity $f^{(2)}$ as stated in Eq.~\eqref{eq:chi}.
With the bounded noise proved in Lemma~\ref{lm:fe}, the regression parameter can be estimated with high accuracy.
\begin{propbis}{lm:f2bound}
Let $\sf{X}\subseteq\sf{P}^n$ and $\sf G$ be a stabilizer group containing $\sf X$, and suppose Assumption~\ref{assump:gtm} and~\ref{assump:Pauli} hold. 
For any small $\epsilon,\delta>0$, the following happens with probability $1-\delta$:
Run Algorithm~\ref{alg:query} $\textbf{FE}({\sf X,G},l,t_0)$ with $t_0$ satisfies $\|\mathcal{H}_{5t_0}-\mathcal{I}\|<\frac{1}{4}$ and  $l = \frac{2}{\epsilon^2} \log\bigl(\frac{2|\kappa|}{\delta}\bigr)$ where $\kappa$ is the set of variant sequence lengths. regression, and the resulting regression array $\hat{f}^{(2)}$ satisfies
\begin{gather*}
    \abs{\hat{f}^{(2)}_x-f^{(2)}_x}\leq\frac{\sigma O(\epsilon)(f_{\mathcal{H}_t,x}r_{\Lambda,x}+r_{\mathcal{H}_t\circ\Lambda,x})}{\left(f_{\Lambda,x}-O(\epsilon)r_{\Lambda,x}\right)t_0^2}+o(t_0^2),
\end{gather*}
where the upper bound term adopts the largest bound term referred from Lemma~\ref{lm:fe} among multiple evolution with different $t$, and $\sigma$ is a constant regrading the regression method.
\end{propbis}
\begin{proof}
According to the ordinary least squares, the parameters $\hat{f^{(2)}}$ can be calculated from the noisy observations and ideal predictors,
\begin{gather}\label{eq:OLS}
    \hat{f}^{(2)}_x\coloneqq\frac{\sum_{j=1}^k(t_j^2-\bar{t^2})\hat{f}_{\mathcal{H}_{t_j},x}}{\sum_{j=1}^k(t_j^2-\bar{t^2})^2},
\end{gather}
where $k$ denotes the number of observations in the regression.
According to Lemma~\ref{lm:fe}, the errors of estimation of $f_{\mathcal{H}_t}$ are also bounded by the predetermined precision $\epsilon$ multiplying the size of residual.
In Eq.~\eqref{eq:OLS}, the estimator is a linear summation.
Thus we choose the largest bias among these fidelity terms with different evolving times $t$ and use it to replace all other bounds,
\begin{gather}
    \abs{\hat{f}^{(2)}_x-f^{(2)}_x}\leq\frac{\sum_{j=1}^k\abs{t_j^2-\bar{t^2}}}{\sum_{j=1}^k(t_j^2-\bar{t^2})^2}\cdot\max_j\{\Delta f_{\mathcal{H}_{t_j},x}\}+o(t^2).
\end{gather}
Note that the higher-order errors come from the systematic errors in the estimation of $f^{(2)}$ as we  only consider the second-order terms of Pauli fidelity.
Furthermore, $o(t^2)$ is due to the observation that only even number order terms contribute to the fidelity.
In the execution of \textbf{FE}, we implement multiple evolutions with $t$ varies from $t_0$ up to $5t_0$ evenly.
Therefore, we have the following constant
\begin{gather}\label{eq:linearreg}
    \frac{\sigma}{t_0^2}=\frac{\sum_{i=1}^5\abs{t_i^2-\bar{t^2}}}{\sum_{i=1}^5(t_i^2-\bar{t^2})^2}\approx\frac{1}{10t_0^2}.
\end{gather}
Combining the above two equations, we can get the claimed bound.
\end{proof}

\subsection{Sparse Pauli Error Rates}\label{sec:errorRate}
In this section, we focus on the post-processing of the fidelity information detected in the last part.
Even though the whole set of fidelity can fully characterize an unknown Pauli channel, it costs an exponential-size calculation to transform the fidelity to Pauli error rates, as stated in Eq.~\eqref{eq:PauliWHT}.
These error rates imply the absolute values of our desired parameters according to Eq.~\eqref{eq:chi_wot}.
Fortunately, the proposal in \cite{harper2020fast} offers an efficient method to complete this transformation given the target channel has sparse nonzero Pauli error rates.
We will illustrate this method with some reorganized algorithms and guarantees.

The protocol first constructs \emph{bins} to collect the linear combinations of certain fidelity terms.
\begin{table}[t!hb]
 \begin{algorithm}[H] 
  \caption{Subsampling and Bin Construction: $\textbf{BC}(\{\mathbf{d}_{c;t}\},l,b,C,P,t_0)$}\label{alg:subsampling}
  \begin{algorithmic}[1]
    \State \texttt{Input:} Number of sampling groups $C$, number of observations P, offsets $\{\mathbf{d}_{c;t}\}_{CP}$, positive integers $l$ and $b$, and the unit time length $t_0$ \State \texttt{Initialization:} a $4^n$ array $\hat{f}^{(2)}$ with all -1 
    \State \texttt{Generation:} Random subsampling matrices $\mathbf{M}_c\in\mathbb{F}_2^{2n\times b}$ for $c \in [C]$. 
    \State \texttt{Modify:} $\mathbf{M}_c'\leftarrow J_n \mathbf{M}_c$ $\forall\,c\in[C]$.
    \ForAll{ $c\in[C]$, $t\in[P]$, and $\ell\in\mathbb{F}_2^{b}$}
        \State $k \leftarrow \mathbf{M}_c' \ell+\mathbf{d}_{c;t}$
        \If{$\hat{f}_k^{(2)}=-1$}
        \State Query Algorithm~\ref{alg:query} with $(G_k,\{P_k\},l,t_0)$
        \State $\hat{f}_k^{(2)}\leftarrow\textbf{FE}(G_k,\{P_k\},l,t_0)$
        \EndIf
    \EndFor
    \State $B \leftarrow 2^{b}$
    \ForAll{ $c\in[C]$ and $t\in[P]$}
    	\State \label{algline:W} 
	$U_{c;t}[j] \leftarrow\frac{1}{B} \sum_{\ell\in\mathbb{F}_2^{b}} (-1)^{\langle j,\ell\rangle}\hat{f}^{(2)}_{\mathbf{M}_c' \ell+\mathbf{d}_{c;t}}$
	\State Return $U_{c;t}[j]$
	\EndFor	
  \end{algorithmic}
\end{algorithm}
\end{table}
As shown in Algorithm~\ref{alg:subsampling}, the procedure chooses fidelity according to the random indices and queries Algorithm~\ref{alg:query} for the fidelity value.
According to the Pauli index defined in Appendix~\ref{sec:PauliConcept}, the stabilizer group $G_k$ is generated by $\langle P_{\phi(\overline{k_1k_2},1)},P_{\phi(\overline{k_1k_2},1)},\cdots,P_{\phi(\overline{k_1k_2},n)}\rangle$, where result of $\phi: \mathbb{F}_2^2\times \{1,\cdots,n\}\rightarrow\mathbb{F}_2^{2n}$ is 0-extended to $2n$-bit string:
\begin{gather}
    \phi(\alpha,i)=\begin{cases}
       \overline{\alpha}\times 2^{2(n-i)} & \alpha\neq00\\
       \overline{01}\times2^{2(n-i)} & \alpha=00.
    \end{cases}
\end{gather}
Intuitively, that means we extend the original Pauli to 1-weight Pauli on every qubit as a generator and use $X$ when encountering $I$ on some qubits.

$\mathbf{M}$ here is some random Boolean matrix, and the procedure reshuffles the fidelity by this matrix.
In order to make bins represented in the same form of a hash function, the procedure further introduces the additional $J_n\coloneqq  I_n\otimes X$ to modify random matrices.
Vector $\mathbf{d}$ works as offsets to create variant phases.
Therefore, every bin is labeled by the matrix index $c$, the offset index $t$, and the intrinsic bit string $j$ with length $b$.
We denote the bins with the same $j$ and $c$ to be in a single \emph{bin set}.
Therefore, there are $B=2^b$ bin sets regarding a random matrix.
With the summation and Walsh-Hadamard Transform, the bin set $\{U_{c;t}[j]\}_t$ actually constructs a partition over the complete Pauli error rates as stated in Lemma~\ref{lm:prop_hashing_obs}.
\begin{lemma}[Rephrased from Lemma 1 in~\cite{harper2020fast}]\label{lm:prop_hashing_obs}
The $B$-point WHT subsampled bin coefficients with index $j\in\mathbb{F}_2^{b}$ can be written as:
\begin{align}\label{eq:U_cp}
	{U}_{c;t}[j] &= \sum_{\alpha:\,\mathbf{M}_c^T\alpha = j} p_\alpha(-1)^{\langle\mathbf{d}_{c;t},\alpha\rangle_p}+{W}_{c;t}[j], \,\forall t\in[P].
\end{align}
 Moreover, the sampling error is as follows
\begin{align*}
    {W}_{c;t}[j] &= \sum_{\alpha:\,\mathbf{M}_c^T\alpha = j} {W}_\alpha(-1)^{\langle\mathbf{d}_{c;t},\alpha\rangle_p},
\end{align*}
where ${W}_\alpha$ is the noise of the Pauli error rate $p_\alpha$. 
\end{lemma}
From Lemma~\ref{lm:prop_hashing_obs}, matrix $\mathbf{M}$ determines the hash function to partition the whole set of Pauli error rates into these bin sets.
We thus denote these matrices by \emph{subsampling matrices}.
Since each bin set is a partitioned piece and there are only sparse nonzero error rates, it will be likely that every bin set contains only a single nonzero error rate, which allows us to learn the value directly.
To achieve this, the procedure must recruit bin sets number $B=O(s)$.

The procedure also constructs methods to determine the precise status of every bin set to verify we can acquire the desired error rate directly.
The following Algorithm~\ref{alg:bin_detect} gives a detailed illustration of the construction of this detector.
We can shed the most light on the case of noiseless processing.
According to Lemma~\ref{lm:prop_hashing_obs}, an offset vector decides the sign for every term in a bin.
Because every bin set employs $P$ different offsets, if all the terms are zero, the absolute values of these bins would be all zero.
In a similar manner, when there exists one nonzero term, all the absolute values would be the same.
In contrast, if these absolute values vary from different bins, the procedure claims there are several nonzero terms in a bin set.
For convenience, we denote the above three types of bins by zero-ton, single-ton, and multi-ton bins, respectively.
Obviously, bins in a bin set remain in the same type, so we may also use these types to label a whole bin set without leading to ambiguity.

Besides, this detector aims to extract both the value and the index of the nonzero error rate when it is the only nonzero term in a bin set.
The idea of index estimation relies on the choice of offsets $\mathbf{d}$.
Ideally, if there is one nonzero term without noise, the procedure would choose the set of offsets $\{d_{c,1},\cdots d_{c,P}\}$ to contain all the linear independent unit vectors in the binary space of length $2n$.
Therefore, the sign of each of the bins suggests the binary value of $\alpha$ on the corresponding bit according to Lemma~\ref{lm:prop_hashing_obs}. 
As for the value of the corresponding error rate, the procedure arbitrarily fetches a bin in the bin set and eliminates the effect of the random sign to get the rate since we have known the index of this rate.

Suppose there exists some noise during the construction of bins.
The idea of distinguishing a bin set among the aforementioned three types keeps the same.
The procedure checks bins with the first $P_1$ random offsets and determines the type of this bin set.
Given a single-ton bin set with the error rate $p_\alpha$, we have the following
\begin{gather}
    \sgn{\mathbf{U}_c[j]}=\mathbf{D}_c^T\alpha+W,
\end{gather}
where $\mathbf{D}_c$ is offset matrix from $P=P_1+P_2$ $\mathbf{d}$s, and $W$ represents the effects from noise.
In order to protect the information about the index $\alpha$ of the only nonzero term against noise, the procedure employs some linear error-correcting codes and implements each offset $\bf d$ among the last $P_2$ as a column of the code generating matrix.
Therefore, the signs of the last $P_2$ bins serve as an element of codewords that remains resilient against the noise from corrupting the signs. 
In the detection, the procedure thus decodes signs back to index $\alpha$ if the code distance is large enough.
Given the index of the nonzero term, the procedure will eliminate signs in the first $P_1$ bins and average over these bins.
This helps to suppress the statistical fluctuations in these bins and get a much more precise estimation of the nonzero error rate.
\begin{table}[thb]
  \begin{algorithm}[H]
  \caption{Bin Detector: $\mathrm{BD}(\mathbf{U}_c,\mathbf{D}_c,T)$}\label{alg:bin_detect} 
  \begin{algorithmic}[1]
  \State {\tt Input}: bin $\mathbf{U}_c$, offsets $\mathbf{D}_c$ and the number $T$ to indicate error size;
  \State\label{algline:zero-ton} {\tt Parameter}: real numbers\label{algline:parameter} $\gamma_1,\gamma_2\in(0,1)$;
  \If{$\frac{1}{P_1}\sum_{t=0}^{P_1-1}U_{c;t}^2\leq T(1+\gamma_1)\nu^2$}
    \State $\widehat{\mathfrak{B}} \leftarrow $\texttt{zero-ton}
    \State Return ($\widehat{\mathfrak{B}}$, \textit{nil, nil})    \Comment \textit{zero-ton verification}
  \EndIf
  \State $\widehat{\alpha}\leftarrow \texttt{Decode}([\sgn {U_{P_1}},\cdots,\sgn{U_{P-1}}]^T)$
  \State $\widehat{p}_{\widehat{\alpha}}\leftarrow\frac{1}{P_1}\sum_{t=0}^{P_1-1}(-1)^{\langle\mathbf{d}_{c;t},\widehat{\alpha}\rangle_p}U_{c;t}$ \Comment \textit{single-ton search}
  \If{$\frac{1}{P_1}\sum_{t=0}^{P_1-1}(U_{c;t}-(-1)^{\langle\mathbf{d}_{c;t},\widehat{\alpha}\rangle_p}\widehat{p}_{\widehat{\alpha}})^2\leq T(1+\gamma_2)\nu^2$}
    \State $\widehat{\mathfrak{B}} \leftarrow$ \texttt{single-ton}
    \State\label{algline:single-ton_verification} Return ($\widehat{\mathfrak{B}}$, $\widehat{\alpha}$, $\widehat{p}_{\widehat{\alpha}}$) \Comment \textit{single-ton verification} 
  \Else 
  \State $\widehat{\mathfrak{B}} \leftarrow$ \texttt{multi-ton}
  \State Return ($\widehat{\mathfrak{B}}$,\textit{nil, nil}) 
  \EndIf
  \end{algorithmic}
\end{algorithm}
\end{table}

However, the procedure cannot extract all the nonzero error rates in this way since it is highly possible that there are some bins including several nonzero terms.
The protocol implements the peeling among different bin sets using the single-ton bin set with the support of the bin detector.
As shown in Algorithm~\ref{alg:peelingIn}, the procedure will fetch the single-ton bin set to peel out the error rate of all other bin sets in different sampling groups which contain the same error rate.
By constantly peeling, the procedure generates more artificial single-ton bins, and the peeling will keep until there are no remaining nonzero error rates.
In the following, we give the full version of the peeling process with detailed parameters and calculations.
The procedure also maintains an array $\mathbf{T}$ to keep track of noise propagated among bins, which helps the bin detector return a reliable detection even with the occurrence of noise. 
\begin{table}[thb]
\begin{algorithm}[H] 
  \caption{Peeling Decoder\label{alg:peeling}}
  \begin{algorithmic}[1]
    \State ${\tt Input}:$ Number of sampling groups $C$, number of observations P, offsets $\{\mathbf{d}_{c;t}\}_{CP}$, positive integers $l$ and $b$, and the unit time length $t_0$ 
    \State ${\tt Initialize}:$ $\mathbf{T}_{c}[j]\gets$ array for all bins with element to be all 1 
    \State ${\tt Initialize}:$ $\mathcal{P}\gets$ empty list of second-order Pauli error rates ($\widehat{\alpha},\widehat{{p}}^{(2)}_{\widehat{\alpha}}$)
    \State Run $\textbf{BC}(\{\mathbf{d}_{c;t}\},l,b,C,P,t_0)$ to construct sampling matrices $\{\mathbf{M}_c\}$ and bins $\{\mathbf{U}_c[j]\}$
    \For{$i=1,\cdots$}
    \State $\text{hassingle}\gets$  false
    	\ForAll{$c\in[C]$ and $\emph{j}\in\mathbb{F}_2^{b}$}
			\State $(\widehat{\mathfrak{B}},\widehat{\alpha},\widehat{{p}}^{(2)}_{\widehat{\alpha}}) \leftarrow \textbf{BD}(\mathbf{U}_c[\emph{j}],\mathbf{D}_c,\mathbf{T}_{c}[\emph{j}])$
			\If{$\widehat{\mathfrak{B}}= \textrm{single-ton}$}
			\State $\text{hassingle}\gets$ true
			\State $\mathcal{P}\gets$ ($\widehat{\alpha},\widehat{{p}}^{(2)}_{\widehat{\alpha}}$)
			\ForAll{$c'\in[C]$ and $c'\neq c$}
			    \State ${\tt Locate~bin~index}$ $\emph{j}_{c'} \leftarrow \mathbf{M}_{c'}^T \widehat{\alpha}$
			    \State \label{algline:T} $\mathbf{T}_{c'}[j_{c'}]\leftarrow\mathbf{T}_{c'}[j_{c'}]+\frac{\mathbf{T}_c[j]}{P_1}+\frac{(P_1-1)B}{P_1N}$
			        \State \label{algline:U} $\mathbf{U}_{c'}[\emph{j}_{c'}] \leftarrow \mathbf{U}_{c'}[\emph{j}_{c'}] - \widehat{p}_{\widehat{\alpha}}(-1)^{\langle\mathbf{D}_{c'},\widehat{\alpha}\rangle_p}$
			 \EndFor
			\ElsIf {${\widehat{\mathfrak{B}}}\neq \textrm{single-ton}$}
				 \State continue to next $\emph{j}$
			\EndIf	
			
		\EndFor
		\If{$\text{hassingle}=false$}
			\State Break
			\EndIf
    \EndFor
    \State Return: $\mathcal{P}$
  \end{algorithmic}
\end{algorithm}
\end{table}

Even though the algorithms are rephrased from \cite{harper2020fast}, the model of noise from desired fidelity terms varies in our implementation of this Pauli error rate estimator.
In this work, this estimator plays as a subroutine utilized to process the estimated fidelity from the fidelity estimation.
Therefore, we cannot impose a similar assumption as stated in \cite{harper2020fast} on the estimated fidelity.
On the other hand, the theorem and corresponding proofs can be modified in our implementation to coordinate with the fidelity estimator's form.

According to the fidelity estimator, there exist some instructions, such as the ratio estimation, that are inherently biased.
Generally, we denote the noise of the estimated fidelity $f^{(2)}$ by $w=\Delta+\omega$, where $\Delta=\mathbb{E}(w)$ is a constant that represents the bias of the estimation.
In this sense, the noise $\omega$ is random with zero mean, and we will further assume that $\omega$ is zero-mean Gaussian noise.
This is due to the central limit theorem as there include many samples of fidelity estimation in the first step.
Even though we have little knowledge about this noise, the whole noise $w$ is bounded according to Lemma~\ref{lm:f2bound}.
For convenience, we denote the bound by $\mathcal{B}>0$.
Thus we model this noise as a constant belonging to an interval $\Delta\in[-\mathcal{B},\mathcal{B}]$ along with zero-mean Gaussian noise $\omega\sim \mathcal{N}(0,\mathcal{B}^2)$.

We must consider even further how this form of noise behaves in every bin.
A bin is constructed from up to $B$ fidelity terms, and all these fidelity terms contain independent random noise.
Thus the Gaussian part of bins' noise that comes from the linear combination of different $\omega$ is $W\sim\mathcal{N}(0,\frac{\mathcal{B}^2}{B})$.
As for the constant parts, we can treat them as biases of the underlying error rates. 
From the Walsh-Hadamard transform, the size of bias $\Delta^{(p)}_\alpha$ for an error rate $p_\alpha$ is no larger than $\frac{B\mathcal{B}}{4^n}$.
In this case, we divide the general noise from the noisy estimation of Pauli fidelity into the two parts mentioned above.

In our execution of this error rate estimator, the main challenge comes from the noisy inputs.
The major part of dealing with the ubiquitous noise is the robust bin detector.
Based on this, the analysis of the failure or exception for running this estimator will focus on the resilience of the bin detector.
Primarily, we rephrase the tail bounds in \cite{Li2015} since we will employ this bound heavily in our proof.
\begin{lemma}[{Tail bound~\cite[Lemma 11]{Li2015}}]\label{lem:tailbound}
Given $\textbf{g},\textbf{k} \in \mathbb{R}^N$ where $\textbf{k}$ is an isotropic Gaussian random variable $\textbf{k}\sim\mathcal{N}(0,\nu^2\mathbbm{1}_N)$, then the following tail bound holds:
\begin{align}
    \Pr\left(\frac{1}{N}\|\textbf{g}+\textbf{k}\|^2\geq\tau_1\right)\leq&\mathrm{e}^{-\frac{N}{4}\left(\sqrt{2\tau_1/\nu^2-1}-\sqrt{1+2\theta_0}\right)^2}\label{tailbound1}\\
    \Pr\left(\frac{1}{N}\|\textbf{g}+\textbf{k}\|^2\leq\tau_2\right)\leq&\mathrm{e}^{-\frac{N}{4}\frac{\left(1+\theta_0-\tau_2/\nu^2\right)^2}{1+2\theta_0}}\label{tailbound2},
\end{align}
for $\tau_1, \tau_2$, and $\theta_0$ satisfying
\begin{gather*}
    \tau_1\geq\nu^2(1+\theta_0),\ \tau_2\leq\nu^2(1+\theta_0),\ \theta_0 = \frac{\|\textbf{g}\|^2}{N\nu^2}.
\end{gather*}
\end{lemma}

As stated in Algorithm~\ref{alg:bin_detect}, the procedure aims to distinguish every input bin set as zero-ton, single-ton, or multi-ton.
Thus the number of type errors is six by enumerating all pairs between the preceding three.
To be more clear, we follow the notation in \cite{harper2020fast}, and let $\mathfrak{B}$ be the real type and $\hat{\mathfrak{B}}$ to be the detected type.
Besides that, we also need to consider the case in which a bin is detected to be single-ton correctly while the detector estimates the error rate with either a value with a huge deviation or a wrong index.
Let us define these failures explicitly.
\begin{definition}\label{df:failure}
The execution of the bin detector (Algorithm~\ref{alg:bin_detect}) is regarded  as  a failure if the output falls into one of the following cases,
\begin{itemize}
    \item Confusion between single-ton and zero-ton bins
    \begin{align*}
        f_1 = \Prob{\widehat{\mathfrak{B}}=\texttt{z}\,|\,\mathfrak{B}=\texttt{s}}\ \ f_2=\Prob{\widehat{\mathfrak{B}}=\texttt{s}\,|\,\mathfrak{B}=\texttt{z}}
    \end{align*}
    \item Confusion between single-ton and multi-ton bins
    \begin{align*}
        f_3=\Prob{\widehat{\mathfrak{B}}=\texttt{m}\,|\,\mathfrak{B}=\texttt{s}}\ \ f_4=\Prob{\widehat{\mathfrak{B}}=\texttt{s}\,|\,\mathfrak{B}=\texttt{m}}
    \end{align*}
    \item Confusion between zero-ton and multi-ton bins
    \begin{align*}
        f_5=\Prob{\widehat{\mathfrak{B}} = \texttt{m}\, |\, \mathfrak{B} = \texttt{z}}\ \ f_6=\Prob{\widehat{\mathfrak{B}} = \texttt{z}\, |\, \mathfrak{B} = \texttt{m}}
    \end{align*}
    \item Errors on the estimated value or index
    \begin{align*}
        f_7=\Prob{\widehat{\mathfrak{B}}=\texttt{s},\widehat{\alpha}\not=\alpha\, |\,\mathfrak{B}=\texttt{s}, \alpha}\ \ f_8=\Prob{\widehat{\mathfrak{B}}=\texttt{s},\alpha,|\hat{p}_\alpha - p_\alpha|> \frac{2\mathcal{B}}{\sqrt{B}} \,|\,\mathfrak{B}=\texttt{s}, \alpha, p_\alpha}
    \end{align*}
\end{itemize}
\end{definition}

We then show a comprehensive lemma to illustrate the robustness in most executions of this detector.
The lemma states an exponentially vanishing failure probability with regard to some mild conditions that we expect to rule out later by the total probability formula.
In the proof, we will separately analyze the probability of all these eight types of failures.
\begin{lemma}\label{lm:failure_bound}
Suppose Assumption A1 \& A2 are satisfied, and suppose $B\geq O\left(\frac{\mathcal{B}^2}{\epsilon_0^2}\right)$.
With the fidelity extracted from Algorithm~\ref{alg:query} and bins constructed from Algorithm~\ref{alg:subsampling}, the detector will deal with noisy bins.
Let $E$ denote the event that an arbitrary bin detection with inputs as those in \Cref{alg:peeling} get failed as defined in Definition~\ref{df:failure}.
Let $D$ be the event that every bin contains at most $P_1$ nonzero terms
Let $V$ be the event that all the prior bin detection executes successfully. 
Let $H$ be the event that the peeling graph is cycle-free.
Let $K$ denote that all the peelings do not rule out the randomness and that the bias parts remain limited dependence.
Then
\begin{gather}
    \Prob{E|D,V,H,K} \leq  \mathrm{e}^{-O(n)}.
\end{gather}
\end{lemma}
\begin{proof}
During the proof we denote $\frac{\mathcal{B}}{\sqrt{B}}$ by $\nu$.
The original noise in each bin can be divided as zero-mean Gaussian noise $\mathcal{N}(0,\nu^2)$ and a bias with variance less than $\frac{B^2\nu^2}{N}$ given the event $K$ as illustrated in Lemma~\ref{lm:bias}.
From Assumption A1 \& A2 and the constraint that $B\geq O\left(\frac{\mathcal{B}^2}{\epsilon_0^2}\right)$, we choose $B$ such that $\epsilon_0\geq \frac{4\mathcal{B}}{\sqrt{B}}$.
In the remainder of this proof, we will keep this constraint on $\epsilon_0$.

We will follow the Definition~\ref{df:failure} to check  all the possibility of variants of detection failures.
For the first one, we consider the case that a bin detection identify a single-ton bin as a zero-ton bin.
The failure rate is defined as follows,
\begin{align}
    f_1 = \Prob{\widehat{\mathfrak{B}}=\texttt{z}\,|\,\mathfrak{B}=\texttt{s}}=\Prob{\frac{1}{P_1}\|\mathbf{U}_{\texttt{s}}\|^2\leq T(1+\gamma_1)\nu^2}.
\end{align}
According to the lemma 7 in \cite{harper2020fast}, the maintained array $\mathbf{T}$ keeps track of the variance information of Gaussian noise in each bin.
Thus, the above failure rate can be bounded by Lemma~\ref{lem:tailbound}, Lemma~\ref{lm:bias}, and Markov inequality.
\begin{align}
    f_1=\Prob{\frac{1}{P_1}\|\mathbf{U}_{\texttt{s}}\|^2\leq T(1+\gamma_1)\nu^2}\leq& \Prob{\frac{1}{P_1}\|\mathbf{U}_{\texttt{s}}\|^2\leq T(1+\gamma_1)\nu^2\,|\,\|\Delta \mathbf{W}\|_\infty\leq\frac{2\mathcal{B}}{B}}+\Prob{\|\Delta \mathbf{W}\|_\infty\geq\frac{2\mathcal{B}}{B}}\notag\\
    \leq& \mathrm{e}^{-\frac{P_1}{4}\frac{((\epsilon_0-2\nu/\sqrt{B})^2/T\nu^2-\gamma_1)^2}{1+2(\epsilon_0-2\nu/\sqrt{B})^2/T\nu^2}}+\frac{P_1B^3}{N}. 
\end{align}
Since the number $N=4^n$, $B$ is a polynomial parameter and $P_1= O(n)$, the failure rate is exponentially low.

Then, we consider the case that we recognize a zero-ton bin as a single-ton bin.
\begin{align}
    f_2=\Prob{\widehat{\mathfrak{B}}=\texttt{s}\,|\,\mathfrak{B}=\texttt{z}}=\Prob{\frac{1}{P_1}\|\Delta \mathbf{W}+\mathbf{W}\|^2\geq T(1+\gamma_1)\nu^2}.
\end{align}
Similarly, we can use the formula of the total probability to bound this failure rate.
\begin{align}
    f_2\leq& \Prob{\frac{1}{P_1}\|\Delta \mathbf{W}+\mathbf{W}\|^2\geq T(1+\gamma_1)\nu^2\,|\,\|\Delta \mathbf{W}\|_\infty\leq\frac{2\mathcal{B}}{B}}+\Prob{ \|\Delta \mathbf{W}\|_\infty\geq\frac{2\mathcal{B}}{B}}\notag\\
    \leq& \mathrm{e}^{-\frac{P_1}{4}\left(\sqrt{1+2\gamma_1}-\sqrt{1+\frac{8}{BT}}\right)^2}+\frac{P_1B^3}{N}.
\end{align}

As for the failure rate $f_3$, it measures the probability of identifying a single-ton bin as a multi-ton bin.
For clearness, we denote the bin contains error rate $p_m$ actually, and the estimated rate is $\hat{p}_{\hat{\alpha}}$.
And the sign vector added to peel the error rate is represented by $\mathbf{s}_{\widehat{\alpha}}$
\begin{align}
    f_3=\Prob{\widehat{\mathfrak{B}}=\texttt{m}\,|\,\mathfrak{B}=\texttt{s}}=\Prob{\frac{1}{P_1}\left\|\mathbf{U}_{\texttt{s}}-\widehat{p}_{\widehat{\alpha}}\mathbf{s}_{\widehat{\alpha}}\right\|^2 \geq  T(1+\gamma_2)\nu^2}.
\end{align}
Since here involves multiple error rates, we only consider the case that the estimated error rate satisfies that $\alpha=\widehat{\alpha}$ and $\abs{p_\alpha-\hat{p}_{\widehat{\alpha}}}\leq \nu$.
By Lemma~\ref{lem:tailbound}, we have
\begin{align}
    f_3\leq f_7+f_8+\mathrm{e}^{{-\frac{P_1}{4}\left(\sqrt{1+2\gamma_2}-\sqrt{1+\frac{2(\sqrt{B}+2)^2}{BT}}\right)^2}}+\frac{P_1B^3}{N}.
\end{align}

As for the next failure, we consider the case that the bin contains several nonzero terms while the detector identifies it as a single-ton bin.
\begin{align}
    f_4=\Prob{\widehat{\mathfrak{B}}=\texttt{s}\,|\,\mathfrak{B}=\texttt{m}}=\Prob{\frac{1}{P_1}\left\|\mathbf{U}_\texttt{m}-\widehat{p}_{\widehat{\alpha}}\mathbf{s}_{\widehat{\alpha}}\right\|^2 \leq T(1+\gamma_2)\nu^2}.
\end{align}
We then decompose the bin vector by the Gaussian noise part $\mathbf{W}$, the bias part $\Delta \mathbf{W}$, the error rates $\mathbf{\beta}$, and the sign matrix $\mathbf{s}$.
The failure rate can be calculated as follows,
\begin{align}
    f_4\leq \Prob{\frac{1}{P_1}\left\|\mathbf{U}_\texttt{m}-\widehat{p}_{\widehat{\alpha}}\mathbf{s}_{\widehat{\alpha}}\right\|^2 \leq T(1+\gamma_2)\nu^2\,|\,\frac{1}{P_1} \|\mathbf{s}\mathbf{\beta}+\Delta \mathbf{W}\|^2\geq 2T\gamma_2\nu^2}+\Prob{\frac{1}{P_1}\|\mathbf{s}\mathbf{\beta}+\Delta \mathbf{W}\|^2\leq 2T\gamma_2\nu^2}.
\end{align}
The first term in the right-hand side can be easily bounded by Lemma~\ref{lem:tailbound} since $T$ tracks the variance size of every element in $\mathbf{W}$,
\begin{align}
    \Prob{\frac{1}{P_1}\left\|\mathbf{U}_\texttt{m}-\widehat{p}_{\widehat{\alpha}}\mathbf{s}_{\widehat{\alpha}}\right\|^2 \leq T(1+\gamma_2)\nu^2\,|\,\frac{1}{P_1} \|\mathbf{s}\mathbf{\beta}+\Delta \mathbf{W}\|^2\geq 2T\gamma_2\nu^2}\leq \mathrm{e}^{-\frac{P_1}{4}\frac{\gamma_2^2}{1+4\gamma_2}}.
\end{align}
As for the second term, we also need to decompose the remaining two terms,
\begin{align}
    \Prob{\frac{1}{P_1}\|\mathbf{s}\mathbf{\beta}+\Delta \mathbf{W}\|^2\leq 2T\gamma_2\nu^2}\leq  \Prob{\frac{1}{P_1}\|\mathbf{s}\mathbf{\beta}+\Delta \mathbf{W}\|^2\leq 2T\gamma_2\nu^2\,|\,\frac{1}{P_1}\|\mathbf{s}\beta\|^2\geq3T\gamma_2\nu^2}+\Prob{\frac{1}{P_1}\|\mathbf{s}\beta\|^2\leq3T\gamma_2\nu^2}.
\end{align}
Note that $\Delta W$ is a random variable with the  mean to be zero  from the summation of multiple Bernoulli variables.
Thus it is an approximate Gaussian variable by the central limit theorem.
The variance falls into the interval $\left[\frac{B}{N}\mathcal{B}^2,\frac{2B}{N}\mathcal{B}^2\right]$.
By Lemma~\ref{lem:tailbound}, the first term can be bounded by
\begin{align}
    \Prob{\frac{1}{P_1}\|\mathbf{s}\mathbf{\beta}+\Delta \mathbf{W}\|^2\leq 2T\gamma_2\nu^2\,|\,\frac{1}{P_1}\|\mathbf{s}\beta\|^2\geq3T\gamma_2\nu^2}\leq \mathrm{e}^{-\frac{P_1}{4}\frac{(1+\gamma_2N/2B^2)^2}{1+3\gamma_2N/B^2}}\ll1.
\end{align}
The second term can follow the analysis in Appendix E of \cite{Li2015}, and we get
\begin{align}
    \Prob{\frac{1}{P_1}\|\mathbf{s}\beta\|^2\leq3T\gamma_2\nu^2}\leq N\mathrm{e}^{-\frac{P_1\epsilon}{4}\left(1-\frac{3T\gamma_2\nu^2}{\epsilon_0^2}\right)}.
\end{align}
Even $N=4^n$, we can choose a proper $P_1$ and guarantee that this is still an exponential vanishing failure rate.
Therefore, we have
\begin{align}
    f_4\leq \mathrm{e}^{-O(n)}.
\end{align}

The following two failure types are easy to handle as they can be reduced to the previous rates.
For the failure that the procedure identifies a zero-ton bin as a multi-ton bin, this bin must pass the zero-ton verification first.
Thus, we have
\begin{align}
    f_5=\Prob{\widehat{\mathfrak{B}} = \texttt{m}\, |\, \mathfrak{B} = \texttt{z}}\leq f_2=\Prob{\widehat{\mathfrak{B}}=\texttt{s}\,|\,\mathfrak{B}=\texttt{z}}
\end{align}

If the detector recognizes a multi-ton as a zero-ton bin, the probability of this failure is very similar to $f_4$, and we can claim that
\begin{align}
    f_6=\Prob{\widehat{\mathfrak{B}} = \texttt{z}\, |\, \mathfrak{B} = \texttt{m}}\leq \mathrm{e}^{-O(n)}.
\end{align}

Then we need to consider the remaining index and value errors.
As the sign for every single-ton bin is generated from the inner product of offsets and the index as follows,
\begin{align}
    \sgn{U_{c;t}[j]}=\sgn{(-1)^{\langle d_{c;t},\alpha\rangle_p}p_\alpha+\Delta W+W}=\langle d_{c;t},\alpha\rangle_p\oplus Z_t,\ \forall\,t\in\{P_1+1,\cdots,P_1+P_2\}, 
\end{align}
where sgn is a function that returns the sign of input and the Bernoulli variable $Z$ denotes whether the bias and noise flips the sign of this bin.
By deliberately choosing these $P_2$ offsets, we can use these bit strings to construct a generating matrix $\mathbf{G}$ of some error correcting code,
\begin{align}
    \sgn{\mathbf{U}_c[j]}=\langle\mathbf{G}_c,\alpha\rangle\oplus\mathbf{Z}.
\end{align}
From the error-correcting codes' perspective, there are $2n$ logical bits, and the code words have length $P_2$.
Thus, the code rate is $R=\frac{2n}{P_2}\leq1$.
Besides, we denote the code protects error with Hamming weight up to $\beta P_2$ with a positive parameter $\beta\in[0,1]$.
To decide whether a code with a given distance can protect the index information, we need to figure out the property of the Bernoulli variables $\mathbf{Z}$.
By definition, the variable $Z=1$ if and only if $\abs{\Delta W+W}\geq p_\alpha$ and they contain opposite signs.
Therefore, we have
\begin{align}
    \Prob{Z=1}\leq \frac{1}{2}\Prob{\abs{\Delta W+W}\geq p_\alpha}\leq\Prob{W\geq p_\alpha-\abs{\Delta W}\,|\,\abs{\Delta W}\leq \frac{2\mathcal{B}}{B}}+\Prob{\abs{\Delta W}\geq \frac{2\mathcal{B}}{B}}\notag.
\end{align}
Due to the fact that $W$ is Gaussian noise with zero mean and variance $T\nu^2$, there exists a tail bound for Gaussian random variables.
\begin{align}
    \Prob{W\geq p_\alpha-\abs{\Delta W}\,|\,\abs{\Delta W}\leq \frac{2\mathcal{B}}{B}}\leq \mathrm{e}^{-\frac{\left(\epsilon_0-\frac{2\nu}{\sqrt{B}}\right)^2}{2T\nu^2}}\leq\mathrm{e}^{-\frac{\left(4-\frac{2}{\sqrt{B}}\right)^2}{2T}}.
\end{align}
The last inequality comes from the assumption about the smallest  Pauli error rate.
Combining this tail bound and the bound of the bias, the success probability of the Bernoulli variable is
\begin{align}
    \Prob{Z=1}\leq\mathrm{e}^{-\frac{\left(4-\frac{2}{\sqrt{B}}\right)^2}{2T}}+\frac{B^3}{N}\ll\frac{1}{2}.
\end{align}
We denote this success rate by $\mathbb{P}$, and the index failure rate is therefore bounded by Hoeffding's inequality,
\begin{align}
    f_7=\Prob{\widehat{\mathfrak{B}}=\texttt{s},\widehat{\alpha}\not=\alpha\, |\,\mathfrak{B}=\texttt{s}, \alpha}\leq \mathrm{e}^{-2P_2(\beta-\mathbb{P})}.
\end{align}

The value estimation comes from the \emph{single-ton search} step in Algorithm~\ref{alg:bin_detect}.
\begin{align}
    f_8=&\Prob{\widehat{\mathfrak{B}}=\texttt{s},\alpha,|\hat{p}_\alpha - p_\alpha|> 2\nu \,|\,\mathfrak{B}=\texttt{s}, \alpha, p_\alpha}=\Prob{\abs{\frac{\mathbf{s}_{\alpha}^T\mathbf{U}_{\texttt{s},\alpha}}{P_1}-p_\alpha}\geq 2\nu}\notag\\
    \leq&\Prob{\abs{\frac{\mathbf{s}_{\alpha}^T}{P_1}(\Delta \mathbf{W}+\mathbf{W})}\geq 2\nu\,|\,\abs{ \frac{\mathbf{s}_{\alpha}^T\Delta\mathbf{W}}{P_1}}\leq \nu}+\Prob{\abs{ \frac{\mathbf{s}_{\alpha}^T\Delta\mathbf{W}}{P_1}}\geq \nu}\notag\\
    \leq& \mathrm{e}^{-\frac{P_1}{2T}}+\frac{4B^2}{P_1N}.
\end{align}

Till now, we have completed the analysis of all types of failures that would occur in the detection of bins.
In the implementation of our algorithm, we choose $P_1 P_2\sim O(n)$ and $N=4^n$ is an exponentially increasing parameter.
Therefore, by the union bound, we have that
\begin{align}
    \Prob{E|D,V,H,K} \leq \sum_{i=1}^8f_i\leq  \mathrm{e}^{-O(n)}.
\end{align}
\end{proof}

\begin{lemma}\label{lm:bias}
Suppose Assumptions A1 \& A2 are satisfied.
Suppose that every bin contains at most $\frac{n}{2}$ nonzero error rates and that every previous bin detection ran successfully before. In an arbitrary peeled bin set $U_{c_0}[j_0]$, the bias in every bin keeps bounded dependence, and the variance after the peeling will be bounded as follows,
\begin{align}
    \text{Var}(\Delta W_{c_0}[j_0])\leq\frac{4B}{N}\mathcal{B}^2.
\end{align}
This above-mentioned observation succeeds with probability at least $1-\mathrm{e}^{-O(n)}$.

\end{lemma}
\begin{proof}
We prove this lemma inductively.
In the beginning, we first focus on the bias from an arbitrary bin before peeling.
According to Lemma~\ref{lm:fe}, the Pauli fidelity terms are carrying bounded noise $w\in[-\mathcal{B},\mathcal{B}]$, and we can regard the bias of fidelity as still falling into this bound naturally.
The procedure constructs every bin by a specific linear combination of multiple fidelity terms, and Lemma~\ref{lm:prop_hashing_obs} tells the form of the resulting bins.
Let us consider the bias $\Delta W$ in a bin $U_{c;t}[j]$, which is in the form of the summation of constants with random signs, 
\begin{align}
    \Delta W_{c;t}[j]=\sum_{\alpha:\mathbf{M}^T_c\alpha=j}(-1)^{\langle d_{c;t},\alpha\rangle_p}\Delta_\alpha^{(p)},
\end{align}
where these $\Delta^{(p)}$ are the effective biases of every error rate by Walsh-Hadamard tranform.
Since all the bias terms are pair-wise independent, we can calculate the variance,
\begin{align}
    \text{Var}(\Delta W_{c;t}[j])\leq \frac{B}{N}\mathcal{B}^2.
\end{align}

Then we need to analyze the effect of peeling processing on the biases of bins.
The peeling decoder algorithm points out that the single-ton bin will be averaged over random offsets to peel the corresponding error rate term in other bins.
The bias in a single-ton bin is propagated to the target bin during every peeling.

For brevity, we denote the $l$ single-ton bins and the target bin by $\{U_{c_1}[j_1],\cdots,U_{c_l}[j_l]\}$ and $U_{c_0}[j_0]$, respectively.
Averaging over those single-ton bins with $P_1$ randomly chosen offsets, and the target bin after the peeling is
\begin{align}\label{eq:peeling}
    U_{c_0;t_0}[j_0]'= U_{c_0;t_0}[j_0]+\frac{1}{P_1}\sum_{k=1}^l\sum_{t=0}^{P_1}(-1)^{\langle d_{c_k;t}+d_{c_0;t_0},\alpha_k\rangle_p}U_{c_k;t}[j_k].
\end{align}
When the previous detection works perfectly, the right-hand side of Eq.~\eqref{eq:peeling} becomes a single-ton bin, and all the other terms are either Gaussian noise or biases.
By the induction assumption, we can find all the involving bins have biases with variance less than  $\frac{4B}{N}\mathcal{B}$.
Note that the bias terms span all possible Pauli indices in a bin, and the average over different random signs is not completely independent.
The bias linked to the nonzero error rate will be propagated faithfully.
Thus we have
\begin{align}
\text{Var}(\Delta W_{c_0;t_0}[j_0]')\leq\frac{B}{N}\mathcal{B}^2+\frac{3lB^2}{N^2}\mathcal{B}^2+\frac{4lB}{P_1N}\mathcal{B}^2\leq \frac{4B}{N}\mathcal{B}^2,
\end{align}
where $2l\leq n\leq P_1\sim O(n)$.
The second term comes from the covariance of the desired nonzero terms' biases, and the third term is the independent part of summation.

In the above analysis, we have presumed two things. 
For the first, we assume there is only one bias term from the peeling that performs dependently with the original bias terms.
This actually contains two aspects that both the peeling bin and the peeled bin have untouched biases with the corresponding Pauli index.
These rely on the constraint that the accumulation of all previous signs that come from peelings would not cause any collapses to the signs of nonzero terms in a bin.
Then we need to bound the probability that the accumulation of random signs in multiple peelings rules out the randomness.
This is possible because the propagating terms carry one more random sign during every peeling process.
Random signs must come from the inner product of the corresponding offset and a nonzero bit string in the kernel space of the subsampling matrix.
To check whether these random signs will annihilate themselves, we can construct a directed group with nodes representing bins and edges representing peeling processes. 
Therefore, only the peeling that lays upstream in the corresponding path will affect the target bin's bias part.
Note there are $BC$ bin sets, and each bin set will be affected by at most $BC$ bin sets.
By the fact that the sum of all previous bitstrings must be zero to eliminate the random sign or equal to the original nonzero terms' indices, we have the failure rate $\Prob{K^c}$ for the above variance bound to be
\begin{align}
    \Prob{K^c}\leq\frac{nB^3C^2}{2N}.
\end{align}
Since $C$ is a constant and $N$ is exponentially large, the failure rate is exponentially vanishing.
\end{proof}

\subsection{Theoretic bounds}\label{sec:bounds}
The abovementioned lemmas help to remove obstacles over the engagement of the fidelity estimator and the error rate transformer.
Consequently, we can develop a comprehensive proposition to demonstrate the validity of combining these two subroutines.
This is based on the following assumptions.

\begin{propbis}{prop:Pauli}
Suppose the Assumption A1 \& A2, \ref{assump:gtm}, and \ref{assump:Pauli} hold.
Execute Algorithm~\ref{alg:peeling} with $t_0$ satisfies $\|\mathcal{H}_{5t_0}-\mathcal{I}\|<\frac{1}{4}$ and $l=\frac{2}{\epsilon^4}\log\left(\frac{4ns|\kappa|}{\delta}\right)$ sequences with a set $\kappa$ of variant sequence lengths, $B=2^b=\max\left\{O(s),O\left(\frac{\epsilon^4}{t_0^4\epsilon_0^2}\right)\right\}$, $C=O(1)$, the unit time length $t_0$, and offsets $\mathbf{D}$ with $P=O(n)$ for each subsampling group.
The transformer will estimate all Pauli error rates $\hat{\mathbf{p}}$
with accurate support information and error bounds to be $\|\hat{\mathbf{p}}-\mathbf{p}\|_\infty\leq O\left(\frac{\epsilon^2}{t_0^2\sqrt{s}}\right)$. 
Therefore, the absolute values of these nonzero decomposition parameters can be estimated by $\|\abs{\hat{s}^\star}-\abs{s^\star}\|_\infty\leq O\left(\frac{\epsilon}{t_0\sqrt[4]{s}}\right)$
The estimation works successfully with probability at least $1-\delta-O\left(\frac{1}{s}\right)$.
\end{propbis}
\begin{proof}
According to Lemma~\ref{lm:fe}, each round of the fidelity estimator results in fidelity terms with errors in the size at most $\mathcal{B}=\frac{\sigma\epsilon^2(f^{(1)}_ir^{(0)}_i+r^{(0,1)}_i)}{(f^{(0)}_i-\epsilon r^{(0)}_i)t_0^2}$, where $\sigma$ is a constant related to the choice of fitting times.
Hence, we reduce the proposition to a claim that with bounded-noise fidelity, the transformer can estimate all the $s$ error rates correctly with their support indices.
This claim can be verified by a series of lemmas.

We first consider the failure probability and error bounds of the transformer.
Since the peeling decoder is the main routine for this transformer, the failure probability can be tracked as follows,
\begin{align}
    \mathbb{P}_F 
	\leq \Prob{\text{Peeling decoder fails}\big|{E}_{\rm bin}^c}
	+\Prob{E_{\rm bin}|D,H}+\Prob{D^c}+\Prob{H^c} \label{error_rate}.
\end{align}
The event ${E}_{\rm bin}$ denotes that there exist some failing bin detections.
On the contrary, ${E}_{\rm bin}^c$ denotes that no bin detection error occurred in the entire execution of \Cref{alg:peeling}.
$D$ and $H$ denote that the maximum number of nonzero terms is at most $\frac{n}{2}$ and that the peeling routines are all cycle-free, respectively.
According to the Proposition 4 in \cite{Li2015}, we know the first term is vanishing with $s$ as $O\left(\frac{1}{s}\right)$.
The third term is exponentially vanishing according to Lemma 8 in \cite{harper2020fast}, and the fourth term is at most $O\left(\frac{\log^{\log\log s}s}{s}\right)$ as stated in Lemma 6 of \cite{Li2015}.
And the remaining is the second term.

To bound the second term, we have the following equation,
\begin{align}
    \Prob{E_{\text{bin}}\big| D,H}=1-\Prob{E_{\text{bin}}^c\big| D,H}=1-\Prob{\bigcap_{i=1}^M E_i^c\big| D,H},
\end{align}
where $M$ denotes the number of bin detection, and $E_i^c$ denotes the event that the $i$-th detection succeeds.
Thus, the question is derived to bound the probability of no anomaly detection.
By the definition of conditional probability, we have
\begin{align}
    \Prob{\bigcap_{i=1}^M E_i^c\big| D,H}=\Prob{E_M^c\bigg|\bigcap_{i=1}^{M-1}E_i^c,D,H}\Prob{\bigcap_{i=1}^{M-1}E_i^c\bigg|D,H}
    =\Prob{E_M^c\bigg|\bigcap_{i=1}^{M-1}E_i^c,D,H}\cdots\Prob{E_1^c|D,H}.
\end{align}
By Lemma~\ref{lm:failure_bound} and the fact that $B\geq O\left(\frac{\mathcal{B}^2}{\epsilon^2_0}\right)$, every term above is exactly what we bounded.
Therefore, all these terms are close to 1 with an exponentially decaying gap
\begin{gather}
    \Prob{E^c_k\big|\bigcap_{i=1}^{k-1}E^c_i,D,H}\geq 1-\mathrm{e}^{-O(P_1)}.
\end{gather}
And we have that
\begin{align}
    \Prob{E_{\text{bin}}\big| D,H}\leq 1-(1-\mathrm{e}^{-O(n)})^M\leq BCs\mathrm{e}^{-O(n)}\leq \mathrm{e}^{-O(n)}.
\end{align}
Thus, we get the stated failure probability.

We then consider the query number of the fidelity estimator, which represents the number of experiments needed to estimate the Pauli information of a Hamiltonian.
As illustrated by Proposition 4 in \cite{Li2015}, the peeling decoder succeeds with the probability of at least $1-O\left(\frac{1}{s}\right)$ adopting $C=O(1)$ subsampling groups and $B=\Theta(s)$ bin sets in every group.
Therefore, to prepare these bins, the number of queried eigenvalues is $BPC=O(Ps)$. 
Note in every bin set, so we choose $P=P_1+P_2=O(n)$.
Therefore, we get the illustrated query number.

As for the error bound, since the event that the procedure estimates the error rates with noise larger than $O\left(\frac{\epsilon^2}{t_0^2\sqrt{s}}\right)$ is classified as value error, the above analysis of the failure probability demonstrates that all error rates fall into the stated interval in successful execution.
Therefore, the absolute values of decomposition parameters can be estimated by the square root of these Pauli error rates, and we get the stated precision.
\end{proof}

\subsection{With Prior Knowledge}\label{sec:prior_info}
From an experimental view of system's Hamiltonian learning, we will always hold some prior knowledge about the system's underlying dynamics.
For example, we can find out some significant interactions of the system that are gained from the experimental details or the designed properties of the device.
Formally speaking, given the prior knowledge, we can expect that the unknown Hamiltonian contains the corresponding terms on the Pauli basis with unknown parameters, which we refer to as the structure information.
Indeed, it is not always true that this prior information covers all the possible structure, so we cannot fully rely on the prior structure to estimate the Hamiltonian.
In this subsection, we would like to discuss how the partial prior knowledge about the structure would boost this Hamiltonian learning method, especially in the Pauli error rate estimation.

During the bin detection, the structure information can serve as a lower bound for the nonzero-term detection.
For example, given that there exists one known term in the detected bin $U$, we shall accept the detection only if it is found to be a single-ton or multi-ton bin.
Otherwise, the algorithm can output an "exception", and we shall adaptively adjust the execution parameter like $\gamma_1$ in Algorithm~\ref{alg:bin_detect}.
And the same happens when we have two or more terms in one bin.
As for the peeling step, since this procedure would invoke the bin detector as a subroutine, the prior information would not directly improve this step.

As discussed above, the prior information about the possible structure can help to improve the accuracy of the overall estimation.
However, since our method is designed without prior information, the benefits we get from the structure knowledge remains in a heuristic way and can be mainly observed in numerical results.

\section{Sign Estimation}\label{sec:signestimate}
We have briefly summarized the sign estimation procedure in Sec~\ref{sec:sign}, including the motivation and basic ideas.
In this section, we will illustrate the protocol to estimate signs of all decomposed parameters.
As this sign estimation plays the third step in the whole Hamiltonian estimation procedure, this estimation employs some information gained from previous Pauli estimation.

In the above section, we regard the Hamiltonian channel as a processing matrix form in Eq.~\eqref{eq:processmatrix}, or, more specifically, as a Pauli channel.
That is due to the fact that our target information is faithfully stored in diagonal terms.
On the contrary, we estimate the sign information mainly from off-diagonal terms in this section, so we need to care for the whole channel.
With the expansion in Eq.~\eqref{eq:Hamiltonian}, we consider the SPAM as a local Pauli eigenstate $\rho_\gamma$ and a Pauli operator $P_\beta$ with the measurement outcome as follows,
\begin{align}\label{eq:originlineq}
    \Tr(\mathcal{H}_t(\rho_\gamma)P_\beta)=\Tr(\rho_\gamma P_\beta)+it\sum_{\alpha\in{\sf P}^n}s_\alpha\Tr(\rho_\gamma P_\alpha P_\beta-P_\alpha\rho_\gamma P_\beta)+o(t^2),
\end{align}
where we want to solve all decomposition parameters $\bf s$.

In order to solve all the parameters to get the sign information, we construct a family of linear equations to extract all the interesting parameters.
Firstly, we want to rule out the effect of the constant term in Eq~\eqref{eq:originlineq}.
We use variant $t$ and fit the measurement results on these $t$. 
Therefore, we get the first-order version of the Hamiltonian evolution.
\begin{align}\label{eq:lineq}
    \Tr(\mathcal{H}^{(1)}(\rho_\gamma)P_\beta)=i\sum_{\alpha\in{\sf P}^n}s_\alpha \Tr(\rho_\gamma [P_\alpha, P_\beta]),
\end{align}
where $[P_\alpha, P_\beta]$ is the commutator notation.
Moreover, since we know the support of nonzero ${\bf s}^\star$ exactly from the last section, it helps to reduce the number of unknown variables and equally the number of equations needed.
Sequentially, we choose multiple state and measurement settings to construct the necessary equations.
The collection of SPAM settings is denoted by $\{(\rho_1,M_1),\cdots,(\rho_m,M_m)\}$.
Collect these first-order measurement results, and this vector is denoted by $\mathcal{M}$.
For the coefficient matrix $\Phi$, we define $\Phi_{j,\alpha}\coloneqq i\Tr(\rho_{\gamma_j}[P_\alpha,M_{\beta_j}])$.
As we pick $m$ pairs states and measurements, and $\Phi$ only covers $s=\abs{{\bf s}^\star}$ nonzero parameters, the matrix is in the shape of $m\times s$.
Therefore, we construct the following equations based on Eq.~\eqref{eq:lineq}
\begin{align}
    \mathcal{M}=\Phi\cdot {\bf s}^\star.
\end{align}

The challenge arises when the observations $\mathcal{M}$ are noisy.
This is possible as we have some noisy SPAMs, and the linear regression brings systematic errors due to the higher-order terms.
Consider the case that the protocol can only estimate $\mathcal{M}$ with the additive noise $\omega$, so the equation becomes
\begin{gather}\label{eq:noisyeq}
    \hat{\mathcal{M}}=\Phi\cdot{\bf s}^\star+\omega.
\end{gather}
Indeed solving the equation directly is not applicable due to the unknown noise. Instead, we adopt the method from compressed sensing \cite{2005decode,candes2006stable,baraniuk2008simple} and consider the problem of finding a feasible vector $x$ satisfying the following constraints
\begin{align*}
\begin{split}
    \min_x \|x\|_1&\\
    ||\Phi x- \hat{\mathcal{M}}||_2&\leq \epsilon
\end{split}.
\end{align*}
This problem is also illustrated in the main text as Eq.~\eqref{eq:optimize}.
Notably, this problem has a nonempty feasible set for every $\epsilon\geq\|\omega\|_2$ since the ideal value of ${\bf s}^\star$ is a feasible solution.
We always choose $\epsilon=\sqrt{m}\left(\frac{\sigma\tau}{t_1}+o(t_1)\right)$, which is the upper bound of $\|\omega\|_2$.
We denote this problem as \emph{sign optimization}.

We will show that the solution $x$ is a well-qualified estimation of the ideal parameters given the measurements and states are randomly chosen from the Pauli group and local eigenstates thereof.
In the following, we first give a definition of the \emph{restricted isometry property} introduced by~\cite{2005decode,baraniuk2008simple}.
Then we show how to prove this property for the randomly sampled process matrices in our specific case and how to bound the noise effects.

\begin{definition}\label{de:rip}
We say that a matrix $\Phi$ satisfies the Restricted Isometry Property (RIP) of order $s$ if there exists a
$\delta_s\in(0,1)$ such that
\begin{gather}
    (1-\delta_s)\|x_T\|_2^2\leq\|\Phi x\|_2^2\leq(1+\delta_s)\|x_T\|_2^2
\end{gather}
holds for all $x$ on all possible support $T$ such that $\#(T)\leq s$.
\end{definition}

\begin{lemma}\label{prop:signRIP}
Suppose we know the exact support of the vector ${\bf s}^\star$ on the Pauli group.
Thus, with the number of equations $m\geq \Omega(s)$, the solution $x^\star$ from Algorithm~\ref{alg:sign} is close to the ideal parameter ${\bf s}^\star$ with probability $1-\mathrm{e}^{-O(m-s)}$,
\begin{gather}\label{eq:spambound}
    \|x^\star-{\bf s}^\star\|_2\leq C\cdot \epsilon, 
\end{gather}
where $C$ is a constant related to the choice of the sampling set of SPAM.
\end{lemma}
\begin{proof}
We first figure out the size of noise $\omega$ that comes from SPAM errors and linear regressions in~\eqref{eq:noisyeq}.
According to Assumption \ref{assump:SPAM}, every measurement carries noise at most $\tau$.
Therefore, we have
\begin{gather}
     \Tr(\mathcal{H}_t(\tilde{\rho}_\gamma)\tilde{P}_\beta)=\Tr(\rho_\gamma P_\beta)+it\sum_{\alpha\in{\sf P}^n}s_\alpha\Tr(\rho_\gamma P_\alpha P_\beta-P_\alpha\rho_\gamma P_\beta)+o(t^2)+\tau,
\end{gather}
which satisfies for all $t$.
Considering the ordinary least square, we have 
\begin{gather}
    \abs{\Tr(\mathcal{H}^{(1)}_t(\tilde{\rho}_\gamma)\tilde{P}_\beta)-\Tr(\mathcal{H}^{(1)}_t(\rho_\gamma)P_\beta)}\leq\frac{\sum_{i=1}^K\abs{i\cdot t_1-\Bar{t}}}{\sum_{i=1}^K(i\cdot t_1-\Bar{t})^2}\cdot\tau+o(t),
\end{gather}
where $K$ denotes the number of evolving times used to extract the first-order term and $\Bar{t}$ is the averaged time length.
As we choose the $t$ from $t_1$ up to $Kt_1$ arithmetically, the foregoing bound is derived to $\frac{\sigma\tau}{t_1}+o(t_1)$ for some regression constant $\sigma$, which is very similar to Eq.~\eqref{eq:linearreg}. 
Therefore, the norm of the total noise term $\omega$ can be bounded by
\begin{gather}
    \|\omega\|_2\leq\sqrt{m}\left(\frac{\sigma\tau}{t_1}+o(t_1)\right).
\end{gather}
According to Algorithm~\ref{alg:sign}, the chosen $\epsilon$ makes sure that $\epsilon\geq\|\omega\|_2$, which guarantees the existence of the solution of the optimization problem.

We then prove the approximate RIP for the constructed $\Phi$ matrix.
The proof is inspired by~\cite{Shabani2011,baraniuk2008simple}.
In the noisy version, we randomly choose the local eigenstates of Pauli operators and Pauli measurements, so $\Phi$ is a $m\times s$ random matrix.
By the fact that all these states and measurements are in a finite dimension system, the matrix $\Phi$ is always bounded.
Consider an arbitrary $s$-sparse vector $x$, and we have the following upper and lower bounds,
\begin{gather}
    \frac{w_l}{m}\|x\|^2\leq x^T\phi_i\phi_i^Tx\leq\frac{w_u}{m}\|x\|^2,\ \ \forall\,x\in\mathbb{R}^s\\
    l\|x\|^2\leq \mathbb{E}_{\Phi}\|\Phi x\|^2\leq u\|x\|^2,\ \ \forall\,x\in\mathbb{R}^s,
\end{gather}
where we use $w_l,w_u,l,u$ to show the upper and lower bounds, and we denote the $i$th row of $\Phi$ by $\phi_i$.
Implement Hoeffding's inequalities on the random matrix $\Phi$, we have
\begin{align}
    \Prob{\|\Phi x\|^2\geq (1+\delta_s)\|x\|^2}\leq&\mathrm{e}^{-2m(1+\delta_s-u)^2/(w_u-w_l)^2}\\
    \Prob{\|\Phi x\|^2\leq (1-\delta_s)\|x\|^2}\leq&\mathrm{e}^{-2m(1-\delta_s-l)^2/(w_u-w_l)^2},
\end{align}
for some positive $\delta_s$.
Combining these two, we get,
\begin{align}\label{eq:rip}
    \Prob{\abs{\|\Phi x\|^2-\|x\|^2}\geq \delta_s\|x\|^2}\leq 2\mathrm{e}^{-2m(\delta_s+\mu_0)/(w_u-w_l)^2},
\end{align}
where $\mu_0$ denotes $\min(1-u,l-1)$.

Now we need to generalize this inequality to every vector in a specific $s$-sparse support.
In our case, that is the parameter support we learned from stage 1.
According to Lemma 5.1 in \cite{baraniuk2008simple},we can use the counting and union bound to prove with probability $1-\mathrm{e}^{-2m(\delta_s+\mu_0)/(w_u-w_l)^2+s\log(12/\delta_s)}$, every $x$ on the support gained from stage 1 satisfies
\begin{gather}
    (1-\delta_s)\|x\|_2^2\leq\|\Phi x\|_2^2\leq(1+\delta_s)\|x\|_2^2.
\end{gather}
Note that we can decide the sparse support, so it is not necessary to further count all possible supports.
We need to choose $m>s$ to get a vanishing failure probability, and this is consistent with the intuition that the number of linear equations must be larger than that of unknown variables.

Suppose we get a solution $x^\star$ in the feasible set of the problem in Eq.~\eqref{eq:optimize}.
Therefore, according to Theorem 1.3 in \cite{CANDES2008589} and the definition of RIP matrix in \cite{2005decode}, the result of Eq.~\eqref{eq:optimize} satisfies the following distance bound toward the real parameter ${\bf s}^\star$
\begin{align}
    \|x^\star-{\bf s}^\star\|_2\leq \frac{2\sqrt{1+\delta_s}}{1-(1+\sqrt{2})\delta_s}\cdot\epsilon. \notag
\end{align}
Note the constant $\delta_s$ is artificially picked.
Therefore, we can always choose a proper $\delta_s<\sqrt{2}-1$.
\end{proof}
To actually rule out the effects of noise on the sign estimation, we need to make sure the additive noise $\omega$ is bounded by some bars that depend on the lower bound of the nonzero decomposition parameters.
This can be achieved by calculating the propagation of the noise and the chosen number of equations $m$.
\begin{propbis}{prop:phase2}
Suppose the absolute value estimator perfectly returns the support information of decompose parameters of the Hamiltonian, and Assumption A1 \& A2 and \ref{assump:SPAM} hold.
Run Algorithm~\ref{alg:sign} with the support information of decomposition parameters and the unit time length $t_0$.
By setting $m$ to be $O(s)\leq m\leq O\left(\frac{t_1^2\epsilon_0}{\tau^2}\right) $ The solution of Algorithm~\ref{alg:sign} contains perfect sign information of all nonzero decomposition parameters with probability at least $1-\mathrm{e}^{-O(m-s)}$.
\end{propbis}
\begin{proof}

As proved in Lemma~\ref{prop:signRIP}, the process of solving the linear equations will not enlarge the errors.
The solved decomposition parameters $x^\star$ satisfy that following bound with a probability of at least $1-\mathrm{e}^{-O(m-s)}$
\begin{gather*}
    \|x^\star-{\bf s}^\star\|_2\leq C\cdot\sqrt{m}\left(\frac{\sigma\tau}{t_1}+o(t_1)\right),
\end{gather*}
where $C$ comes from the constant in Lemma~\ref{prop:signRIP}.
According to Assumption A1 \& A2 and \ref{assump:SPAM}, $m$ determines the upper bound of norm distance, which is supposed to be smaller than $O(\sqrt{\epsilon_0})$.
Therefore, this error indicates that there is no sign flip between the estimated parameters $x^\star$ and the actual decomposition vector ${\bf s}^\star$.
\end{proof}
\begin{remark}\label{rm:pha2t}
\rm Similar to Remark~\ref{rm: pha1t}, we need to choose some proper $t_1$ to attune the two types of noise.
Thus, in our model, we would choose $t_1\sim\sqrt{\sigma\tau}$.
\end{remark}

\section{Main Theorem}\label{sec:mainproof}

Till now, we have introduced our protocol that combines the fidelity estimator, Pauli error rates estimator, and the sign estimator.
The validity of this protocol, therefore, relies on those guarantees of subroutines.
In this section, we introduce the main theorem that demonstrates the correctness and efficiency of our protocol working on the general $n$-qubit quantum systems.

\begin{thmbis}{thm:main}
Suppose Assumption A1 \& A2, \ref{assump:SPAM}, \ref{assump:gtm}, and \ref{assump:Pauli} hold. Run Algorithm~\ref{alg:main} with $l=\frac{2}{\epsilon^4}\log\left(\frac{4ns|\kappa|}{\delta}\right)$ sequences for each length where $\kappa$ is the set of variant sequence lengths, $B=2^b=\max\left\{O(s),O\left(\frac{\epsilon^4}{t_0^4\epsilon_0^2}\right)\right\}$, $C=O(1)$, $m=O(s)$, unit times $t_0$ satisfies $\|\mathcal{H}_{5t_0}-\mathcal{I}\|<\frac{1}{4}$ and $t_1$, and the offsets $D$ with size $P=O(n)$ for each subsampling group.
The Hamiltonian estimator will return all nonzero decomposition parameters with the perfect support estimation and $\|\hat{\bf s}^\star-{\bf s}^\star\|_\infty\leq O\left(\frac{\epsilon}{t_0\sqrt[4]{s}}\right)$, which succeeds with probability at least $1-\delta-O\left(\frac{1}{s}\right)$.
The circuit measurement complexity of this execution is $\tilde{O}\left(\frac{sn}{\epsilon^4}\right)$ where $\tilde{O}$ notation ignores the logarithmic terms.
The post-processing complexity is $\text{poly}(n,s)$, where poly denotes that the scaling is polynomial with the elements.
\end{thmbis}
\begin{proof}
According to Proposition~\ref{prop:Pauli}, Algorithm~ \ref{alg:peeling} and \ref{alg:subsampling} with these parameters can estimate the absolute values of decomposition parameters up to the accuracy of $\|\abs{\hat{s}^\star}-\abs{s^\star}\|_\infty\leq O\left(\frac{\epsilon}{t_0\sqrt[4]{s}}\right)$. 
Moreover, the sign estimator will extract perfect sign information from Proposition~\ref{prop:phase2}.
Therefore, the illustrated accuracy can be achieved.
The failure probability comes from the union bound of these two stages, which are $\delta+\left(\frac{1}{s}\right)$ and $\mathrm{e}^{-O(m-s)}$, respectively.
As $m=O(s)$, we can claim the stated failure probability.

As for the measurement complexity, it is gained from simple counting.
In the first stage, we run the cascading circuit with altogether $l\cdot\abs{\kappa}$ sequences.
The set $\kappa$ includes exponentially increasing sequence lengths from 1 to $O\left(\frac{1}{\Delta}\right)$ where $\Delta$ refers to the smallest fidelity residual of the detected channel.
In the second stage, the protocol needs to prepare $m$ equations, which are equal to $m$ measurements.
Therefore, the complexity $M_Q$ satisfies
\begin{gather*}
    M_Q=l\cdot \abs{\kappa}+m=\tilde{O}\left(\frac{sn}{\epsilon^4}\right).
\end{gather*}

The classical complexity is also combined from the foregoing two stages.
In the first stage, the post-processing is mainly contributed by the peeling algorithm, of which the complexity is stated as $O(sn^2)$ in Theorem 1 of \cite{harper2020fast}.
Also, to extract the desired fidelity terms, the procedure needs $O(sn)$ fittings.
In the second stage, the optimization problem illustrated in Eq.~\eqref{eq:optimize} is of size $m\times s$.
Since we have confirmed the existence of the solution, the complexity is $\text{poly}(s)$.
Therefore, we have
\begin{gather}
    M_C=O(sn^2)+O(sn)+\text{poly}(s)=\text{poly}(n,s).
\end{gather}
\end{proof}

\section{Implementation and numerical results}\label{sec:addnum}
This section serves as the further analysis and numerical detection of some concerns that are raised in the main text.

\subsection{Threshold Behaviors}\label{sec:threshold}
In this section, we perform the exhaustive numerical simulation to show how we choose the bin parameter $b$ in each case of the simulations introduced in \Cref{sec:numerical}.
Even though the procedure runs with an inputted parameter, $b$, this parameter is supposed to rely on the structure of the target Hamiltonian to guarantee the successful execution.
Hence, we have to employ a trial process to determine the proper $b$ for each Hamiltonian.
More clearly, we increase the parameter $b$ and run the procedure on the target Hamiltonian to find the proper choice of $b_0$ such that the estimation with $b\geq b_0$ is steadily close to the ideal values.
In Figure~\ref{fig:ising}, we exhibit the simulation results of running the procedure exhaustively to detect the behaviors of different $b$ on different systems.

\begin{figure}[tbh]
    \centering
    \includegraphics[width=\columnwidth]{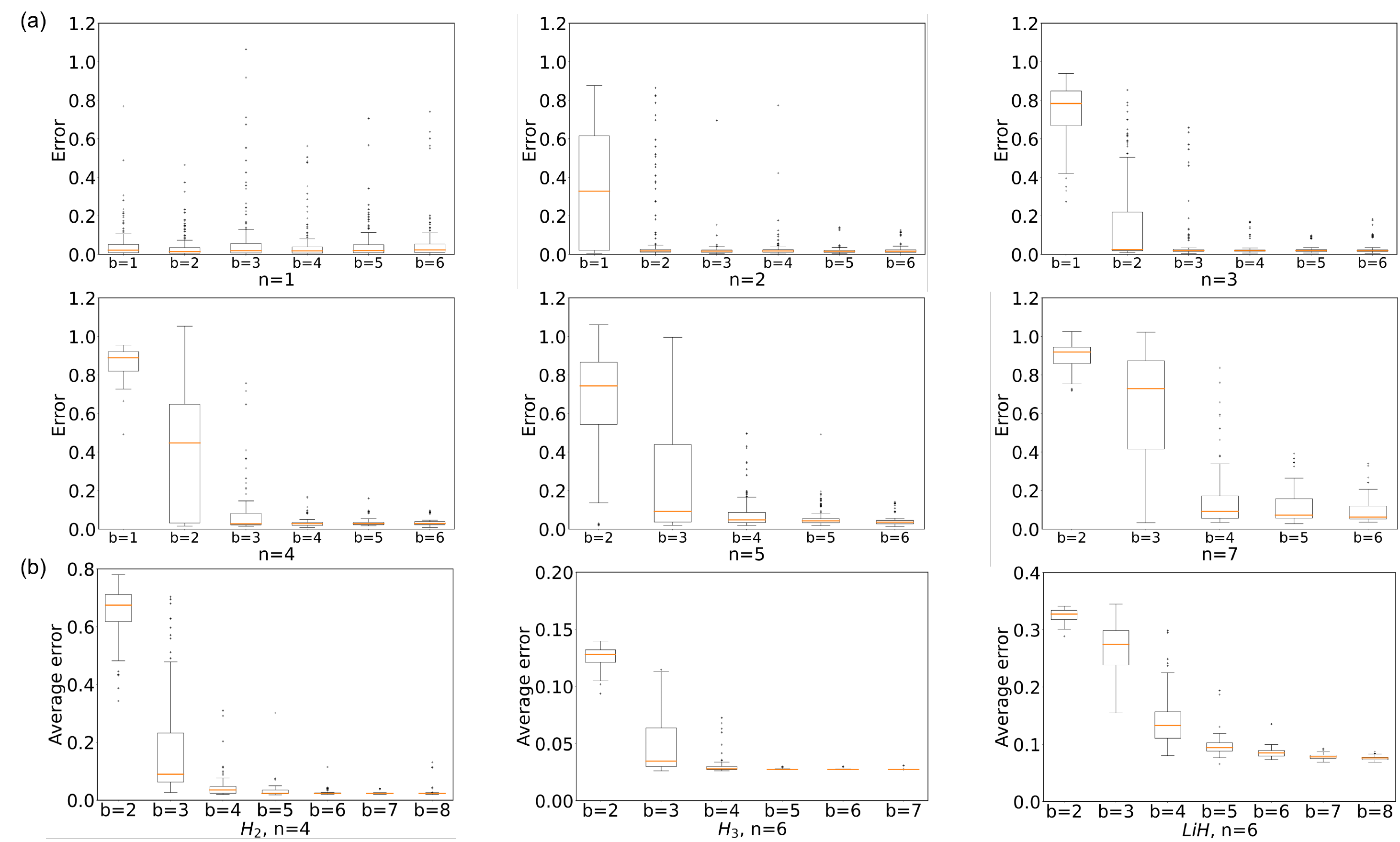}
    \caption{This figure includes the thorough results for the trial processes for successfully executing the learning procedure.
    In figure~(a), we show the execution of random TFIM Hamiltonian with multiple system sizes. 
    The metric for errors is chosen as the relative 1-norm distance.
    Note in the case of $n=1$, all the choices of $b$ perform a very similar estimation distribution, which is due to the only one nonzero term in this case. In the case of $n=2$, the reconstruction becomes constantly close to the ideal values when $b\geq4$, which implies the proper choice, $b=4$. For $n=3$, the steady reconstruction shows up when the procedure chooses $b\geq4$.
    As for cases of $n=4$, $n=5$, and $n=7$, the steady estimations are recognized when $b\geq 5$. These trials solve the concerns of choosing proper $b$ raised in the TFIM simulation in Figure~\ref{fig:numerical}.
    In figure~(b), we focus on the $b$ choosing problem of the molecular Hamiltonian learning. 
    In the case of learning $H_2$ Hamiltonian, the reconstruction errors decay fast by increasing $b$ up to 5.
    After $b\geq 6$, the estimation performs well, which indicates that $b=6$ is a proper choice.
    In the case of $H_3$ and $LiH$ molecules, the proper choices are $b=6$ and $b=7$, respectively. }
    \label{fig:ising}
\end{figure}

In Figure~\ref{fig:ising}, we display the trial results for the random TFIM Hamiltonian with different system sizes.
In the case of $n=1$, all these $b$ perform well due to the fact that there is only one nonzero parameter in the Hamiltonian.
In the case of the $n=2$ system, these results are more typical, which is very similar to Figure~\ref{fig:threshold}.
According to the box plot, the procedure performs much better when $b\geq2$, and the performance keeps steady after $b$ grows up to 4.
Therefore, the procedure will choose $b=4$ as a proper parameter to execute on the 2-qubit system.
In all other cases, this process works very similarly.
Based on the understanding of the requirement of the parameter $b$ and the observations on Figure~\ref{fig:threshold} and Figure~\ref{fig:ising}, parameters for random TFIM Hamiltonian learning are just the same as those used in \Cref{sec:numerical}.


In addition to the random TFIM systems, we also check the threshold behaviors for molecular systems.
The numerical results of the trial processes are shown in Figure~\ref{fig:ising}(b).
A key property of molecular Hamiltonian is that they contain much more nonzero terms than the Ising models.
Hence, the threshold behaviors are generally asking for larger $b$.
In the first plot, the reconstruction errors of $H_2$ molecules decay rapidly from $b=2$ to $b=5$.
After the procedure increases $b$ up to 6, the reconstruction becomes steadily close to the ideal decomposition parameters.
Based on these results, the procedure regards $b=6$ as a proper choice of the bin parameter, and we exhibit the corresponding distribution of $b=6$ estimations in Figure~\ref{fig:numerical}.
Similarly, we run the trial simulation for $H_3$ and $LiH$ molecules, respectively.
By witnessing the reconstruction with steady errors, the procedure can then determine the proper parameter in each case.

\subsection{Higher-Order Fitting}\label{sec:fitting}
During the simulation for the numerical results, we found there exist some lower bounds of the errors after we detect the distribution of our reconstructions for systems with comparable large sizes.
This means there are some systematic errors in some processes of the procedure.
In this section, we will show these errors are mainly contributed by the higher-order terms during the fitting of the second-order fidelity.

Firstly, we need to consider the definite expression of the fidelity terms of a Hamiltonian channel.
According to Eq.~\eqref{eq:Hamiltonian}, we need to expand this series to even higher-order terms.
\begin{align}
    \mathcal{H}_t(\rho)
    =&\rho+it\sum_{\alpha\in{\sf P}^n}s_\alpha(\rho P_\alpha-P_\alpha\rho)\notag
    +t^2\sum_{\alpha,\beta\in{\sf P}^n}s_\alpha s_\beta\left[ P_\alpha\rho P_\beta-\frac{1}{2}(P_\alpha P_\beta\rho+\rho P_\alpha P_\beta)\right]\\
    &+it^3\sum_{\alpha,\beta,\gamma\in{\sf P}^n}s_\alpha s_\beta s_\gamma\left[\frac{1}{6}(\rho P_\alpha P_\beta P_\gamma-P_\alpha P_\beta P_\gamma\rho)+\frac{1}{2}(P_\alpha\rho P_\beta P_\gamma-P_\alpha P_\beta \rho P_\gamma)\right]\notag\\
    &+t^4\sum_{\alpha,\beta,\gamma,\delta\in{\sf P}^n}s_\alpha s_\beta s_\gamma s_\delta\left[\frac{1}{24}(\rho P_\alpha P_\beta P_\gamma P_\delta+P_\alpha P_\beta P_\gamma P_\delta\rho)-\frac{1}{6}(P_\alpha \rho P_\beta P_\gamma P_\delta+P_\alpha P_\beta P_\gamma \rho P_\delta)+\frac{1}{4}P_\alpha P_\beta\rho P_\gamma P_\delta\right]\notag\\
    &+o(t^5).
\end{align}
Note that these Pauli indices are all equivalent, and we can use this property to determine the effects on diagonal terms (Pauli terms).
For the first-order terms, they simply cause no diagonal effects. For the second-order terms, the diagonal effects from them are depicted by Eq.~\eqref{eq:chi_wot}, which is the major source of learning the Hamiltonian in our protocol.
As for the third-order terms, there are also zero effects for them on the diagonal terms.
The first terms obviously contribute no effects on the Pauli terms since they annihilate themselves on the Pauli terms.
The second terms have no effects either due to the rotational symmetry of the indices.
The fourth-order terms contribute significantly to the diagonal terms.
Apparently, so far, the odd-order terms contain no effects of the diagonal terms, and this can be easily extended to all the odd-order terms.
This guides us to further improve our fitting process by considering other even-order terms.

We also show the numerical results to verify this idea explicitly.
\begin{figure}
    \centering
    \includegraphics[width=0.8\columnwidth]{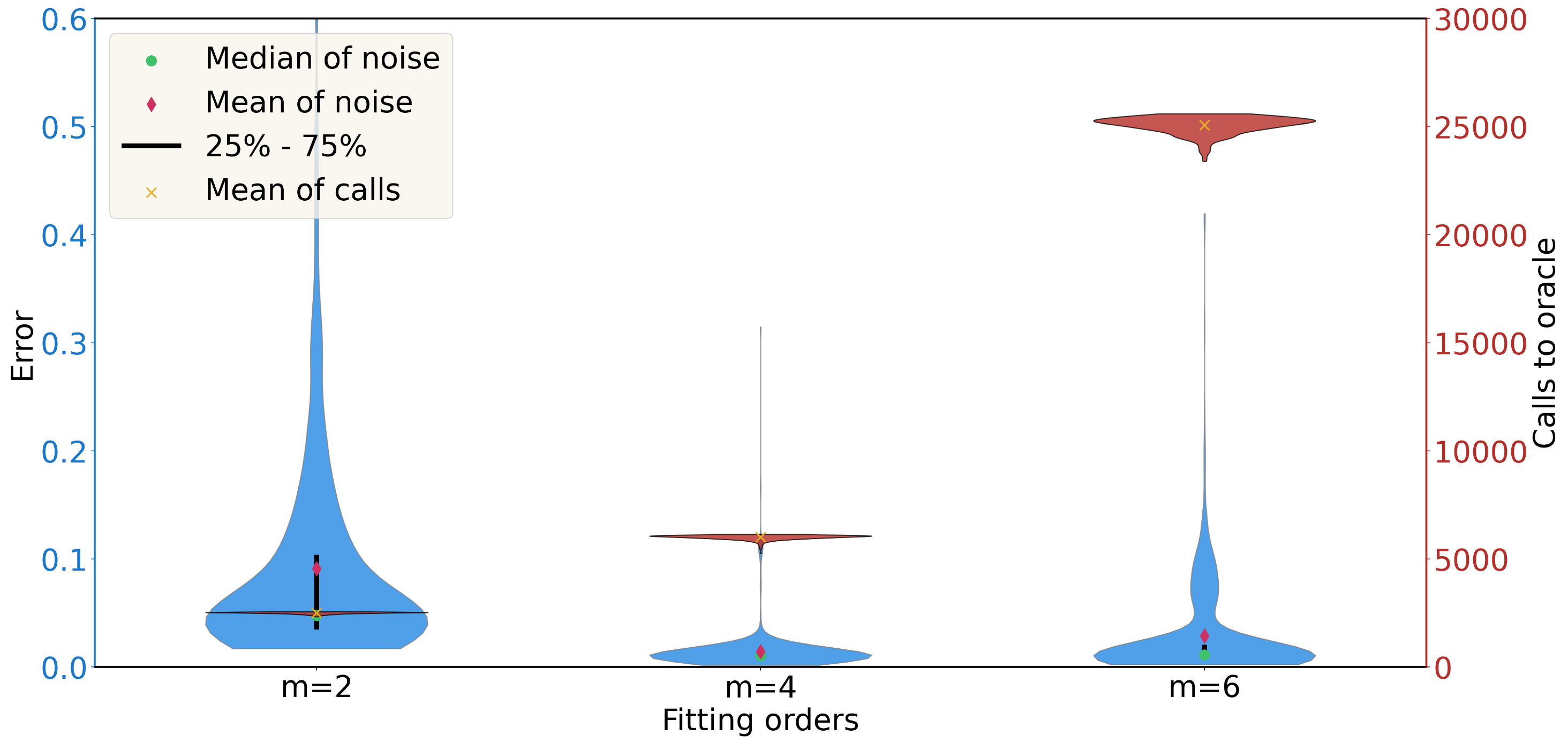}
    \caption{This figure depicts the effects of different fitting settings on the estimation of Hamiltonians. 
    To rule out the statistical fluctuations in the detection, we run different fitting processes on 50 random 4-qubit TFIM Hamiltonians for 10 rounds each. 
    For $m=2$, it is the ordinary fitting used for the vast estimation exhibited in the main text.
    For $m=4$, we consider the effects of fourth-order terms besides the second-order terms in the fitting process. This modification asks for a little more time steps for regression while bringing much better estimation.
    As for $m=6$, we consider the further sixth-order terms, and this modification requires even more time steps to implement six-order regression, which increases the numbers of oracle used. }
    \label{fig:varm}
\end{figure}
In Figure~\ref{fig:varm}, we exhibit the reconstruction cases with different fitting processes and fixed fidelity noise from finite shot of circuit.
In this figure, $m$ denotes the highest terms that we consider in the fitting.
As can be observed, the estimation is better when considering the fourth-order terms than the ordinary fitting with a mild sacrifice of the measurement complexity due to more time step needed for regression.
This proves the previous idea about systematic errors.
On the other hand, when the procedure keeps fetching higher-order terms, both the errors and complexity grow due to the effects of the circuit noise.
Even though the systematic errors are suppressed by considering higher-order terms, the circuit noise would be enlarged by the higher-order regression.
Going from $m=4$ to $m=6$, the circuit noise is now the dominant contributor of the overall errors.
Therefore, we propose an effective solution to significantly alleviate the systematic errors in the fitting process, while we still need to suppress the circuit noise by further increase the shot numbers in order to balance the systematic errors and circuit noise of fidelity.

\end{document}